%% file: nips_main.tex
\documentclass{article}

\usepackage[preprint]{neurips_2023}

\pdfoutput=1
\usepackage[utf8]{inputenc} % allow utf-8 input
\usepackage[T1]{fontenc}    % use 8-bit T1 fonts
\usepackage{hyperref}       % hyperlinks
\usepackage{url}            % simple URL typesetting
\usepackage{booktabs}       % professional-quality tables
\usepackage{amsfonts}       % blackboard math symbols
\usepackage{nicefrac}       % compact symbols for 1/2, etc.
\usepackage{microtype}      % microtypography
\usepackage{xcolor}         % colors

\usepackage{natbib}
\usepackage{microtype}
\usepackage{graphicx}
\usepackage{subfigure}
\usepackage{booktabs} % for professional tables
\usepackage[font=small,labelfont=bf]{caption}

\usepackage{caption}
\usepackage{graphicx}
\usepackage{float} 

\usepackage{parskip}

\usepackage{amsthm}
\usepackage{amsmath}
\usepackage{amsfonts}
\usepackage{nicefrac}
\usepackage{microtype}
\usepackage{comment}
\usepackage{subfigure}
\usepackage{color}
\usepackage{cases}

\usepackage{tikz}
\usetikzlibrary{matrix}
\usetikzlibrary{calc}
\usetikzlibrary{decorations.pathreplacing}

\usepackage{amscd,latexsym,amsthm,amsfonts,amssymb,amsmath,amsxtra}

\usepackage{tikz}
\usepackage{pgf, tikz}
\usetikzlibrary{arrows, automata}
\RequirePackage{amssymb}
\definecolor{mygray}{gray}{0.85}

\usepackage{thmtools}
\usepackage{thm-restate}
\theoremstyle{plain}

\theoremstyle{definition}

\newcommand{\norm}[2][]{\ensuremath{\left\Vert #2 \right\Vert}}

    \newcommand{\BC}{{\mathbb {C}}}

     \newcommand{\BN}{{\mathbb {N}}}
     
     \newcommand{\BR}{{\mathbb {R}}}

     \newcommand{\BZ}{{\mathbb {Z}}}

     \newcommand{\CN}{{\mathcal {N}}}
    \newcommand{\CO}{{\mathcal {O}}} 
     
     \newcommand{\CT}{{\mathcal {T}}}

    \newcommand{\Dim}{{\mathrm{dim}}}

    \theoremstyle{plain}
    \newtheorem{thm}{Theorem}[section] \newtheorem{cor}[thm]{Corollary}
    \newtheorem{lem}[thm]{Lemma}  \newtheorem{prop}[thm]{Proposition}
     \newtheorem{defn}[thm]{Definition}
     \newtheorem {rem}[thm]{Remark}

\newcommand \footnoteONLYtext[1]
{
	\let \mybackup \thefootnote
	\let \thefootnote \relax
	\footnotetext{#1}
	\let \thefootnote \mybackup
	\let \mybackup \imareallyundefinedcommand
}

\title{On the Last-iterate Convergence in Time-varying Zero-sum Games: Extra Gradient Succeeds where Optimism Fails}

\author{%
  Yi Feng\\
  SUFE\thanks{Shanghai University of Finance and Economics}\\
  \texttt{2021310186@live.sufe.edu.cn} \\
  \And
  Hu Fu \\
  SUFE \\
  fuhu@mail.shufe.edu.cn \\
  \AND
  Qun Hu\\
  SUFE \\
  huqun29@163.com \\
  \And
  Ping Li \\
  SUFE \\
  lping0423@163.com \\
  \And
  Ioannis Panageas \\
  University of California, Irvine \\
  ipanagea@ics.uci.edu \\
  \And
  Bo Peng\\
  SUFE\\
  ahqspb@163.sufe.edu.cn\\
  \And
  Xiao Wang \\
  SUFE \\
  wangxiao@mail.shufe.edu.cn\\
}

\begin{document}
\footnoteONLYtext{Authors are listed according to the alphabetical order.}

\maketitle

\begin{abstract}
Last-iterate convergence has received extensive study in two player zero-sum games starting from bilinear, convex-concave up to settings that satisfy the MVI condition. Typical methods that exhibit last-iterate convergence for the aforementioned games include extra-gradient (EG) and optimistic gradient descent ascent (OGDA). However, all the established last-iterate convergence results hold for the restrictive setting where the underlying repeated game does not change over time.
Recently, a line of research has focused on regret analysis of OGDA  in time-varying games, i.e., games where payoffs evolve with time; the last-iterate behavior of OGDA and EG in time-varying environments remains unclear though. In this paper, we study the last-iterate behavior of various algorithms in two types of unconstrained, time-varying, bilinear \textcolor{black}{zero-sum} games: periodic and \textcolor{black}{convergent perturbed games}. These models expand upon the usual repeated game formulation and incorporate external environmental factors, such as the seasonal effects on species competition and vanishing external noise. In periodic games, we prove that EG will converge while OGDA and momentum method will diverge. This is quite surprising, as to the best of our knowledge, it is the first result that indicates EG and OGDA have qualitatively different last-iterate behaviors and do not exhibit similar behavior. In \textcolor{black}{convergent perturbed games}, we prove all these algorithms converge as long as the game itself stabilizes with a faster rate than $1/t$.

%\textcolor{red}{Red parts are different from submitted abstract.}

%we should not put epsilon in the abstract, it is not defined, we can be informal here

% The last-iterate behaviors of online learning algorithms in zero-sum games have received great attentions due to their widespread use in machine learning tasks. Typically, the payoffs of these games are assumed to be time-independent. Recently, a line of research has focused on analyzing these algorithms in a time-varying game where payoffs evolve with time. However, the last-iterate behaviors of these algorithms in a time-varying environment still remain unclear. In this paper, we examine last-iterate behaviors of various algorithms in two types of time-varying games : periodic games and convergent games. These models expand upon the usual repeated game formulation and incorporate external environmental factors, such as the seasonal effects on species competition and vanishing external noise. Interestingly, we find that even in simple cases algorithms have last-iterate convergence properties in time-independent games fail to converge when the game is time-varying. Specially, we prove in periodic game, Extra-gradient will converge while several other commonly used algorithms such as optimistic gradient descent ascent and momentum method will diverge. In convergent game, we prove all these algorithms converge if the noise vanish with a rate faster than linear (\CO(1/t^{1+\epsilon})).
\end{abstract}

\input{introduction}

\input{preliminaries}

%\input{learning dynamics as difference systems}

\input{main_result}

\input{experiments}

\section{Discussion}\label{Discussion}
In this paper, we study the last-iterate behavior of extra-gradient, optimistic gradient descent ascent, and negative momentum method in two types of time-varying games : periodic game and convergent perturbed game. In the case of periodic game, we prove that extra-gradient will converge to a Nash equilibrium while other two methods diverge. To the best of our knowledge, this is the first result that provides a clear separation between the behavior of extra-gradient methods and optimistic gradient descent ascent. In the case of convergent perturbed game, we prove all three learning dynamics converge to Nash equilibrium under the BAP assumption, which is commonly used in the literature on dynamical systems. 

Our results leave many interesting open questions. Firstly, is the BAP assumption necessary for ensuring convergence in optimistic gradient descent ascent and negative momentum method? Secondly, the bound of convergence rate in Theorem \ref{thm:Condition-Pertub} may rather slight as we shown in experiments section. Obtaining a tighter bound on the convergence rate is an important future research problem. Thirdly, can the results here be generalized to other settings, such as constrained zero-sum game?

\section*{Acknowledgement}
Yi Feng is supported by the Fundamental Research Funds for the Central Universities. Ioannis
Panageas wants to thank a startup grant. Xiao Wang acknowledges Grant 202110458 from SUFE and support from the Shanghai Research Center for Data Science and Decision Technology.

\bibliographystyle{plainnat}
\bibliography{Bibli}

\newpage
\appendix
\input{appendix_1}

\end{document}

%% file: introduction.tex
\section{Introduction}

%\textcolor{red}{
%Reviewers' question in the paper "Online Learning in Periodic Zero-Sum Games "
%on openreview page:
%\begin{itemize}
%    \item More importantly, given that the goal is to investigate time-evolving games, is there a reason to believe that Poincare recurrence or time-invariant equilibria are the right measures of performance. Finally, what can be said in terms of finite time convergence, computational efficiency, etc? (for future work)
 %   \item As described in the significance, periodic zero-sum games are limited to dealing with the evolution over time and have few applications.
%\end{itemize}
%}

A central problem in game theory and min-max optimization is to come up with a pair of vectors $(x,y)$ that solves
\begin{align}
\min_{x\in \mathcal{X}}\max_{y \in \mathcal{Y}} x^{\top} A y, \label{eq:min-max problem}
\end{align}
where $\mathcal{X} \subset \mathbb{R}^n$ and $\mathcal{Y} \subset \mathbb{R}^m$ are convex sets, and $A$ is a $n\times m$ payoff matrix. The above captures two-player zero-sum games in which $x^{\top} A y$ is interpreted as the payment of the ``min player'' $x$ to the ``max player'' $y$. 
If $\mathcal{X} = \mathbb{R}^n$ and $\mathcal{Y} = \mathbb{R}^m$ the setting is called \emph{unconstrained}, otherwise it is \emph{constrained}.  
Soon after the minimax theorem of Von Neumann was established (for compact $\mathcal{X},\mathcal{Y}$), learning dynamics such as fictitious play (\cite{brown1951iterative}) were proposed for solving min-max optimization problems. Blackwell's approachability theorem (~\cite{B56}) further propelled the field of online learning, which lead to the discovery of several learning algorithms; such learning methods include multiplicative-weights-update method, online gradient descent/ascent and their optimistic variants and extra-gradient methods.

%\textcolor{red}{Should we add a paragraph on time varying game ? For example, 
%\cite{fiez2021online,Zhang22:No,Cardoso19:Competing,zhao2022efficient,Duvocelle18:Multi}.}

\paragraph{Last Iterate Convergence.} There have been a vast literature on whether or not the aforementioned dynamics converge in an average sense or exhibit last-iterate convergence when applied to zero-sum games. Dating back to Nesterov (\cite{smooth05}), there have been quite a few results showing that online learning algorithms have last-iterate convergence to Nash equilibria in zero-sum games. Examples include optimistic multiplicative weights update (\cite{DaskalakisP19, WeiLZL21}), optimistic gradient descent ascent (OGDA) (\cite{DISZ17, LiangS18}) (applied even to GANs) for unconstrained zero-sum games, OGDA for constrained zero-sum games (\cite{WeiLZL21, cai22, guathier22}) and extra-gradient methods (\cite{MertikopoulosLZ19, cai22, GorbunovLG22}) using various techniques, including sum of squares.  

Nevertheless, all aforementioned results assume that the underlying repeated zero-sum game remains invariant throughout the learning process. In many learning environments that assumption is unrealistic, see (\cite{Duvocelle18:Multi, mai18, Cardoso19:Competing}) and references therein. One more realistic learning setting is where the underlying game is actually changing; this game is called time-varying. There have been quite a few works that deal with time-varying games, where they aim at analyzing the duality gap or dynamic regret (\cite{Zhang22:No}) and references therein for OGDA and variants. However, in all these prior works, last-iterate convergence has not been investigated; the main purpose of this paper is to fill in this gap. We aim at addressing the following question:

\begin{center}
  \emph{Will learning algorithms such as optimistic gradient descent ascent or extra-gradient exhibit last-iterate converge in time-varying zero-sum games?}
\end{center}

\paragraph{Our contributions.}

We consider unconstrained two-player zero-sum games with a time-varying payoff matrix (that is the payoff matrix $A_t$ depends on time $t$)
\begin{align}
\tag{Time-varying zero sum game}
\min_{x \in \BR^n} \max_{y \in \BR^m} x^{\top}A_ty,
\end{align}

in which the payoff matrix $A_t$ varies with time in the following two ways:

\begin{itemize}
    \item \textbf{Periodic games:} $A_t$ is a periodic function with period $T$, i.e., $A_{t+T} = A_t$.
    \item \textbf{Convergent perturbed games:} $A_t = A + B_t$, $\lim_{t \to \infty} B_t = 0$.
\end{itemize}

In a repeated time-varying zero-sum game, players choose their learning algorithms and repeatedly 
play the zero-sum game. In the $t$-th round, when the players use strategy $(x_t,y_t)$, they receive their payoff $ -  A_t^{\top} x_t$ and $A_t y_t$. 

In this paper we show the following results:
\begin{itemize}
\item For \textbf{periodic games:} We prove that when two players use extra-gradient, their strategies will converge to the common Nash equilibrium of the games within a period with an exponential rate, see Theorem \ref{thm:EG-Period} for details. Additionally, we provide an example where optimistic gradient descent ascent and negative momentum method diverge from the equilibrium with an exponential rate, see Theorem \ref{thm:OGDMM-Period}.

To the best of our knowledge, this is the first result that provides a clear separation between the behavior of extra-gradient methods and optimistic gradient descent ascent.

\item For \textbf{convergent perturbed games:} Assuming $\sum^{\infty}_{t=1}\lVert B_t \lVert_2 $ is bounded, we prove that the extra-gradient, optimistic gradient descent ascent, and negative momentum method all converge to the Nash equilibrium of the game defined by payoff matrix $A$ with a rate determined by $\{B_t\}_t$ and singular values of $A$, see Theorem \ref{thm:Condition-Pertub}. Furthermore, we prove that extra-gradient will asymptotically converge to equilibrium without any additional assumptions on perturbations besides $\lim_{t \to \infty}B_t = 0$, see Theorem \ref{thm:EG-Pertub}. 
\end{itemize}

\paragraph{Related work on time-varying games.}The closest work to ours that argues about stabilization of mirror descent type dynamics on convergent strongly monotone games is (\cite{ Duvocelle18:Multi}). Most of the literature has been focused on proving either recurrence/oscillating behavior of learning dynamics in time-varying periodic games (\cite{mai18, fiez2021online} ) and references therein or performing regret analysis (\cite{DBLP:journals/corr/abs-1912-01698, Cardoso19:Competing,Zhang22:No}). In particular, the latter work extends results on RVU (\cite{SyrgkanisALS15}) bounds to argue about dynamic regret (\cite{Zinkevich03}). 

\paragraph{Technical Comparison.}We investigate the last iterate behaviors of learning dynamics through their formulation of linear difference systems. This approach has also been used to establish last iterate convergence results for learning dynamics in time-independent games
(\cite{zhang2019convergence,liang2019interaction,gidel2019negative}). One common key point of these works is to prove that a certain matrix has no eigenvalues with modulus larger than $1$, then the last iterate convergence can be guaranteed by the general results of autonomous difference systems. However, this method cannot be generalized to the time-varying games where the corresponding difference systems are non-autonomous. In particular, the dynamical behavior of a  non-autonomous system is not determined by the eigenvalues of a single matrix. In fact, it is difficult to establish convergence/divergence results even for non-autonomous system with special structures, such as  periodic or perturbed systems. In this paper, to get such results, we employ both general results in linear difference systems, such as Floquet theorem and Gronwall inequality, and special structures of the difference systems associated to learning dynamics.

\paragraph{Organization.} In Section \ref{Preliminaries}, we present the necessary background for this work. The main results are stated in Section \ref{main_result}. In Section \ref{experiments}, we present numerical experiments and in Section \ref{Discussion}, we conclude with a discussion and propose some future research problems.

%% file: preliminaries.tex
\section{Preliminaries}\label{Preliminaries}

\subsection{Definitions}

\textbf{Zero-sum blinear game.} An unconstrained two players zero-sum game consists of two agents $\CN = \{1,2\}$, and losses of both players are determined via payoff matrix $A \in \BR^{n \times m}$. Given that player $1$ selects strategy $x \in \BR^{n}$ and player $2$ selects strategy $y \in \BR^{m}$, player 1 receives loss $u_1(x,y) = \langle x, Ay \rangle$, and player 2 receives loss $u_2(x,y) = - \langle y, A^{\top}x \rangle$.  Naturally, players want to minimize their loss resulting the following min-max problem: 
\begin{equation}
\min_{x \in \BR^{n}} \max_{y \in \BR^{m}} x^{\top} Ay \tag{Zero-Sum Game}
\end{equation}
Note that the set $\{(x^*,y^*) | A^{\top}x^*=0, Ay^*=0\}$ represents the set of equilibrium of the game.

\textbf{Time-varying zero-sum bilinear game.} In this paper, we study games in which the payoff matrices vary over time and we define two kinds of such time-varying games.

\begin{defn}[Periodic games]\label{pg2} A periodic game with period $T$ is an infinite sequence of zero-sum bilinear games 
$\{A_t\}^{\infty}_{t = 0} \subset  \BR^{n \times m}$, and $A_{t+T} = A_t$ for all $t \ge 0$. 
\end{defn}
Note that the periodic game defined here is the same as Definition 1 in (\cite{fiez2021online}) except for the fact that we are considering a discrete time setting. Therefore, we do not make the assumption that payoff entries are smoothly dependent on $t$.

\begin{defn}[Convergent perturbed games]\label{cg} A convergent perturbed game is an infinite sequence of zero-sum bilinear games 
$\{A_t\}^{\infty}_{t = 0} \subset  \BR^{n \times m}$, and $\lim_{t \to \infty}A_t = A $ for some $A \in \BR^{n \times m}$. Equivalently, write $A_t = A + B_t$, then
$\lim_{t \to \infty} B_t = 0$. We will refer to the zero-sum bilinear game defined by $A$ as \textbf{stable game}.
\end{defn}

%Both of these three algorithms have been proved to have last-iterate convergence property. A unified proof and related work of these algorithms can be found in
%\cite{zhang2019convergence}.

\textbf{Learning dynamics in games.} In this paper, we consider three kinds of learning dynamics : optimistic gradient descent ascent (OGDA), extra-gradient (EG) , and negative momentum method. All these methods possess the last-iterate convergence property in repeated game with a time-independent payoff matrix, as demonstrated in previous literature. However, here we state their forms within a time-varying context.

\textbf{Optimistic gradient descent-ascent.} 
We study the optimistic  descent ascent method (OGDA) defined as follows:
\begin{align*}\label{Optimistic Gradient}
\tag{OGDA}
& x_{t+1} = x_{t} - 2 \eta A_{t} y_{t} + \eta A_{t-1} y_{t-1}, \\
& y_{t+1} = y_{t} + 2 \eta A^{\top}_{t} x_{t} - \eta A^{\top}_{t-1} x_{t-1}.
\end{align*}

Optimistic gradient descent ascent method was proposed in (\cite{popov1980modification}), and here we choose the same parameters as  (\cite{DISZ17}). The last iterate convergence property of OGDA in unconstrained bilinear game with a time-independent payoff was proved in (\cite{DISZ17}). Recently, there are also works analyzing the regret behaviors of OGDA under a time varying setting (\cite{Zhang22:No,anagnostides2023convergence}).

\textbf{Extra gradient.}
We study the extra gradient descent ascent method (EG) defined as follows:

\begin{align*}\label{Extra Gradient}
\tag{EG}
&x_{t+\frac{1}{2}} = x_t - \gamma A_ty_t, \ \ y_{t+\frac{1}{2}} = y_t + \gamma A_t^{\top}x_t, \\
&x_{t+1} = x_t - \alpha A_ty_{t+\frac{1}{2}},\ \  y_{t+1} = y_t + \alpha A_t^{\top}x_{t+\frac{1}{2}}.
\end{align*}

Note that the extra-gradient method first calculates an intermediate state before proceeding to the next state. Extra-gradient was firstly proposed in (\cite{korpelevich1976extragradient}) with the restriction that $\alpha = \gamma$. Here we choose the parameters same as in (\cite{LiangS18}), where the linear convergence rate of extra-gradient in the bilinear zero-sum game with time-independent was also proven. Convergence of extra-gradient on convex-concave game was  analyzed in (\cite{nemirovski2004prox,monteiro2010complexity}), and convergence guarantees for special non-convex-non-concave game was provided in (\cite{MertikopoulosLZ19}).

\textbf{Negative momentum method.} We study the alternating negative momentum method (NM), defined as follows:
\begin{align*}\label{Negative Momentum}
\tag{NM}
& x_{t+1} = x_{t} -  \eta A_{t} y_{t} + \beta_1(x_{t} - x_{t-1}) , \\
& y_{t+1} = y_{t} +  \eta A_{t+1} ^{\top}x_{t+1}+ \beta_2(y_{t} - y_{t-1}),
\end{align*}
where $\beta_1,\beta_2 \le 0$ are the momentum parameters.

Applications of negative momentum method in game optimization was firstly proposed in  (\cite{gidel2019negative}). Note that the algorithm has an alternating implementation: the update rule of $y_{t+1}$ uses the payoff
$A_{t+1} ^{\top}x_{t+1}$, thus in each round, the second player chooses his strategy after the first player has chosen his. It was shown in (\cite{gidel2019negative}) that both the negative momentum parameters and alternating implementations are crucial for convergence in bilinear zero-sum games with time-independent payoff matrices. Analysis of the convergence rate of negative momentum method in strongly-convex strongly-concave games was provided in (\cite{zhang2021suboptimality}).

\subsection{Difference systems}

The analysis of the last-iterate behavior of learning algorithms can be reduced to analyzing the dynamical behavior of the associated linear difference systems  (\cite{zhang2019convergence,daskalakis2018limit}). When the payoff matrix is time-independent, the associated difference systems is autonomous. However, as we are studying games with time-varying payoff matrices, we have to deal with non-autonomous difference systems. In general, the convergence behavior of non-autonomous difference systems is much more complicated than that of autonomous ones (\cite{colonius2014dynamical}).

\textbf{Linear difference system.}
Given a sequence of iterate matrices $ \{\mathcal{A}_t\}^{\infty}_{t = 1} \subset \BR^{n \times n}$ and initial condition $X_0 \in \BR^n$, a linear difference system has form
\begin{align}\label{lds}
\tag{Linear difference system}
    X_{t+1} = \mathcal{A}_t X_t.
\end{align}

If $\mathcal{A}_t \equiv \mathcal{A}$ is a matrix independent of time, the system is called an \textbf{autonomous system}, otherwise, it is called a \textbf{non-autonomous system}. 
%In three learning dynamics \eqref{Optimistic Gradient}, \eqref{Extra Gradient} and \eqref{Negative Momentum}, $\mathcal{A}_t$ corresponds to the iterate matrix in \eqref{eq:OGDA}, \eqref{eq:EG} and \eqref{eq:Negative Momentum} respectively.
We care about the asymptotic behavior of $X_{t}$, that is, what can we say about $X_{t}$ as $t \to \infty$. 
\begin{defn}
    A point $X$ is called \textbf{asymptotically stable} under the linear difference system  if
\begin{align*}
    \exists \ \delta>0, \forall \ Y,\  s.t. \ \lVert Y- X \lVert_2 \le \delta \Rightarrow 
    \lim_{t \to \infty} \lVert (\prod^{\infty}_{t = 1} \mathcal{A}_t ) Y-  X \lVert_2 = 0.
\end{align*}
Moreover, $X$ is called \textbf{exponentially asymptotically stable} if the above limit has an exponentially convergence rate, i.e., $\exists \ \alpha \in (0,1)$, such that
\begin{align*}
      \lVert ( \prod^{s}_{t = 1} \mathcal{A}_t ) Y-  X \lVert_2 \le \alpha^s \lVert  Y-  X \lVert_2.
\end{align*}
\end{defn}

It is well known that for autonomous linear systems, being asymptotically stable is equivalent to being exponentially asymptotically stable, as shown in Thm 1.5.11 in (\cite{colonius2014dynamical}). However, this equivalence does not hold for non-autonomous systems. \footnote{To gain intuition of the differences between autonomous and non-autonomous system, consider the following  simple 1-dimension example:
$x_{t+1} = (1- \frac{1}{t+1}) x_{t}$. $0$ is a stationary point of this system, but every initial points converges to $0$ with a rate $\CO( \frac{1}{t})$, thus $0$ is asymptotically stable but not exponentially asymptotically stable . }

Linear difference system has a formal solution
$
    X_{T+1} = \prod^{T}_{t=0}\mathcal{A}_t X_0.
$
However, such a representation does not yield much information about the asymptotic behavior of solution as $t \to \infty$, except in the case of autonomous system. In this paper, we mainly consider two classes of non-autonomous linear difference systems: if the iterate matrix $\mathcal{A}_t$ is a periodic function of $t$, the system is called a periodic system; and if $\mathcal{A}_t$ has a limit as $t$ tends to infinity, the system is called a perturbed system. 

\textbf{Periodic linear system.}
If the iterate matrix $\mathcal{A}_t$ in a linear difference system satisfies $\mathcal{A}_t=\mathcal{A}_{t+\CT}$, $\forall t \in \BZ$, 
then the system is called a periodic linear system. Denote $\widetilde{\mathcal{A}} = \prod^\CT_{j=1}\mathcal{A}_{\CT-j}$. For a $T$-periodic equation, the eigenvalues $\alpha \in \BC$ of $\widetilde{\mathcal{A}} $ is called the Floquet multipliers, and the Floquet exponents are defined by $\lambda_j = \frac{1}{\CT} \ln(\lvert \alpha_j \lvert)$. The following Floquet theorem characterizes the stability of a periodic linear difference equation.

\begin{prop}[Floquet theorem, (\cite{colonius2014dynamical})]\label{Floquet} 
%\tag{Floquet theorem}
The zero solution of a periodic linear difference equation is asymptotically stable if and only if all Floquet exponents are negative.
\end{prop}

Although Floquet theorem provides a method to determine whether a periodic system  converges in general, considerable further work may be necessary in order to obtain explicit convergence criteria for specific equation. That is because even if we know the modulus of the largest eigenvalue of each $\mathcal{A}_t$, it is usually difficult to compute the the modulus of the largest eigenvalue of each $\prod^{\CT}_{t=0}\mathcal{A}_t$ due to the complex behavior of eigenvalues under matrix multiplication.

\textbf{Perturbed linear system.} If the iterative matrix $\mathcal{A}_t$ in a linear difference system satisfies $\mathcal{A}_t=\mathcal{A} +\mathcal{B}_t $ and $\lim_{t \to \infty}  \mathcal{B}_t =0$, then the system is called a perturbed linear system. The convergence behavior of a 
perturbed linear system is not clear in the literature. A general result in this direction is the following Perron's theorem:

\begin{thm}[Perron's theorem, (\cite{pituk2002more})]\label{thm:Perron} If $X_n$ is a solution of a perturbed linear system, then either $X_n = 0$ for all sufficient large $n$ or $\rho = \lim_{n \to \infty} \sqrt[n]{\lVert X_n \lVert_2}$ exists and is equal to the modulus of one of the eigenvalues of matrix $\mathcal{A}$.
\end{thm}
This result can only guarantee the convergence of $X_n$ when all eigenvalues of $\mathcal{A}$ have modulus smaller than 1. In this case, $\lim_{n \to \infty} X_n = 0$. However, for analyzing the non-autonomous linear systems associated to learning dynamics, it is not sufficient as we will show that the stablized matrix of these systems generally has eigenvalues equal to $1$. 
%More can be said about the convergence behavior of special kinds of perturbed linear systems under stricter constrains about the vanishing rate of $\mathcal{B}_t$. For example, \cite{benzaid1987asymptotic} considers the special case when $\mathcal{A}$ is a diagonal matrix, and prove an asymptotic result under the assumption of bounded accumulated perturbation (BAP) :
%\begin{align*}
%    \sum^{\infty}_{t = 0} \lVert \mathcal{B}_t \lVert_2 < \infty.
%\end{align*}
%Such results need strong assumptions so that they cannot be applied to our case. However, the above assumption is also necessary for our purpose.

%% file: main_result.tex
\section{Main results}\label{main_result}

In this section, we present our main results. We present the relationship between learning dynamics and linear difference equation in Section \ref{ldsld}, investigate the last-iterate behaviors of learning dynamics in a periodic game in Section \ref{pg}, and investigate the last-iterate behaviors of learning dynamics in a convergent perturbed game in Section \ref{peg}. Proofs are deferred to the appendix.
%Firstly, we show that Extra-gradient will asymptotically converge to the Nash equilibrium of the stable game without any further assumptions. After that, we prove that all three game dynamics will converge to the Nash equilibrium of the stable game under the  \ref{bap} assumption.

\subsection{Learning dynamics as linear difference systems}\label{ldsld}
Formalizing learning dynamics as linear difference systems is useful for studying
their dynamical behaviors. In the following, we present the formulation of optimistic gradient descent ascent, extra-gradient, and negative momentum method
as linear difference systems.

\begin{prop}
Optimistic gradient descent ascent can be written as the following linear difference system: 
\begin{align}\label{eq:OGDA}
\begin{bmatrix}
x_{t+1} \\
y_{t+1} \\
x_{t} \\
y_{t}
\end{bmatrix} 
=
\begin{bmatrix}
I & -2\eta A_{t} & 0 & \eta A_{t-1} \\
2\eta A_{t}^{\top} & I & -\eta A_{t-1}^{\top} & 0\\
I & 0& 0&0 \\
0& I& 0&0 
\end{bmatrix} 
\begin{bmatrix}
x_{t} \\
y_{t} \\
x_{t-1} \\
y_{t-1}
\end{bmatrix} .
\end{align}

Extra-gradient can be written as the following linear difference system : 
\begin{equation}\label{eq:EG}
\left[
\begin{array}{c}
x_{t+1}\\
\\
y_{t+1}
\end{array}
\right]
=
\left[
\begin{array}{cc}
 I - \alpha\gamma A_tA_t^{\top}  & -\alpha A_t\\
 \\
\alpha A_t^{\top} & I - \gamma\alpha A_t^{\top}A_t
\end{array}
\right]
\left[
\begin{array}{c}
x_{t}\\
\\
y_{t}
\end{array}
\right].
\end{equation}

Negative momentum method can be written as the following linear difference system:
\begin{align}\label{eq:Negative Momentum}
\begin{bmatrix}
x_{t+1} \\
\\
y_{t+1} \\
\\
x_{t} \\
\\
y_{t}
\end{bmatrix} 
=
\begin{bmatrix}
(1+\beta_1) I& -\eta A_{t} & -\beta_1 I & 0 \\
\\
\eta(1+\beta_1) A_{t+1}^{\top} & (1+\beta_2) I-\eta^2 A_{t+1}^{\top}A_{t}  & -\eta\beta_1 A_{t+1}^{\top} & -\beta_2 I\\
\\
I & 0& 0&0 \\
\\
0& I& 0&0 
\end{bmatrix} 
\begin{bmatrix}
x_{t} \\
\\
y_{t} \\
\\
x_{t-1} \\
\\
y_{t-1}
\end{bmatrix} .
\end{align}
\end{prop}

It is easy to verify that these linear difference systems are equivalent to their corresponding learning dynamics by directly writing down the matrix-vector product. 

We will also refer to the iterative matrix of these linear difference systems as the \textbf{iterative matrix of their corresponding learning dynamics}. In the following, we will study the convergence/divergence behaviors of these linear difference systems under the condition that $\{A_t\}_t$ is a periodic game or convergent perturbed game. Note that although an intermediate step $(x_{t+\frac{1}{2}},y_{t+\frac{1}{2}} )$ is required in extra-gradient, it is eliminated in \eqref{eq:EG}.

\subsection{Periodic games}\label{pg}

Recall that in a periodic game with period $\CT$, the payoff matrices $\{A_s\}^{\infty}_{s=1}$ satisfy $A_{s+\CT}=A_s$ for any $s>0$. Define 
\begin{align}
    \Delta_{i,t} = \lVert A_i^{\top}x_t \lVert_2 +  \lVert A_i y_t  \lVert_2
\end{align}
 for $i \in [\CT]$. As in (\cite{daskalakis2018limit}), we use $\Delta_{i,t}$ as a measurement of the distance between the current strategy $(x_t,y_t)$ and a Nash equilibrium of the zero-sum bilinear game defined by the kernel space of payoff matrix $A_i$. Note that if $(x^*,y^*)$ is a Nash equilibrium of the game defined by $A_i$, then $(A_i^{\top} x^*, A_i y^*) = (0,0)$, and
 \begin{align*}
     \Delta_{i,t} & =  \lVert A_i^{\top} (x_t -  x^*) \lVert_2 +  \lVert A_i (y_t - y^*) \lVert_2 
 \end{align*}
thus when strategy is close to equilibrium,  $\Delta_{i,t}$ will be small. Moreover, $\Delta_{i,t} = 0$ if and only if $(x_t,y_t)$ is an equilibrium. In this section, we will consider the convergence/growth rate of $\Delta_{i,t}$ at a function of $t$.

We firstly consider Extra-gradient method. Denote the iterative matrix in the linear difference form of Extra-gradient \eqref{eq:EG} as $\mathcal{A}_t$. In the following theorem, we prove that if two players use the Extra-gradient method, their strategies will converge to the common Nash equilibrium of games in the period with an exponential rate.

%and for a fixed $t$, $\ker(\mathcal{A}_t - I)$ consists of all its stationary points, i.e., $\mathcal{A}_t X = X$. Denote $S = \cap^T_{t = 1} \ker(\mathcal{A}_t - I)$ as the common stationary points of \eqref{ext2} for $t = 1,...,T$. Denote $\sigma$ the singular value of matrix $A_t$. We can write the characteristic equations as follows:
%\begin{align*}
%    (\lambda-1)^2+2\gamma\alpha\sigma^2(\lambda-1) + (\alpha^2\sigma^2+\alpha^2\gamma^2\sigma^4)=0
%\end{align*}
%The Schur condition is 
%\begin{align}\label{eq:EG-Schur}
%    \alpha^2\sigma^2+(\alpha\gamma \sigma^2-1)^2 <1
%\end{align}

%\begin{restatable}{theorem}{EGPeriod}
%\label{thm:EG-Period}
%    When two players use \eqref{Extra Gradient} in a periodic games with period $T$, with step size $\alpha$ and $\gamma$ satisfies \eqref{eq:EG-Schur} for all $\sigma \in \cup^T_{i = 1} \{\sigma | \sigma \in \textnormal{Eigenvalue} \left(A_i\right), \sigma \neq 0\}$, then for any $\epsilon > 0$,  $(x_t,y_t)$ satisfies
%    \begin{align*}
%        \textnormal{dist}\left( (x_t,y_t) , S \right) \in \CO \left((\lambda_*+\epsilon)^{t/T} \right),
%    \end{align*}
%where $\lambda_* = \max \{ \ \lvert \lambda \lvert \ | \ \lambda \in \textnormal{Eigenvalue} 
%\left( \prod^T_{t=1}  \mathcal{A}_t \right), \lambda \ne 1   \}.$ Moreover, $\lambda_*$ is strictly smaller than $1$.
%\end{restatable}

\begin{restatable}{theorem}{EGPeriod}
\label{thm:EG-Period}
    When two players use extra-gradient in a periodic games with period $T$, with step size $\alpha = \gamma < \frac{1}{\sigma}$ where $
        \sigma = \max \{ \sigma' | \sigma' \textnormal{\ is a singular value of\ } 
        A_i \textnormal{\ for some\ } i \in [\CT] \}.
    $
 Then 
    \begin{align*}
        \Delta_{i,t} \in \CO \left((\lambda_*)^{t/\CT} \cdot \textnormal{Poly}(t) \right), \ \forall i \in [\CT]
    \end{align*}
where $\lambda_* = \max \{ \ \lvert \lambda \lvert \ | \ \lambda  \textnormal{\ is an eigenvalue of \ } 
\left( \prod^\CT_{t=1}  \mathcal{A}_t \right), \lambda \ne 1   \}$, and $\lambda_* < 1$.
\end{restatable}

Note that in a periodic game, the iterative matrices of learning dynamics are also periodic, which means that the learning difference systems of these learning dynamics are periodic systems. According to Floquet theorem, see proposition \ref{Floquet}, the study of dynamical behaviors of a periodic system can be reduced to the autonomous system whose asymptotic behavior is determined by  $\prod^\CT_{t=1}  \mathcal{A}_t $. 

The key point on the proof of Theorem \ref{thm:EG-Period} is an observation that the iterative matrix of extra-gradient is a normal matrix, which makes it possible to calculate the Jordan normal form of $\prod^\CT_{t=1}  \mathcal{A}_t $ for arbitrary large $\CT$. The details of proof are left to Appendix \ref{egpp}.

%\textcolor{red}{To prove of theorem \ref{thm:EG-Period}, we ...., the details are left to ...}

In the following theorem, we provide an example demonstrating that when two players use the optimistic gradient descent ascent or negative momentum method, their strategy will diverge at an exponential rate, regardless of how they choose their step-sizes and momentum parameters.

\begin{restatable}{theorem}{OGDMMPeriod}
\label{thm:OGDMM-Period}
Consider a periodic game with period $\CT = 2$, and described by the following payoff matrix
\begin{align}\label{p3mg}
 A_t=
\begin{cases}
\left[1,-1\right], & t \textnormal{\ \ is \ odd} \\ 
\left[-1,1\right], & t\textnormal{\ \ is \ even}
\end{cases}
\end{align}
with $x_t \in \BR,\ y_t \in \BR^2$. If two players use optimistic gradient descent ascent or negative momentum method, then regardless of how they choose step sizes and momentum parameters, we have
\begin{align*}
    \sup_{s \in [t]}\Delta_{i,s} \in \Omega(\lambda^t), \ \textnormal{where\ } \lambda > 1,\ i \in \{1,2\}.
\end{align*}
 Here $\lambda$ is determined by the largest modulus of the eigenvalues of the iterative matrix of optimistic gradient descent ascent or negative momentum method.
\end{restatable}

To prove theorem \ref{thm:OGDMM-Period}, we directly calculate the characteristic polynomials of the iterative matrices products in one period for optimistic gradient descent ascent and negative momentum method under the game defined by \eqref{p3mg}. To show that these systems have an exponential divergence rate, it is sufficient to demonstrate that their characteristic polynomials have a root with modulus larger than $1$. We achieve this by using the Schur stable theorem, which is also employed to demonstrate the last iterate convergence of several learning dynamics in time-independent game (\cite{zhang2019convergence}). The proof is deferred to Appendix 
\ref{onpg}.

In Figure \eqref{period2}, we present the function curves of $\Delta_{1,t}$ for these three types of game dynamics under the periodic game defined by  \eqref{p3mg}. From the experimental results, extra-gradient converges, while both optimistic gradient descent ascent and negative momentum method diverge.

\begin{figure}[h]
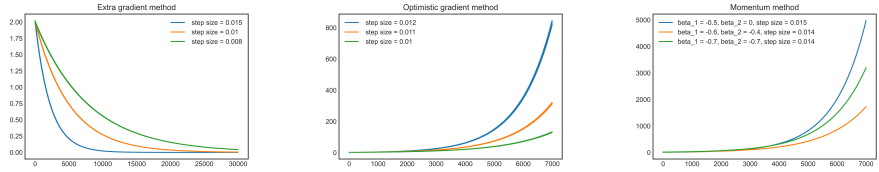

\centering
\subfigure
{
    \begin{minipage}[b]{.28\linewidth}
        \centering
        \includegraphics[scale=0.25]{extra.pdf}
    \end{minipage}
}
\subfigure
{
 	\begin{minipage}[b]{.28\linewidth}
        \centering
        \includegraphics[scale=0.25]{optimistic.pdf}
    \end{minipage}
}
\subfigure
{
 	\begin{minipage}[b]{.28\linewidth}
        \centering
        \includegraphics[scale=0.25]{negative_momentum.pdf}
    \end{minipage}
}
\caption{Function curves of $\Delta_{1,t}$ of the game presented in Theorem \ref{thm:OGDMM-Period}. Extra-gradient converges, while the other two methods diverge.}
\label{period2}
\end{figure}

\subsection{Convergent perturbed game}\label{peg}

Recall that the payoff matrix of a convergent perturbed game has form  $A_t = A + B_t$, where $\lim_{t \to \infty} B_t = 0$, and we refer to the zero-sum game defined by payoff matrix $A$ as the stable game. We denote
\begin{align}
    \Delta_t = \lVert A^{\top} x_t \lVert_2 + \lVert A y_t \lVert_2,
\end{align}
thus $\Delta_t$ measures the distance between the strategy $(x_t,y_t)$ in the $t$-th round and the Nash equilibrium of the stable game defined by $A$. Moreover,
$\Delta_t = 0$ if and only if $(x_t,y_t)$ is an equilibrium of the stable game.

In the literature on linear difference systems, a common assumption that needs to be added for convergence guarantee is the following bounded accumulated perturbations (BAP) assumption (\cite{benzaid1987asymptotic,elaydi1999asymptotic,elaydi1995asymptotic}):
\begin{align*}\label{bap}
\tag{BAP assumption}
    \sum_{t=0}^{\infty} \lVert B_t \lVert_2 \textnormal{\ is \ bounded.}
\end{align*}

In the following theorem, we prove that under BAP assumption, all three learning dynamics considered in this paper will make $\Delta_t$ converge to $0$, with a rate dependent on the vanishing rate of $B_t$. 

\begin{restatable}{theorem}{ConditionPertub}
\label{thm:Condition-Pertub} 
 Assume that the \eqref{bap}  holds, i.e., $\sum_{t=0}^{\infty} \lVert B_t \lVert_2$ is bounded, and let $\sigma$ be the maximum modulus of the singular value of payoff matrix $A$, then with parameters choice: 
\begin{itemize}
    \item for extra-gradient with step size $\alpha = \eta < \frac{1}{2 \sigma}$, 
    \item for optimistic gradient descent ascent with step size $\eta < \frac{1}{2 \sigma}$,
    \item for negative momentum method with step size $\eta < \frac{1}{\sigma}$ and momentum parameters $\beta_1 = -\frac{1}{2}$ and $\beta_2 = 0$,
\end{itemize}
we have $\Delta_t$ converge to $0$ with rate $\CO (f(t))$. Here 
\begin{align*}
    f(t) = \max\{ \lambda^t, \sum^{\infty}_{i =t/2} \lVert B_i \lVert_2 \},
\end{align*}
 and $\lambda \in (0,1)$ is determined by the eigenvalues of the iterative matrix of corresponding learning dynamics and the payoff matrix $A$ of the stable game.
\end{restatable}

There are two main ingredients in the proof of Theorem \ref{thm:Condition-Pertub}: firstly, we show that the iterative matrices associated with these learning dynamics are diagonalizable; secondly, these matrices do not have eigenvalues with modulus larger than $1$. Moreover, we prove a general results which states any linear difference system satisfying these two conditions will converge. The details of proof are left to Appendix \ref{appx:Condition-Pertub}.

%To prove Theorem \ref{thm:Condition-Pertub}, we first use Lemma \ref{lm:eig-static} to show that eigenvalues of iterative matrix $\tilde{\mathcal{A}}$ consists of 1 and other eigenvalues whose  modulus are less than 1. In addition, by Proposition \ref{prop:diagonal}, the iterative matrix $\tilde{\mathcal{A}}$ is diagonalizable. 
% We also give Lemma \ref{lm:norm bounded equivalence} showing that the BAP assumption for $\{B_t\}_t$ holds implies the BAP assumption for $\{\mathcal{B}_t\}_t$ holds. 
%So far, we already satisfy the assumptions of Lemma \ref{lm:SVD convergence}, which proves that $\lim_{t \to \infty} \lVert (\tilde{\mathcal{A}}- I)\tilde{X}_t \lVert_2=  \boldsymbol{0}$.
%Lemma \ref{lm:convergence equivalent} states that it implies $\lim_{t \to \infty} \left(\lVert A^{\top}x_t\lVert_2 + \lVert A y_t \lVert_2 \right)= 0$. The details of proof are left to Appendix \ref{appx:Condition-Pertub}.

\begin{rem}
The \eqref{bap} can be converted into a constraint on the vanishing rate of $B_t$: if $B_t$ has a vanishing rate like $\CO(\frac{1}{t^{1+\epsilon}})$, for some arbitrary small $\epsilon > 0$, then $\sum_{t=0}^{\infty} \lVert B_t \lVert_2$ is bounded. We also note that the condition for $B_t$ to vanish at a rate $\CO(\frac{1}{t^{1+\epsilon}}),\ \forall \epsilon > 0$ is necessary to ensure convergence in general linear difference system. For example, consider the 1-dimension system $x_{t} = (1+\frac{1}{t}) x_{t-1}$ where a $\CO(\frac{1}{t})$ convergence rate of perturbations leads to $x_t$ diverging with a $\Omega(t)$ rate.
\end{rem}

%Specially, for \eqref{Extra Gradient}, we can prove that when $\lim_{t \to \infty} B_t = 0 $, their strategy will always converge, no matter $\sum_{t=0}^{\infty} \lVert B_t \lVert$ is bounded or not.

Surprisingly, in the next theorem, we prove that Extra-gradient makes $\Delta_t$ asymptotically converge to $0$, without making any further assumptions about the converge rate of $B_t$.

\begin{restatable}{theorem}{EGPertub}
\label{thm:EG-Pertub}
In a convergent perturbed game, if two players use Extra-gradient, there holds $\lim_{t \to \infty} \Delta_t = 0$ with step size $\alpha = \eta < \frac{1}{2\sigma}$ where $\sigma$ is the maximum modulus of the singular value of payoff matrix $A$.
\end{restatable}
To prove Theorem \ref{thm:EG-Pertub}, we first observe that $\lVert x_t \lVert_2 +\lVert y_t \lVert_2 $ a non-increasing function of $t$ since the iterative matrix of extra-gradient is a normal matrix. Next, we prove that if Theorem \ref{thm:EG-Pertub} doesn't hold, then $\lVert x_t \lVert_2 +\lVert y_t \lVert_2 $ will decrease by a constant for infinite number of times, thus leading to a contradiction with the non-increasing property and the non-negativity of $\lVert x_t \lVert_2 +\lVert y_t \lVert_2 $ . The proof is deferred to Appendix \ref{extra_convergent}.

%there always exists constant $c>0$ and infinitely many $t$ such that the sum of norm of $x_t$ and $y_t$ is smaller than sum of norm of $x_{t-1}$ and $y_{t-1}$ minus $c$ which combining with the non-increasing property will result in the sum of norm of $x_t$ and $y_t$ will be smaller than 0, that is contradiction to the nonnegativity of norm.
% the above non-increasing result also doesn't hold where appears contradiction. 
% the effect of $B_t$ will be gradually dissipate through iteration, regardless of how slowly the rate of perturbations disappears. 

%Because the iterative matrix of Extra-gradient method is always normal once the condition of parameter is satisfied, there exists a set of orthogonal basis according to the iterative matrix whose payoff matrix is $A$. Then we can decompose $\mathcal{B}_tX_t$ under this basis and

%Note that theorem \ref{thm:EG-Pertub} implies Extra-gradient will converge to the equilibrium regardless how slow the rate of ${B}_t$ vanishes. \textcolor{blue}{I think these two lines should be deleted.}

%% file: experiments.tex
\section{Experiments}\label{experiments}

In this section we present numerical results for our theoretical results in Section \ref{main_result}.

\textbf{Experiments on Theorem \ref{thm:EG-Period}} We verify Theorem \ref{thm:EG-Period} through a period-$3$ game, the payoff matrices are chosen to be
\begin{align*}
    A_1 = \begin{bmatrix}
        1 & 2 \\
        2 & 4
    \end{bmatrix},
    A_2 = \begin{bmatrix}
        3 & 7 \\
        7 & 1
    \end{bmatrix},
    A_3 = \begin{bmatrix}
        4 & 2 \\
        4 & 2
    \end{bmatrix}.
\end{align*}
We run extra-gradient and optimistic gradient descent ascent on this example, both with step size = 0.01, the experimental results are presented in Figure \eqref{pepo1}. We can see extra-gradient (left) makes $\Delta_{i,t}$ converge, while optimistic gradient descent ascent (right) makes $\Delta_{i,t}$ diverge. This result supports Theorem \ref{thm:EG-Period} and provides a numerical example of the separation of extra-gradient and optimistic gradient descent ascent in periodic games.

\begin{figure}[ht]
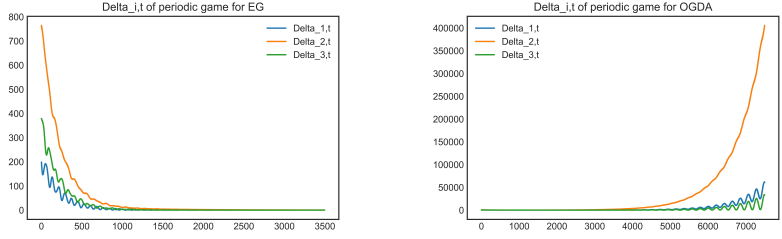

\centering
\subfigure
{
    \begin{minipage}[b]{.4\linewidth}
        \centering
        \includegraphics[scale=0.35]{pe.pdf}
    \end{minipage}
}
\subfigure
{
    \begin{minipage}[b]{.4\linewidth}
        \centering
        \includegraphics[scale=0.35]{po.pdf}
    \end{minipage}
}
\caption{Function curves of $\Delta_t$ for extra-gradient (left), and optimistic gradient descent ascent (right).}
\label{pepo1}
\end{figure}

\textbf{Experiments on Theorem \ref{thm:Condition-Pertub}} We verify Theorem \ref{thm:Condition-Pertub} by examples:
\begin{align*}
    A = \begin{bmatrix}
        2 & 3 \\
        4 & 6
    \end{bmatrix},
    B_{1,t} = B \cdot t^{-1.1}, 
    B_{2,t} = B \cdot t^{-4}, \textnormal{\ and\ },
    B_{3,t} = B \cdot t^{-8}
\end{align*}
where $B = [[-15,70],[-90,90]]$. The step size is chosen to be $0.005$. The initial points are chosen to be $x_0 = [15,13], x_{-1} = [11,3]$ and
$y_0 = [35,1],  y_{-1} = [35,1]$. The experimental results are presented in Figure \ref{ec1}, all of the three dynamics make $\Delta_t$ converge to $0$, and slower convergence rate of perturbations can decelerate the convergence rate of learning dynamics, thus support the convergence result in Theorem  \ref{thm:Condition-Pertub}. We also observe that the convergence rate is faster than the upper bound provided in the theorem. Therefore, we conjecture the bound of convergence rate can be improved.

\begin{figure}[h]
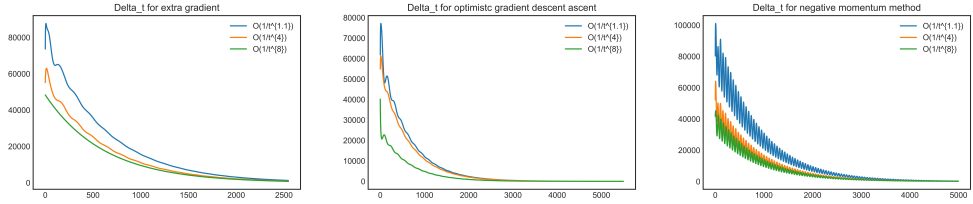

\centering
\subfigure
{
    \begin{minipage}[b]{.3\linewidth}
        \centering
        \includegraphics[scale=0.3]{per_et.pdf}
    \end{minipage}
}
\subfigure
{
    \begin{minipage}[b]{.3\linewidth}
        \centering
        \includegraphics[scale=0.3]{per_ogd.pdf}
    \end{minipage}
}
\subfigure
{
    \begin{minipage}[b]{.3\linewidth}
        \centering
        \includegraphics[scale=0.3]{per_nm.pdf}
    \end{minipage}
}
\caption{Values of $\Delta_t$ for extra-gradient (left), optimistic gradient descent ascent (middle), negative momentum method (right).}
\label{ec1}
\end{figure}

\textbf{Experiments on Theorem \ref{thm:EG-Pertub}} We verify Theorem \ref{thm:EG-Pertub} by two group of examples. The perturbations are :
 \begin{align*}
    B_{1,t} = B \cdot t^{-0.4},\ 
    B_{2,t} = B \cdot t^{-0.3},\ 
    B_{3,t} = B \cdot t^{-0.2},
\end{align*}
and
\begin{align*}
    B_{4,t} =B \cdot \log(t)^{-1.8},\ 
    B_{5,t} =B \cdot \log(t)^{-1.5},\ 
    B_{6,t} =B \cdot \log(t)^{-1.3}.
\end{align*}
where $B = [[-10,10],[-10,10]]$. The payoff matrix of stable game is chosen to be
$A = [[2,3],[4,6]]$. Note that the perturbations do not satisfy \eqref{bap}. The experimental results are shown in Figure \eqref{ec2}. We can see all these curves converge to $0$, thus support the result in Theorem \ref{thm:EG-Pertub}. Furthermore, we can observe that large perturbations may lead to more oscillations during the convergence processes, which in turn slows down the rate of convergence. We present more experiments in Appendix \ref{meme}.

\begin{figure}[h]
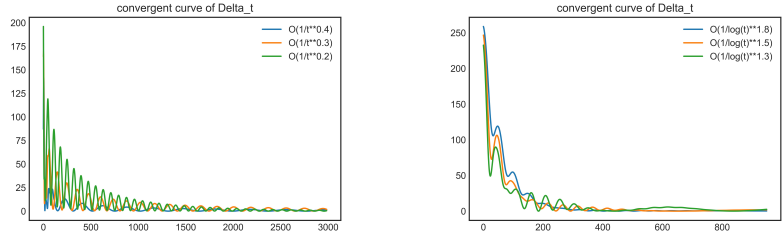

\centering
\subfigure
{
    \begin{minipage}[b]{.4\linewidth}
        \centering
        \includegraphics[scale=0.35]{thmextconvergent1.pdf}
    \end{minipage}
}
\subfigure
{
    \begin{minipage}[b]{.4\linewidth}
        \centering
        \includegraphics[scale=0.35]{thmextconvergent2.pdf}
    \end{minipage}
}
\caption{Values of $\Delta_t$ for extra-gradient with $A+B_{i,t}, i = 1,2,3$ (left), and $A_2+B_{i,t}, i =4,5,6$ (right).}
\label{ec2}
\end{figure}

%\textbf{Effects of the length of periods.}

%% file: appendix_1.tex
%\section{Appendix}
%\section{Property for Iterative Matrix}
\section{Step sizes and eigenvalues of the iterative matrix}

The eigenvalues of the iterative matrices in the linear differences systems in    \eqref{eq:OGDA} \eqref{eq:EG} and \eqref{eq:Negative Momentum} play a crucial role in analyzing the dynamic behavior of learning dynamics. In this section, we study how the choice of step sizes in learning dynamics affects the eigenvalues of these iterative matrices.

We firstly present the following corollary of Schur’s theorem, which was also used in \cite{zhang2019convergence} to demonstrate the convergence of learning dynamics in time-independent games.
\begin{lem}\label{schur stable}
(Corollary 2.1 in \cite{zhang2019convergence}). The roots of a real quartic polynomial $\lambda^4+a \lambda^3+b \lambda^2+c \lambda+d$ are within the (open) unit disk of the complex plane if and only if
$|c-a d|<1-d^2$, $|a+c|<b+d+1$ and $b<(1+d)+(c-a d)(a-c) /(d-1)^2$.  
\end{lem}
\begin{lem}\label{lm:eig-static}
Let $\sigma$ be the maximum modulus of the singular value of payoff matrix $A$. Then if for extra-gradient method with step size $\alpha = \gamma < \frac{1}{2 \sigma}$, optimistic gradient descent ascent with step size $\eta < \frac{1}{2\sigma}$, and negative momentum method with step size $\eta < \frac{1}{\sigma}$ and momentum parameters $\beta_1 = -\frac{1}{2}$ and $\beta_2 = 0$, then for the iterative matrices $\mathcal{A}$ in \eqref{eq:OGDA} \eqref{eq:EG} and \eqref{eq:Negative Momentum}, we have the following conclusion:
\begin{itemize}
    \item If payoff matrix $A$ is non-singular, then the modulus of eigenvalues of these iterative matrices $\mathcal{A}$  are strictly less than $1$. 
    \item  If payoff matrix $A$ is singular, then $1$ is an eigenvalue of the iterative matrix $\mathcal{A}$,  and other eigenvalues of $\mathcal{A}$ have  modulus strictly less than $1$.
\end{itemize}
\end{lem}
\begin{proof}

\textbf{OGDA.}
 We first write the characteristic polynomials of the iterative matrix $\mathcal{A}$ in \eqref{eq:OGDA} when payoff matrix is equal to $A$. Recall in this case, we have
\begin{align}\label{STATIC-OGDA}
\mathcal{A} = 
\begin{bmatrix}
 I & -2\eta A^{\top} & 0 & \eta A^{\top} \\
\\
2\eta A & I & -\eta A & 0\\
\\
I & 0& 0&0 \\
\\
0& I& 0& 0
\end{bmatrix}.
\end{align}
The characteristic polynomial equations are:
\begin{align}\label{eigpolyOGDA}
\lambda^2 (\lambda-1)^2 + \eta^2 \sigma_i^2 (1 - 2\lambda)^2 =0, i \in [m]
\end{align}

where $\sigma_i$ is a singular value of $A$. 
% Denote $a_i = \eta^2 \sigma_i^2 $. For $\eta \sigma\le 1/2$, we can obtain that $0\le a_i\le 1/4$. Above polynomial can be written as
% \begin{align}\label{poly}
% \lambda^4 - 2 \lambda^3 + (4a_i +1)\lambda^2 - 4a_i \lambda + a_i = 0.
% \end{align}
% Note that \eqref{poly} is a quartic polynomial equation, it has $4$ roots and if the complex number $\lambda$ is a root of \eqref{poly}, then its conjugate $\bar{\lambda}$ is also a root. 
And then according to Lemma \ref{schur stable}, it is easy to verify if $0< \eta \sigma < \frac{1}{2}$, then the norm of roots of the above polynomial is always less than 1.
When $\sigma_i=0$, we have the eigenvalues of $\mathcal{A}$ come from \eqref{eigpolyOGDA} are equal to 1. 

In all, if the payoff matrix $A$ is non-singular, we have the modulus of eigenvalue of $\mathcal{A}$ is strictly smaller than 1. And if there exists some singular value of $A$ equals to 0, we can obtain that if $\sigma_i=0$, then $\mathcal{A}$ has eigenvalue equal to 1,  otherwise, $\mathcal{A}$ only has eigenvalues whose norm is less than 1.

\textbf{EG.}
 We first write characteristic polynomial of iterative matrix $\mathcal{A}$ in \eqref{eq:EG}, with payoff matrix equals to $A$. We have
\begin{align*}
    \mathcal{A}
=
\begin{bmatrix}
 I - \alpha\gamma A A^{\top}  & -\alpha A\\
\\
 \alpha A^{\top} &  I - \gamma \alpha A^{\top}A
\end{bmatrix}.
\end{align*}
The characteristic polynomial equations are:
\begin{align*}
    (\lambda-1)^2+2 \gamma \alpha \sigma_i^2 (\lambda -1) +\alpha^2\sigma_i^2 + \alpha^2 \gamma^2 \sigma_i^4 = 0, i \in [m]
\end{align*}
where $\sigma_i$ is a singular value of $A$. 
% According to corollary of Schur's Theorem \ref{schur stable}, then for any singular value $\sigma$, $\gamma$ and $\alpha$ satisfy:
% The Schur condition is 
And then by Lemma \ref{schur stable}, the norm of roots of the above polynomial is always less than 1 if the following holds for all $i\in[m]$,
\begin{align}\label{eq:discriminant-of-EG}
   \alpha^2\sigma_i^2+(\alpha\gamma \sigma_i^2-1)^2 <1.
\end{align}
It is easy to verify that $\alpha = \gamma < \frac{1}{2 \sigma}$ satisfies the above inequalities.
Then we can use similar analysis in the part of OGDA to prove the conclusion for EG.

\textbf{Negative Momentum Method.}
First we write characteristic polynomial of iterative matrix $\mathcal{A}$ defined in \eqref{eq:Negative Momentum} when payoff payoff equals to $A$, we have
\begin{align}\label{STATIC-NM}
\mathcal{A} = 
\begin{bmatrix}
(1+\beta_1) I& -\eta A & -\beta_1 I & 0 \\
\\
\eta(1+\beta_1) A^{\top} & (1+\beta_2) I-\eta^2 A^{\top}A  & -\eta\beta_1 A^{\top} & -\beta_2 I\\
\\
I & 0& 0&0 \\
\\
0& I& 0&0 
\end{bmatrix}.     
\end{align}

The characteristic polynomial equations are:
\begin{align*}
    (\lambda -1)^2(\lambda-\beta_1)(\lambda-\beta_2)+\eta^2\sigma_i^2\lambda^3=0, i \in[m]
\end{align*}
when $\beta_1 = -\frac{1}{2}$, $\beta_2 = 0$ and $\eta < \frac{1}{\sigma}$ satisfies conditions in Lemma \ref{schur stable}. We can also use a similar analysis as in the OGDA part to prove the conclusion for negative momentum method.
\end{proof}

\section{Omitted Proofs from Theorem~\ref{thm:EG-Period}}\label{egpp}
\EGPeriod*

In this section, we prove Theorem \ref{thm:EG-Period}. In the following, we use $\tilde{\mathcal{A}}$ to denote matrix $\prod_{i=1}^T \mathcal{A}_i$. As shown by the Floquet theorem, the asymptotic behavior of a periodic linear system is determined by the product of iterative matrices over one period. Therefore, the analysis can be reduced to that of an autonomous system. We analyze the Jordan normal form of the product matrix for extra-gradient. We prove that the product matrix has no eigenvalues with a modulus larger than $1$. Moreover, the Jordan blocks of $1$ as an eigenvalue of the product matrix have size equals to $1$. These facts are enough to show the exponentially convergent behavior of extra-gradient. Before going through details of the proof, we provide a road map for the proof in Figure \ref{rb}.

\begin{figure}[h]
\flushleft
\subfigure
{
    \begin{minipage}[b]{.8\linewidth}
        \centering
        \includegraphics[scale=0.3]{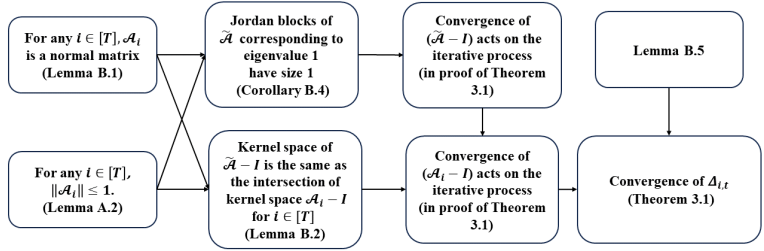}
    \end{minipage}
}
\caption{Road map for the proof of Theorem \ref{thm:EG-Period}}
\label{rb}
\end{figure}

Recall that \ref{Extra Gradient} can be written in a single linear difference system as
\begin{align}\label{eq:ext3}
\begin{bmatrix}
x_{t+1}\\
\\
y_{t+1}
\end{bmatrix}
=
\begin{bmatrix}
 I - \alpha\gamma A_tA_t^{\top}  & -\alpha A_t\\
\\
\alpha A_t^{\top} & I - \gamma\alpha A_t^{\top}A_t
\end{bmatrix}
\begin{bmatrix}
x_{t}\\
\\
y_{t}
\end{bmatrix}.
\end{align}
Denote $\mathcal{A}_t$ the iterative matrix in \eqref{eq:ext3}. The following lemma tells us that  $\mathcal{A}_t$ is a normal matrix.   
\begin{lem}\label{lm:normalEG}
For any $i \in [T]$, $\mathcal{A}_i$ is a normal matrix.
\end{lem}

\begin{proof}We have
\begin{align*}
 \mathcal{A}_i \mathcal{A}_i ^{\top}  =  \mathcal{A}_i^{\top}   \mathcal{A}_i
= 
\begin{bmatrix}
(I - \alpha \gamma A_i A_i^{\top})^2 + \alpha^2 A_i A_i^{\top} & 0\\
\\
0 & (I - \alpha\gamma A_i^{\top} A_i)^2 + \alpha^2 A_i^{\top} A_i 
\end{bmatrix}.
\end{align*}
    
\end{proof}

Using above lemma, we can  present several useful lemmas to describe Jordan form of matrix $\tilde{\mathcal{A}}$.

\begin{lem} \label{lm:ker1}
If $\alpha = \gamma <\frac{1}{\sigma}$, then for any $i\in [T]$, $\lVert \mathcal{A}_i  \lVert_2 \le 1$, and $\ker(\tilde{\mathcal{A}} - I) = \cap^T_{i=1} \ker(\mathcal{A}_i - I)$. Moreover, denote  $$\lambda_* = \max \{ \ \lvert \lambda \lvert \ | \ \lambda  \textnormal{\ is an eigenvalue of \ } 
\tilde{\mathcal{A}}, \lambda \ne 1   \},$$ then we have $\lambda_*<1$. 
\end{lem}
\begin{proof}
$(\Leftarrow)$ : If $v \in \cap^T_{i=1} \ker(\mathcal{A}_i - I)$, then for any $i \in \{1,...,T \}$, $\mathcal{A}_iv = v$, thus 
\begin{align*}
   \tilde{\mathcal{A}}v = \mathcal{A}_T \mathcal{A}_{T-1} ... \mathcal{A}_1 v = v.
\end{align*}
Then we have  $\cap^T_{i=1} \ker(\mathcal{A}_i - I)\subseteq \ker(\tilde{\mathcal{A}} - I) $.

$(\Rightarrow)$ : 
Let $v \in \ker(  \tilde{\mathcal{A}} - I) $, then we have $\lVert  v \lVert_2  = \lVert \mathcal{A}_T  ... \mathcal{A}_1 v \lVert_2$.
Denote $\lVert \cdot \lVert_2$ as 2-norm of matrices and vectors. 
% Since $\{\mathcal{A}_i \}_i$ are normal matrices, we have
% \begin{align*}
% \lVert \mathcal{A}_i  \lVert_2 = \max\{\ \left\vert \lambda \right\vert \ |\  \lambda \mbox{\ is\ an \ eigenvalue\ of\  } \mathcal{A}_i \}.
% \end{align*}
According to Lemma \ref{lm:eig-static}, if $\alpha = \gamma <\frac{1}{\sigma}$, then the spectral radius $\rho({\mathcal{A}_i})$ of $ \mathcal{A}_i$ is no larger than 1. Combining with the fact that $\mathcal{A}_i$ is normal, we have $\lVert{\mathcal{A}_i}\lVert_2=\rho(\mathcal{A}_i)\le 1$ for $i\in [T]$. 
We claim that if $\lVert  v \lVert_2  = \lVert \mathcal{A}_T  ... \mathcal{A}_1 v \lVert_2$, then we have $A_iv=v$ for $i\in [T]$. We prove it by contradiction. Suppose the claim is not true. Let $s$ be the minimum $i$ such that $\mathcal{A}_i v\ne v$. Since $\mathcal{A}_s$ is normal and its eigenvalues whose modulus equal to 1 can only be 1, we have  $\lVert \mathcal{A}_s v \lVert_2 < \lVert v \lVert_2$, then there holds
% \textcolor{red}{need to prove $A_s$ doesn't have eignevalue's module equal to 1}
\begin{align*}
\lVert  v \lVert_2 & = \lVert \mathcal{A}_T ... \mathcal{A}_s ... \mathcal{A}_1 v \lVert_2\\
& = \lVert \mathcal{A}_T ... \mathcal{A}_s v \lVert_2 \\
& <  \lVert \mathcal{A}_T ... \mathcal{A}_{s +1} \lVert_2 \lVert v \lVert_2 \\
& \le \lVert v \lVert_2,
\end{align*}

which leads to a contradiction. Therefore, for any $i \in [T]$, we obtain that $\mathcal{A}_i v= v$, i.e., $v\in \ker(\mathcal{A}_i-I)$. 
From the claim, we know that if $v \in \ker(  \tilde{\mathcal{A}} - I) $ , then $v\in \ker(\mathcal{A}_i-I)$ for $i\in [T]$. Thus we have $\ker(\tilde{\mathcal{A}} - I) \subseteq \cap^T_{i=1} \ker(\mathcal{A}_i - I)$.

Next we prove that $\lambda_*\le 1$. By the definition of $\lambda_*$, we obtain 
\begin{align*}  \lambda_*&\le\rho(\mathcal{A}_T\cdots\mathcal{A}_1)\\
    &\le \lVert \mathcal{A}_T  \cdots \mathcal{A}_1\lVert_2\\&\le \lVert \mathcal{A}_T \lVert_2 \cdots \lVert \mathcal{A}_1\lVert_2\le 1,
\end{align*}
where the second inequality holds because the spectral radius $\rho(A)\le \lVert A\lVert_2$ for any matrix $A$. 
% Now we prove that $\lambda_*<1$. Suppose $\lambda'$ is an eigenvalue of $\tilde{\mathcal{A}}$, and $v$ is the eigenvector of $\tilde{\mathcal{A}}$ corresponding to $\lambda'$, that is $\tilde{\mathcal{A}}v=\lambda'v$. Combining with $\lVert \mathcal{A}_i \lVert_2=1$ for $i\in [T]$, we have 
% \begin{align*}
%     \lVert \lambda'v \lVert_2&=\lVert \mathcal{A}_T ... \mathcal{A}_1 v \lVert_2\\
%     &\le \lVert \mathcal{A}_T \lVert_2 \lVert\mathcal{A}_{T-1}... \mathcal{A}_1 v \lVert_2\\
%     &= \lVert\mathcal{A}_{T-1}... \mathcal{A}_1 v \lVert_2\\
%     &\cdots\\
%     &\le\lVert v \lVert_2.
% \end{align*}
% From the above inequalies we obtain
% \begin{align*}
%     &\lVert \lambda'v\lVert_2 \le \lVert v\lVert_2\\
%     \Rightarrow &\lvert \lambda'\lvert \lVert v\lVert_2 \le \lVert v\lVert_2 \\
%     \Rightarrow &\lvert\lambda'\lvert \le 1.
% \end{align*}
% Combining with the definition of $\lambda_*$, we know that $\lambda_*\le 1 $.

Now we prove that $\lambda_*\neq 1$, which means that $\tilde{\mathcal{A}}$  have no eigenvalue $\lambda$ satisfying $\lambda \ne 1$ and $\lvert \lambda\lvert=1$. 
% It is equivalent to prove that if $\tilde{\mathcal{A}}v=\lambda'v$ and $\lVert\lambda'\lVert_2=1$, then $\lambda'=1$. 
Assuming $v$ is the eigenvector of $\tilde{\mathcal{A}}$ corresponding to $\lambda'$, where $\lvert \lambda' \lvert=1$, we can obtain $ \lVert \mathcal{A}_T ... \mathcal{A}_1 v \lVert_2=\lVert\lambda'v\lVert_2  =\lVert  v \lVert_2 $. Similar to the proof above, $\mathcal{A}_i v=v$ for $i\in [T]$, which implies that $\lambda'=1$. This completes the proof of $\lambda_*<1$.

%\textcolor{red}{In addition, we also prove that if $\tilde{\mathcal{A}}v=\lambda'v$ and $\lVert\lambda'\lVert_2=1$, then $\lambda'=1$. Because we can obtain $\lVert  v \lVert_2 & = \lVert \mathcal{A}_T ... \mathcal{A}_s ... \mathcal{A}_1 v \lVert_2$ from $\tilde{\mathcal{A}}v=\lambda'v$, then $\mathcal{A}_i v=v$ for $i\in [T]$ implies that $\lambda'=1$. We obtain that if the eigenvalue of $\tilde{\mathcal{A}}$ equals to 1, then the eigenvalue must be 1. Then we have $\lambda_*<1$.}
% Then we prove $\lambda_*<1$. Assuming there exists a eigenvalue 
% $\lambda'$ of $\Tilde{\mathcal{A}}$ such that $\lambda'\neq 1$ and $\lVert \lambda'\lVert_2=1$. Then there exists vector $x$ such that 
% \begin{align*}
%     \mathcal{A}_T ... \mathcal{A}_1 x=\lambda' x.
% \end{align*}
% Then we have 
\end{proof}

\begin{lem}\label{per1}
%The  Jordan blocks corresponding to eigenvalue $1$ in the real Jordan canonical form of $\mathcal{A}_1 \mathcal{A}_2 ... \mathcal{A}_T$ has size $1$, i.e., 
%real Jordan canonical form of 
Under a suitable orthogonal normal basis,  $\tilde{\mathcal{A}}$ has form
\begin{align}
\begin{bmatrix}
\bf{I_{r \times r}} & \bf{0} \\
\\
\bf{0} & \bf{C}
\end{bmatrix},
\end{align}
where $\bf{I_{r \times r}} \in \BR^{r \times r}$, $\bf{C} \in \BR^{(n+m-r) \times (n+m-r)}$, and $r = \Dim_{\BR} (\ker (\tilde{\mathcal{A}} - I)  )$.
\end{lem}

\begin{proof} Let $\{v_1, ..., v_r\}$ be an orthogonal normal basis of  $\ker (\tilde{\mathcal{A}} - I)$, i.e., $\langle v_i,v_j \rangle = 1$ if $i = j$, $\langle v_i,v_j \rangle=0$ if $i\ne j$.
First, we extend $\{v_1, ..., v_r\}$  to an orthonormal basis of $\BR^{n+m}$ and denote this basis by $\{v_1, ..., v_r, v_{r+1}, ..., v_{n+m}\}$. We also denote $M$ the matrix consisting of
$\{v_1, ..., v_r, v_{r+1}, ..., v_{n+m}\}$ as columns. With these settings, we have $M^{\top}M = M M^{\top} = I$. 

% By Lemma \ref{lm:ker1}, u
Under this basis, $\mathcal{A}_i$ is represented by matrix
\begin{align}\label{ima}
\begin{bmatrix}
\bf{I_{r \times r}} & \bf{0} \\
\\
\bf{C_{i,1}} & \bf{C_{i,2}}
\end{bmatrix}.
\end{align}

Moreover, as  $\mathcal{A}_i$ is a normal matrix, its representation under an orthogonal normal basis is still a normal matrix, thus we have

\begin{align}\label{asd}
\begin{bmatrix}
\bf{I_{r \times r}} & & \bf{0} \\
\\
\bf{C_{i,1} }& &\bf{ C_{i,2}}
\end{bmatrix}
\begin{bmatrix}
\bf{I_{r \times r}} & & \bf{C^{\top}_{i,1} }\\
\\
\bf{0} & &\bf{C^{\top}_{i,2} }
\end{bmatrix}
=
\begin{bmatrix}
\bf{I_{r \times r}} & & \bf{C^{\top}_{i,1} } \\
\\
\bf{0}& & \bf{C^{\top}_{i,2}}
\end{bmatrix}
\begin{bmatrix}
\bf{I_{r \times r}} & &\bf{0} \\
\\
\bf{C_{i,1} }& &\bf{C_{i,2} }
\end{bmatrix}.
\end{align}

Note that \eqref{asd} is equivalent to 
\begin{align*}
\begin{bmatrix}
\bf{I_{r \times r}}  & & \bf{C^{\top}_{i,1} }\\
\\
\bf{C_{i,1} }& & \bf{C_{i,1}C^{\top}_{i,1} } + \bf{ C_{i,2} C^{\top}_{i,2} }
\end{bmatrix}
=
\begin{bmatrix}
\bf{I_{r \times r}}  + \bf{C^{\top}_{i,1}C_{i,1} }& & \bf{C^{\top}_{i,1}C_{i,2} }\\
\\
\bf{C^{\top}_{i,2}C_{i,1} }& & \bf{C^{\top}_{i,2}C_{i,2} }
\end{bmatrix}.
\end{align*}
As a consequence, we have $\bf{C^{\top}_{i,1} } \bf{C_{i,1}} = 0$, and furthermore, this implies $\bf{C_{i,1} }= 0$. Thus \eqref{ima} has form

\begin{align*}
\begin{bmatrix}
\bf{I_{r \times r}} & \bf{0} \\
\\
\bf{0} & \bf{C_{i,2} }
\end{bmatrix},
\end{align*}
and under this basis, $\tilde{\mathcal{A}}$ can be represented by 
\begin{align*}%\label{ima2}
\begin{bmatrix}
\bf{I_{r \times r}} & & \bf{0} \\
\\
\bf{0} & & \prod^T_{i=1} \bf{C_{i,2}}
\end{bmatrix}.
\end{align*}
Since $\prod^T_{i=1} \bf{C_{i,2}}$ is a matrix with size of $(n+m-r)\times(n+m-r)$, we complete the proof.
\end{proof}

\begin{cor}\label{cor:Jordan block}
If $\lambda=1$ is an eigenvalue of  $\tilde{\mathcal{A}}$, then the Jordan blocks of $\tilde{\mathcal{A}}$ corresponding to eigenvalue $1$ has size $1$.
\end{cor}
\begin{proof} From Lemma \ref{per1}, we have a decomposition  $\BR^{m+n} = \ker( \tilde{\mathcal{A}}- I ) \oplus V'$, and both these two spaces are invariant under the action of $\tilde{\mathcal{A}}$. Thus we can choose a basis of $V'$ consisting of Jordan chains and denote this basis by $\{w_1,...,w_{m+n-r}\}$, then under the basis $\{w_1,...,w_{m+n-r}\} \cup \{v_1,...,v_r\}$, $\tilde{\mathcal{A}}$ is a block diagonal matrix. Moreover, there is no eigenvectors corresponding to eigenvalue $1$ in $\{w_1,...,w_{m+n-r}\}$, because any $w_i$ is linearly independent with $\{v_1,...,v_r\}$ (since they are basis), thus if some $w_i$ is an eigenvector of eigenvalue $1$, then a contradiction is conducted since it is assumed that $\Dim_{\BR}( \ker( \tilde{\mathcal{A}} - I) ) = r$.
\end{proof}

\begin{lem}\label{lm:EG-convergence equiavalent}
    Denote $X_t = (x_t,y_t)$ be the strategies of players at round of $t$ when they are playing extra-gradient. For any $i \in [T]$, if $\lVert (\mathcal{A}_i-I) X_t \lVert_2 $ converges to 0 with rate $\CO \left((\lambda_*)^{t/T} \cdot \textnormal{Poly}(t) \right)$, then $\Delta_{i,t} $ converges to 0 with rate $\CO \left((\lambda_*)^{t/T} \cdot \textnormal{Poly}(t) \right)$.
\end{lem}
\begin{proof}
Writing $(\mathcal{A}_i-I) X_{t}$ in a matrix form:
\begin{align*}
\begin{bmatrix}
-\alpha \gamma  A_i A_i^{\top} & -\alpha A_i  \\
\\
\alpha  A_i ^{\top} & -\alpha \gamma A_i^{\top} A_i 
\end{bmatrix}
\begin{bmatrix}
 x_{t-1} \\
 \\
 y_{t-1} \\
\end{bmatrix}
= \begin{bmatrix}
-\alpha \gamma A_i A_i^{\top}  x_{t-1} - \alpha A_iy_{t-1}\\
\\
\alpha  A_i ^{\top} x_{t-1} - \alpha \gamma A_i^{\top} A_i  y_{t-1}
\end{bmatrix}.
\end{align*}
For the sake of readability, we denote $g(t)= (\lambda_*)^{t/T} \cdot \textnormal{Poly}(t)$. According to the assumption, there is a constant $c$ such that $\lVert(\mathcal{A}_i-I)X_{t}\lVert_2 \leq c g(t) $, then we have
\begin{align*}
&\lVert - \gamma A_i A_i^{\top} x_{t-1} -  A_i y_{t-1}\lVert_2  \leq \frac{c g(t)}{\alpha},\\
 &\lVert A_i ^{\top}  x_{t-1} - \gamma A_i^{\top} A_i  y_{t-1}\lVert_2 \leq \frac{c g(t)}{\alpha}. 
\end{align*}
Let $c_1 = \max\{ \lVert A_i \lVert_2, i \in [T]\}$. Using these two inequalities to bound
 $\lVert A_i^{\top} x_{t}\lVert_2$, we have
 \begin{align*}
  &\lVert (\gamma^2  A_i^{\top}A_i+I)A_i^{\top} x_{t-1}\lVert_2 \\
  \\
  =&\lVert \gamma^2  A_i^{\top} A_i A_i^{\top} x_{t-1} + A_i^{\top}x_{t-1}\lVert_2 \\
  \\
  =&\lVert  A_i^{\top} x_{t-1} -  \gamma A_i^{\top} A_i y_{t-1} - \gamma A_i^{\top} (- \gamma A_i A_i^{\top} x_{t-1} -  A_i y_{t-1})\lVert_2\\
  \\
  \leq & \lVert  A_i^{\top} x_{t-1} -\gamma A_i^{\top} A_i y_{t-1} \lVert_2 + 
  \gamma \lVert A_i^{\top} \lVert_2 \lVert  - \gamma A_i A_i^{\top} x_{t-1} -  A_i y_{t-1} \lVert_2\\
  \\
  \leq & \frac{c (1+\gamma c_1) g(t)}{\alpha}.
 \end{align*}

Since matrix $\gamma^2  A_i^{\top}A_i+I$ is invertible, then
\begin{align*}
\lVert A_i x_{t}\lVert_2
=&(\gamma^2  A_i^{\top}A_i+I)^{-1}(\gamma^2  A_i^{\top}A_i+I)A_i^{\top} x_{t-1}\lVert_2 \\
\\
\leq& \lVert (\gamma^2  A_i^{\top}A_i+I)^{-1} \lVert_2 
\lVert (\gamma^2  A_i^{\top}A_i+I)A_i^{\top} x_{t-1}\lVert_2\\
\\
\leq &\lVert (\gamma^2  A_i^{\top}A_i+I)A_i^{\top} x_{t-1}\lVert_2\\
\\
\leq& \frac{c (1+\gamma) g(t)}{\alpha},
\end{align*}
where the last inequality is due to $\lVert (\gamma^2  A_i^{\top}A_i+I)^{-1} \lVert_2 \leq 1$.
 Similarly, we can obtain $$\lVert A_i y_{t}\lVert_2 \leq \frac{c (1+\gamma c_1) g(t)}{\alpha}.$$ Thus by definition of $\Delta_{i,t}=\lVert A_i^{\top} x_{t}\lVert_2 +\lVert A_i y_{t}\lVert_2$, $\Delta_{i,t}$ converges to 0 with rate $\CO \left((\lambda_*)^{t/T} \cdot \textnormal{Poly}(t) \right)$.
\end{proof}

Now we are ready to prove Theorem \ref{thm:EG-Period}.

\begin{proof}[Proof of Theorem \ref{thm:EG-Period}.]
We have proved $\lambda_*<1$ in Lemma \ref{lm:ker1}, now we prove the part of convergence rate.
Note that here we cannot directly apply the Floquet theorem in Proposition \ref{Floquet}, as it requires all iterative matrices within a period to be invertible. However, the proof
here follows the same idea as  the Floquet theorem : the convergence behavior of a periodic linear difference system is determined by the product of all iterative matrices of the system in a period.
% According to Lemma \ref{lm:eig-static}, when $A$ is singular, 
According to Corollary \ref{cor:Jordan block}, we can write Jordan form $J$ of $\tilde{\mathcal{A}}$ in the following way:
\begin{align*}
    J=
    \begin{bmatrix}
    I & 0\\
    0  & \tilde{J}
    \end{bmatrix},
\end{align*}
where $\tilde{J}$ consists of Jordan blocks corresponding to eigenvalues whose modulus not equal to 1. According to Lemma \ref{lm:ker1}, we have that the modulus of  eigenvalues of $\tilde{J}$ are less than 1. Moreover, we assume 
\begin{align*}
    J = P^{-1} \tilde{A}P.
\end{align*}

Denote $J_k(\lambda)$ as a Jordan block corresponding to eigenvalue $\lambda$ with size $k$, and $\lvert \lambda \lvert <1$. We can write $J_k(\lambda) = \lambda I+ N$, where $N$ represents the nilpotent matrix whose superdiagonal contains $1$'s and all other entries are zero. Moreover, we have $N^k=0$ and $\lVert N \lVert_2 = 1$. 

For each Jordan block $J_k(\lambda)$, without loss of generality, when $s > 2k$, by the binomial theorem:
\begin{align*}
  J^s_k(\lambda)  =   (\lambda I +N)^s
  = \sum_{r=0}^s\left(\begin{array}{l}
s \\
r
\end{array}\right) \lambda^{s-r} N^r.
\end{align*}
Then 
$$
\lVert J^s_k(\lambda) \lVert_2 \leq
(k-1) \left(\begin{array}{l}
s \\
k-1
\end{array}\right)
\lvert \lambda \lvert^{s-k+1},
$$
since $\lVert N \lVert_2 = 1$ and $s >2k$.
We know that $\left( \begin{array}{l}
s \\
k-1
\end{array}\right)$ is a polynomial of $s$ with degree $k \leq n+m$. Since $\lvert \lambda \lvert < 1$, $\lVert J^s_k(\lambda) \lVert_2$ goes to zero in rate $\CO \left((\lambda_*)^s \cdot \textnormal{Poly}(s) \right)$.
Since $J^s_k(\lambda)$ are blocks in block diagnol matrix $\tilde{J}^s$, 
then 
\begin{align*}
    \lVert \tilde{J}^s \lVert_2 \leq \sum_{ \lambda \in \text{Eigenvalue}\left(\tilde{\mathcal{A}}\right), \lambda \neq 1 }\lVert J^s_k(\lambda) \lVert_2,
\end{align*}

and $\lVert \tilde{J}^s \lVert_2$ goes to zero in rate $\CO \left((\lambda_*)^{s} \cdot \textnormal{Poly}(s) \right)$.
For any $t$, without loss of generality, we assume that $t = s T + j$, and $j\in [T]$ is the remainder. Then we have 
\begin{align*}
(\tilde{\mathcal{A}}-I) X_{t}
=&(\tilde{\mathcal{A}}-I)\tilde{\mathcal{A}^{s}} X_j\\
\\
=&(\tilde{\mathcal{A}}^{s+1}  - \tilde{\mathcal{A}}^{s}) X_j\\
\\
=& P^{-1} (J^{s+1} - J^{s}) P X_j\\
\\
=& P^{-1}\left(
\begin{bmatrix}
    I & 0\\
    0  & J_1^{s+1}
\end{bmatrix}
-
\begin{bmatrix}
    I & 0\\
    0  & J_1^{s}
\end{bmatrix}
\right) P X_j \\
\\
=&P^{-1}\left(
\begin{bmatrix}
    0 & 0\\
    0  & J_1^{s+1} - J_1^{s}
\end{bmatrix}
\right) P X_j.
\end{align*}
Taking norm on both sides, we have
\begin{align*}
    \lVert (\tilde{\mathcal{A}}-I) X_{t} \lVert_2
    \leq (\lVert J_1^{s+1} \lVert_2 + \lVert J_1^{s} \lVert_2) \lVert X_j \lVert_2
    \leq 2 \lVert J_1^{s} \lVert_2 \lVert X_j \lVert_2.
\end{align*}
From definition of $s$, we know that
$s = \lfloor t/T \rfloor \geq t/T-1 $, leading to  $(\lambda_*)^{s} \leq \frac{1}{\lambda_*}(\lambda_*)^{t/T} $. Since $\lVert J_1^{s} \lVert_2$ converges to zero with rate $\CO \left((\lambda_*)^{s} \cdot \textnormal{Poly}(s) \right)$, then $\lVert (\tilde{\mathcal{A}}-I) X_{t }\lVert_2$ converges to zero with rate  $\CO \left((\lambda_*)^{t/T} \cdot \textnormal{Poly}(t) \right)$. By Lemma \ref{lm:ker1}, for any $i \in [T]$, $\lVert (\mathcal{A}_i-I) X_{t}\lVert_2$ goes to zero in rate $\CO \left((\lambda_*)^{t/T} \cdot \textnormal{Poly}(t) \right)$.

According to Lemma \ref{lm:EG-convergence equiavalent}, we conclude for any $i \in [T]$, $\Delta_{i,t}$ goes to zero with convergence rate $\CO \left((\lambda_*)^{t/T} \cdot \textnormal{Poly}(t) \right)$, this completes the proof.
\end{proof}

\section{Omitted Proofs from Theorem \ref{thm:OGDMM-Period}}\label{onpg}
\OGDMMPeriod*
%First, we prove that the largest modulus of eigenvalues of iterative matrix is larger than 1 when players use \eqref{Optimistic Gradient} and \eqref{Negative Momentum}.

\subsection{On initialization}

Before proving Theorem \ref{thm:OGDMM-Period}, we discuss a more detailed question : 

\centerline{\emph{Which initial points will make (OGDA) and (NM) diverge ? }}

In fact, it is obviously that not every initial point will make optimistic gradient descent ascent and negative momentum method diverge. For example, if the initial point is chosen to be
\begin{align}
    (x_0,y_0) \in (\ker A_t^{\top}, \ker A_t),
\end{align}
then these point will be not diverge because they are stationary points of the game dynamics.

In the proof of Theorem \ref{thm:OGDMM-Period} below, we explicitly construct  initial points that diverge exponentially fast under \eqref{Optimistic Gradient} or \eqref{Negative Momentum}. In Figure \ref{ccccc}, we present an example of an initial point that converges under \eqref{Negative Momentum} with the game defined by \eqref{p2mg}. In fact, we can see that these converge initial points of \eqref{Optimistic Gradient} or \eqref{Negative Momentum} lie on a low dimension space, thus have measure zero. Note that this doesn't conflict with Theorem \ref{thm:OGDMM-Period}, since we are \textbf{not} claiming that optimistic gradient descent ascent or negative momentum method will make every initial point diverge. 

\begin{figure}[h]
\centering
\subfigure
{
    \begin{minipage}[b]{.6\linewidth}
        \centering
        \includegraphics[scale=0.5]{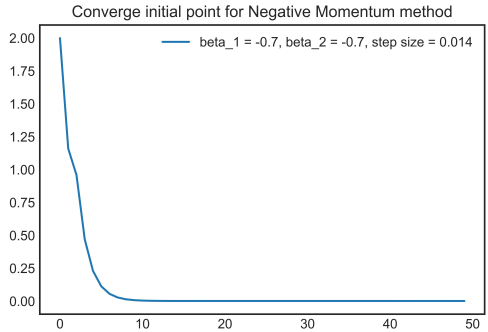}
    \end{minipage}
}
\caption{A converge initial point for negative momentum method, with initial condition $x_0 = x_{-1} = 0$, and 
$y_0 = (-0.4,1), \ y_{-1} = (1,-1)$. The curve is $\Delta_{1,t}$.}
\label{ccccc}
\end{figure}

In the following sections \ref{subnm} and  \ref{subogda}, we will prove that both negative momentum method and OGDA diverge with an exponential rate under certain initial conditions. The proof idea is the same for these two learning dynamics : firstly, we prove that the product of iterative matrices in a period for these learning dynamics have an eigenvalue with modulus larger than $1$; then, we show 
that eigenvectors corresponding to this eigenvalue as initial condition will diverge under the learning dynamics.

\subsection{Negative Momentum Method}\label{subnm}
We first consider negative momentum method with step size $\eta$, recall it can be written as:
\begin{align*}
& x_{t+1} = x_{t} -  \eta A_{t} y_{t} + \beta_1(x_{t} - x_{t-1}) , \\
& y_{t+1} = y_{t} +  \eta A_{t+1} ^{\top}x_{t+1}+ \beta_2(y_{t} - y_{t-1}),
\end{align*}
where $\beta_1,\beta_2 \le 0$ are the momentum parameters. 
Writing negative momentum method in matrix form, we have
\begin{align}\label{eq:Negative Momentum2}
\begin{bmatrix}
x_{t+1} \\
\\
y_{t+1} \\
\\
x_{t} \\
\\
y_{t}
\end{bmatrix} 
=
\begin{bmatrix}
(1+\beta_1) I& -\eta A_{t} & -\beta_1 I & 0 \\
\\
\eta(1+\beta_1) A_{t+1}^{\top} & (1+\beta_2) I-\eta^2 A_{t+1}^{\top}A_{t}  & -\eta\beta_1 A_{t+1}^{\top} & -\beta_2 I\\
\\
I & 0& 0&0 \\
\\
0& I& 0&0 
\end{bmatrix} 
\begin{bmatrix}
x_{t} \\
\\
y_{t} \\
\\
x_{t-1} \\
\\
y_{t-1}
\end{bmatrix} 
\end{align}
Denote the iterative matrix in \eqref{eq:Negative Momentum2} as $\mathcal{A}_t$ and $X_t= (x_t^{\top},y_t^{\top},x_{t-1}^{\top},y_{t-1}^{\top})^{\top}$. 
%Since payoff matrix has period of $2$, then we only have to analyze matrix $\mathcal{A}_{t+1} \mathcal{A}_{t}$. 
Let $\tilde{\mathcal{A}}_{NM}=\mathcal{A}_{t+1} \mathcal{A}_{t}$, by Floquet Theorem, $\tilde{\mathcal{A}}_{NM}$ will determine the dynamical behaviors of negative momentum method.
%Denote $X_t= (x_t^{\top},y_t^{\top},x_{t-1}^{\top},y_{t-1}^{\top})^{\top}$. 
We have
$$
X_t=\tilde{\mathcal{A}}_{NM} X_{t-2}, \ \text{for any $t\geq 2$}.
$$
% According to Theorem \ref{Floquet}, the above linear dynamical systems converges iff spectral radius (the maximum modulus of eigenvalue) of matrix $\tilde{\mathcal{A}}_{NegM}$ is less than 1. A corollary of Schur stable theorem (\textcolor{red}{Schur stable theorem}):
% \begin{cor}\label{schur stable}
% (e.g. Mansour (2011)). The roots of a real quartic polynomial $\lambda^4+a \lambda^3+b \lambda^2+c \lambda+d$ are within the (open) unit disk of the complex plane iff $|c-a d|<1-d^2$, $|a+c|<b+d+1$ and $b<(1+d)+(c-a d)(a-c) /(d-1)^2$.  
% \end{cor}
In the following lemma, we show that the spectral radius of $\tilde{\mathcal{A}}_{NM}$ is always larger than $1$.
\begin{lem}\label{thm:NM-Period}
    For any step size $\eta > 0$, and momentum parameters $\beta_1,\ \beta_2 \le 0$, the spectral radius of $\tilde{\mathcal{A}}_{NM}$ is larger than $1$.
\end{lem}
\begin{proof}
 We directly compute the characteristic polynomial $P_{\tilde{\mathcal{A}}_{NM}}(\lambda)$ of matrix $\tilde{\mathcal{A}}_{NM}$ as follows \footnote{Symbolic computing software, such as Matlab, can be used to perform the computation of characteristic polynomial.} : 
\begin{align*}
 P_{\tilde{\mathcal{A}}_{NM}}(\lambda)& = \det( \lambda I - \tilde{\mathcal{A}}_{NM} )  \\
 \\
= &[\lambda^4-\big(4\eta^4+4\eta^2(\beta_2-\beta_1) + \beta_1^2 + \beta_2^2 + 2\big) \cdot \lambda^3 \\ 
\\
& +\big(4\eta^2(\beta_2 -\beta_1) + \beta_1^2\beta_2^2 +2\beta_1^2 + 2\beta_2^2 + 1\big) \cdot \lambda^2\\
\\
& -\big(2\beta_1^2\beta_2^2 + \beta_1^2 + \beta_2^2 \big) \cdot  \lambda +\beta_1^2\beta_2^2] \cdot (\lambda-1) \cdot (\lambda-\beta_2^2).
\end{align*}
Note that this is a polynomial on $\lambda$ of degree $6$, with two roots $\lambda = 1$ and
$\lambda = \beta^2_2$.

Thus, eigenvalues of matrix $\tilde{\mathcal{A}}_{NM}$ consists of $1$, $\beta_2^2$ and roots of the quartic polynomial 
\begin{align*}
    g(\lambda) = \lambda^4 + a\lambda^3 + b\lambda^2 + c\lambda + d,
\end{align*}
with coefficients 
    \begin{align*}
         &a=-(4\eta^4+4\eta^2(\beta_2 -\beta_1) + \beta_1^2 + \beta_2^2 + 2),\\
         \\
         &b = 4\eta^2(\beta_2 -\beta_1) + \beta_1^2\beta_2^2 +2\beta_1^2 + 2\beta_2^2 + 1,\\
         \\
         &c=-(2\beta_1^2\beta_2^2 + \beta_1^2 + \beta_2^2),\\
         \\
         &d=\beta_1^2\beta_2^2.
    \end{align*}

%$a=-(4\eta^4+4\eta^2(\beta_2 -\beta_1) + \beta_1^2 + \beta_2^2 + 2)$, $b=4\eta^2(\beta_2 -\beta_1) + \beta_1^2\beta_2^2 +2\beta_1^2 + 2\beta_2^2 + 1$, $c=-(2\beta_1^2\beta_2^2 + \beta_1^2 + \beta_2^2)$ and $d=\beta_1^2\beta_2^2$.
In order to prove the spectral radius of $\tilde{\mathcal{A}}_{NM}$ is larger than 1, we just need to verify that the maximal modulus of the roots of $g$ is larger than 1.
According to Lemma \ref{schur stable}, the polynomial $g(\lambda)$ has a root with modulus no less than $1$ if $|a+c| > b+d+1$. Next we want to prove that $|a+c| > b+d+1$ holds for any step size $\eta > 0$ and any momentum parameter $\beta_1$, $\beta_2$. Computing directly,
\begin{align*}
|a+c|-(b+d+1)
= \ & -(a+c)-(b+d+1)\\
= \ & 4\eta^4+4\eta^2(\beta_2 -\beta_1)+ \beta_1^2 + \beta_2^2 + 2 +2\beta_1^2\beta_2^2 + \beta_1^2 + \beta_2^2\\
&-(4\eta^2(\beta_2 -\beta_1) + \beta_1^2\beta_2^2 +2\beta_1^2 + 2\beta_2^2 + 1+\beta_1^2\beta_2^2+1)\\
= \ &4\eta^4 > 0,
\end{align*}
where the first equality holds since 
\begin{align*}
    a+c & = -(4\eta^4+4\eta^2(\beta_2 -\beta_1) + 2\beta_1^2 + 2\beta_2^2 + 2+2\beta_1^2\beta_2^2) \\
    & \le -(4\eta^4+4\eta^2(\beta_2 -\beta_1)+(\beta_1 -\beta_2)^2) \\
    & \le -(2\eta^2+\beta_2-\beta_1)^2 \\
    & \le 0.
\end{align*}
The inequality $|a+c| > b+d+1$ violates the second condition in Corollary \ref{schur stable}, which means the maximal modulus of the roots of $g$ is at least $1$. 

Next, we want to prove the maximal modulus of the roots of $g$ is strictly larger than $1$. 
Assuming that the maximal modulus of the roots of $g$ is equal to $1$, then for any given $r>1$, the roots of  
$$(r \lambda )^4+a (r \lambda )^3+b (r \lambda )^2+c (r \lambda )+d=0$$ 
are within the (open) unit disk of the complex plane (the roots $\lambda$ satisfy $\left\vert r\lambda \right\vert \le 1 $). Divide the quartic polynomial by $r^4$, we have $\lambda ^4+\frac{a}{r}\lambda ^3+\frac{b}{r^2} \lambda^2+\frac{c}{r^3}\lambda+\frac{d}{r ^4}=0$. 
%Let $r=1+\epsilon$,where $\epsilon = \frac{\eta^4}{15(|a|+|b|+|c|+|d|}$. 
By Corollary \ref{schur stable}, we have the second condition for the polynomial above, that is,
\[
\left\vert \frac{a}{r}+\frac{c}{r^3}  \right\vert< \frac{b}{r^2}+\frac{d}{r ^4}+1. 
\]
Notice that the inequality above holds for any $r>1$. Let $r \to 1^{+}$, we obtain,
\begin{align*}
    & \left\vert a+c \right\vert  \le b+d+1\\
    & (\Leftrightarrow) \\
    &\ |4\eta^4+4\eta^2(\beta_2 -\beta_1) + 2\beta_1^2 + 2\beta_2^2 + 2+2\beta_1^2\beta_2^2 | \\
    & \le 4\eta^2(\beta_2 -\beta_1) + 2\beta_1^2 + 2\beta_2^2 +2\beta_1^2\beta_2^2+2 \\
    & (\Leftrightarrow)\\
    &\ \eta^{4} \le 0.
\end{align*}
which contradicts with the step size $\eta >0$.
%the second condition in Corollary \ref{schur stable}, 
Therefore, our assumption that the maximal modulus of the roots of $g$ is equal to 1 cannot hold. 

% Combining with the previous result, 
% the spectral radius of $\mathcal{A}$ is at least 1
In conclusion, we have that the spectral radius of $\tilde{\mathcal{A}}_{NM}$ is strictly greater than $1$.
\end{proof}

Now we are ready to proof Theorem \ref{thm:OGDMM-Period} for the part of negative momentum method.

\begin{proof}[ proof of Theorem \ref{thm:OGDMM-Period}, \bf(part \uppercase\expandafter{\romannumeral1}, Negative Momentum)]
% It can be verified that $[0,1,1,0,1,1]^{\top}$
% and $[0,\beta_2,\beta_2,0,1,1]^{\top}$ are eigenvectors of eigenvalues $1$ and $\beta_2^2$ when $\beta_2 \neq -1$. When $\beta_2=-1$, $[0,1,1,0,1,1]^{\top}$ and $[0,0,0,0,1,1]^{\top}$ are eigenvectors of eigenvalues $1$.
% Since $\beta_2 < 0$, then for any $a$ and $b$, $[0,a,a,0,b,b]^{\top}$ which is the combination of $[0,1,1,0,1,1]^{\top}$ and $[0,\beta_2,\beta_2,0,1,1]^{\top}$ is not eigenvector of other eigenvalues of $\tilde{\mathcal{A}}_{NM}$.
As we have shown in Lemma \ref{thm:NM-Period}, $1$ and $\beta_2^2$ are two eigenvalues of $\tilde{\mathcal{A}}_{NM}$.
We claim that if $[0,a,a,0,b,b]^{\top} \in \BR^6$ is an eigenvector of $\tilde{\mathcal{A}}_{NM}$, with condition $a $ and $b$ not simultaneously equal to $0$, then it can only be an eigenvector corresponds to either $1$ or $\beta_2^2$. In the following,we prove the above claim.

Without loss of generality, we assume that $b \neq 0$ (the case $ a \neq 0$ is similar), moreover, we can assume that $b = 1$ by a normalization.
Then, we have
\begin{align*}
\tilde{\mathcal{A}}_{NM} \cdot
\begin{bmatrix}
 0
 \\ 
 a
 \\ 
 a 
 \\
 0 
 \\1 
 \\ 1   
\end{bmatrix}
= \begin{bmatrix}
 0 
 \\ 
 a(\beta_2^2+\beta_2+1)-\beta_2(\beta_2+1) \\ a(\beta_2^2+\beta_2+1)-\beta_2(\beta_2+1)
 \\
 0 
 \\
 a(\beta_2+1)-\beta_2 
 \\
 a(\beta_2+1)-\beta_2  
\end{bmatrix}
\end{align*}
Firstly, if $a=0$ and $[0,a,a,0,1,1]^{\top} = [0,0,0,0,1,1]^{\top}$  is an eigenvector of $\tilde{\mathcal{A}}_{NM}$, then either $\beta_2 = 0$ or $\beta_2 = -1$. If $\beta_2 = 0$, then $[0,0,0,0,1,1]^{\top}$  is eigenvector corresponding to eigenvalue $ \beta_2^2$; if $\beta_2 = -1$, then $[0,0,0,0,1,1]^{\top}$  is eigenvector corresponding to eigenvalue $1$.

Secondly, if $a \neq 0$ and $[0,a,a,0,1,1]^{\top}$ is eigenvector of $\tilde{\mathcal{A}}_{NM}$, 
then 
\begin{align*}
   & \dfrac{a(\beta_2^2+\beta_2+1)-\beta_2(\beta_2+1)}{a} 
    =\dfrac{a(\beta_2+1)-\beta_2}{1}
    \\
\Rightarrow 
& (a-\beta_2)(a-1)(\beta_2+1)=0.
\end{align*}
When $\beta_2 \neq -1$, then $a = \beta_2$ or $a=1$, and 
$[0,a,a,0,1,1]^{\top}$ is an eigenvector corresponding to eigenvalue $\beta_2^2$ or $1$. When $\beta_2 = -1$, $[0,1,1,0,1,1]^{\top}$ and $[0,0,0,0,1,1]^{\top}$ are eigenvectors of $\tilde{\mathcal{A}}_{NM}$ corresponding to eigenvalue $1$. Thus we conclude for any $a$ and $b$ not simultaneously equal to $0$, $[0,a,a,0,b,b]^{\top}$ can only be an eigenvector of $\tilde{\mathcal{A}}_{NM}$ corresponding to eigenvalue $1$ and $\beta^2_2$, this completes the proof of the claim. %since this vector can always 
%be represented by the linear combination of eigenvectors corresponding to eigenvalue $1$ and $\beta_2^2$.

Next we construct an initial condition that has exponential divergence rate under negative momentum method. Let $\lambda'$ be the eigenvalue
of $\tilde{\mathcal{A}}_{NM}$  with largest modulus except $\beta_2^2$,
 then by Lemma  \ref{thm:NM-Period},  $\lvert \lambda' \lvert > 1$. 
 We also denote $X_0=[x_0, y_{0,1}, y_{0,2}, x_{-1}, y_{-1,1}, y_{-1,2}]^{\top} \in \BR^6$ as the corresponding eigenvector of $\lambda'$. Here $y_i = (y_{i,1},y_{i,2})$, and $x_i$ for $i=0,-1$ are initial conditions. Then from the claim proved above, one of  $x_0$, $x_{-1}$, $y_{0,1}-y_{0,2}$ and  $y_{-1,1}-y_{-1,2}$ not equals to $0$. Let $c = \max\{\lvert x_0 \lvert, \lvert x_{-1} \lvert, \lvert y_{0,1}-y_{0,2} \lvert,\lvert y_{-1,1}-y_{-1,2} \lvert \}$, then $c > 0$.
 
 In the following, we construct the initial point by  considering two cases :
$\lambda'$ is a real number or complex number.

%We do this by considering two cases :  the eigenvalue of $\tilde{\mathcal{A}}_{NM}$ with largest modulus except $\beta_2^2$ is a real number or complex number. We denote this eigenvalue as $\lambda'$, by Lemma  \ref{thm:NM-Period},  $\lvert \lambda' \lvert > 1$. 

%We also denote $X_0=[x_0, y_{0,1}, y_{0,2}, x_{-1}, y_{-1,1}, y_{-1,2}]^{\top}$ as the corresponding eigenvector of $\lambda'$. Here $y_i = (y_{i,1},y_{i,2})$, and $x_i$ for $i=0,-1$ are initial conditions.
%Then one of  $x_0$, $x_{-1}$, $y_{0,1}-y_{0,2}$ and  $y_{-1,1}-y_{-1,2}$ is not $0$. Let $c = \max\{\lvert x_0 \lvert, \lvert x_{-1} \lvert, \lvert y_{0,1}-y_{0,2} \lvert,\lvert y_{-1,1}-y_{-1,2} \lvert \}$, then $c > 0$.

Firstly, we consider the case that $\lambda'$ is a real number.
We can write the iterative process using $\tilde{\mathcal{A}}_{NM}$ as follows:
\begin{align*}
X_{2t} = 
\tilde{\mathcal{A}}_{NM}^{t} X_0= (\lambda') ^{t} X_0.
\end{align*}
which implies 
\begin{align*}
    & x_{2t} =  (\lambda')^t  x_0, 
    y_{2t,1} = (\lambda')^t y_{0,1},
    y_{2t,2} =(\lambda')^t y_{0,2},\\
    & x_{2t-1} = (\lambda')^t x_{-1},
    y_{2t-1,1} =(\lambda')^t y_{-1,1},  
    y_{2t-1,2} = (\lambda')^{t} y_{-1,2}.\\
\end{align*}

Since $A_1= [1,-1]$, then 
\begin{align*}
A_1^{\top} x_t =& [x_t, -x_t]^{\top},
A_1 y_t = y_{t,1} - y_{t,2}, \\
\Rightarrow
 \Delta_{1,t} = &\lVert A_1^{\top} x_t \lVert_2 + \lVert A_1 y_t  \lVert_2\\
=& \sqrt{2} \lvert x_t \lvert + \lvert y_{t,1} - y_{t,2} \lvert \\
= & \left\{
\begin{array}{c}
\lvert \lambda' \lvert^{\frac{t}{2}} (\sqrt{2} \lvert x_0 \lvert + \lvert y_{0,1}-y_{0,2} \lvert), \  \text{if t is even}\\
\\
\lvert \lambda' \lvert^{\frac{t+1}{2}} (\sqrt{2} \lvert x_{-1} \lvert + \lvert y_{-1,1}-y_{-1,2} \lvert), \ \text{if t is odd}. 
\end{array}
\right.
\end{align*}

Then, $\max \{\Delta_{1,t-1}, \Delta_{1,t}\} \geq c \lvert \lambda' \lvert^{\frac{t}{2}}$. Let $\lambda = \lvert \lambda' \lvert^{\frac{1}{2}} $, we have
$\sup_{s \in [t]} \Delta_{1,s} \geq c \lambda^t \in \Omega \left( \lambda ^t \right)$. Similarly, then we have $\sup_{s \in [t]} \Delta_{2,s} \in \Omega \left(\lambda^t \right)$.

Secondly, we consider $\lambda'$ as a complex number. Denote this eigenvalue by $a+bi$, then $a-bi$ is also an eigenvalue of $\tilde{\mathcal{A}}_{NM}$. Denote $v$ the eigenvector of eigenvalue $a+bi$, then  $\Bar{v}$ is the eigenvector of eigenvalue $a-bi$. Let $X_0 = v + \Bar{v}$. In the following, we prove $ X_0 \neq 0$ by contradiction. Assuming $X_0 = 0$ which means $v = v' i$, where $v'$ is a real vector. Then, $A v = A v ' i =(a+b i ) v' i = a v' i - b v'$. Since $A$ is a real matrix, then vector $A v ' i$ only consists of pure imaginary numbers, leading to $b=0$. Then the contradiction appears since 
$\lambda' = a+b i$ is a complex number. According to previous analysis, one of  $x_0$, $x_{-1}$, $y_{0,1}-y_{0,2}$ and  $y_{-1,1}-y_{-1,2}$ is not $0$.
Here we analyze the case when $ c = \lvert y_{0,1}-y_{0,2} \lvert $ is not equal to zero and omit other cases because these analyses are very similar. According to the iterative process, we have
\begin{align*}
X_{2t} =&\tilde{\mathcal{A}}_{NM}^t X_0 \\
=&\tilde{\mathcal{A}}_{NM}^t(v + \Bar{v})\\
=&(a+bi)^t v + (a-bi)^t \Bar{v}\\
=&e^{it\theta}(a^2+b^2)^{\frac{t}{2}} v +e^{-it\theta}(a^2+b^2)^{\frac{t}{2}}  \Bar{v},
\end{align*}
where $\theta=\textbf{sign}(b)\frac{\pi}{2}$ if $a=0$, otherwise $\theta=\arctan(\frac{b}{a})$. Since $A_1 = [1,-1]$, then,
\begin{align*}
\lVert A_1 y_t \lVert_2 = &\lvert y_{t,1}-y_{t,2} \lvert \\
=& \lvert 2c (e^{it\theta}+e^{-it\theta})(a^2+b^2)^{\frac{t}{2}} \lvert \\
=&  \lvert 4c \cdot \cos(t\theta) \lvert (a^2+b^2)^{\frac{t}{2}}.
\end{align*}
For $\cos(t\theta)$, either $\cos(t\theta)\equiv 1$ when $\theta = 0$,  or $\lim_{t\rightarrow +\infty}\cos(t\theta)$ doesn't exist, which means there exists a constant $\delta>0$ and $\{t_j\}_{j=1,2,\dots}$ , where $\{t_j\}_{j=1,2,\dots}$ is a sequence that goes to infinity, such that $\lvert \cos(t_j\theta)\lvert>\delta$.
We know that 
$ \lvert \lambda' \lvert = (a^2+b^2)^{\frac{1}{2}}$, then $\lvert \lambda' \lvert > 1$.
Let $\lambda = \lvert \lambda' \lvert ^{\frac{1}{2}}$, leading to $\lambda>1$. In addition to $c \neq 0$, we have $$\Delta_{1,t_j} \ge \lVert A_1 y_{t_j} \lVert_2 \ge \delta\lambda^t\in \Omega(\lambda^t).$$ 
Thus $\sup_{s\in [t]}\Delta_{1,s}\in \Omega(\lambda^t)$, and $\sup_{s\in [t]}\Delta_{2,s}\in \Omega(\lambda^t)$ can be proven in the same way.
\end{proof}

\subsection{Optimistic Gradient Descent Ascent}\label{subogda}
In this subsection, we consider optimistic gradient descent ascent with step size $\eta$. Recall that the linear difference form of OGDA can be written as following:
\begin{align}\label{OGDA3}
\begin{bmatrix}
x_t \\
\\
y_t \\
\\
x_{t-1} \\
\\
y_{t-1}
\end{bmatrix} 
=
\begin{bmatrix}
I & -2\eta A_{t-1} & 0 & \eta A_{t-2} \\
\\
2\eta A_{t-1}^{\top} & I & -\eta A_{t-2}^{\top} & 0\\
\\
I & 0& 0&0 \\
\\
0& I& 0&0 
\end{bmatrix} 
\begin{bmatrix}
x_{t-1} \\
\\
y_{t-1} \\
\\
x_{t-2} \\
\\
y_{t-2}
\end{bmatrix}.
\end{align}

We denote the matrix in \eqref{OGDA3} as $\mathcal{A}_t$ and let $X_t= (x_t^{\top},y_t^{\top},x_{t-1}^{\top},y_{t-1}^{\top})^{\top}$. Since payoff matrix has period of $2$, by Floquet Theorem, we only have to analyze matrix $\mathcal{A}_{t+1} \mathcal{A}_{t}$. Let $\tilde{\mathcal{A}}_{OGDA}=\mathcal{A}_{t+1} \mathcal{A}_{t}$, 
%Denote $X_t= (x_t^{\top},y_t^{\top},x_{t-1}^{\top},y_{t-1}^{\top})^{\top}$. 
then, we have
$$
X_t=\tilde{\mathcal{A}}_{OGDA} X_{t-2}, \ \text{for any $t\geq 2$}.
$$
\begin{lem}\label{thm:OGDA-Period}
    For any step size $\eta > 0$,  the spectral radius of $\tilde{\mathcal{A}}_{OGDA}$ is larger than $1$.
\end{lem}
\begin{proof}
Directly compute the characteristic polynomial $P_{\tilde{\mathcal{A}}_{OGDA}}(\lambda)$ of matrix $\tilde{\mathcal{A}}_{OGDA}$ gives
\begin{align*}
P_{\tilde{\mathcal{A}}_{OGDA}}(\lambda) = & \ \det( \lambda I - \tilde{\mathcal{A}}_{OGDA} ) \\
= & \ \lambda \cdot ( \lambda -1)\cdot \left(\lambda- \left(4\eta^2 - \frac{1}{2}\sqrt{64\eta^4 + 8\eta^2 + 1} + \frac{1}{2}\right)
\right)^2\\
& \cdot  \left(\lambda- \left(4\eta^2 + \frac{1}{2}\sqrt{64\eta^4 + 8\eta^2 + 1} + \frac{1}{2}\right)\right)^2.
\end{align*}
Then, $\tilde{\mathcal{A}}_{OGDA}$ has an eigenvalue $\lambda' = 4\eta^2 + \frac{1}{2}\sqrt{64\eta^4 + 8\eta^2 + 1} + \frac{1}{2}$.
It is easy to verify that $\lambda'$ is strictly monotonically increasing with $\eta \in [0, +\infty)$, and $\lambda'$ equals to $1$ iff $\eta = 0$. Since step size $\eta > 0$, thus
the spectral radius of matrix $\tilde{\mathcal{A}}_{OGDA}$ is larger than 1. 
\end{proof}

Now we are ready to prove  Theorem \ref{thm:OGDMM-Period} for the part of optimistic gradient descent ascent method.

\begin{proof}[proof of Theorem \ref{thm:OGDMM-Period}, \bf(part \uppercase\expandafter{\romannumeral2}, OGDA)]

Let  $X_0=[x_0, y_{0,1}, y_{0,2}, x_{-1}, y_{-1,1}, y_{-1,2}]^{\top}$ be the eigenvector corresponding to the eigenvalue $\lambda'$ defined above. Then, it is directly to verify $x_0, x_{-1} \ne 0$. 
%In fact, this eigenvector has form :

%\begin{align*}
%\left[
%\begin{array}{c}
%    \frac{\frac{1}{2}\sqrt{64\eta^4 + 8\eta^2 + 1} + \frac{1}{2}}{8\eta^2 + 1} \\
%    \\
%    \frac{\eta + 2\eta( \eta^2 + \frac{1}{2} \sqrt{64\eta^4 + 8\eta^2 + 1} + \frac{1}{2})}{8\eta^2 + 1} \\
%    \\
%    -\frac{\eta + 2\eta( \eta^2 + \frac{1}{2} \sqrt{64\eta^4 + 8\eta^2 + 1} + \frac{1}{2})}{8\eta^2 + 1} \\
%    \\
%    1 \\
%    0 \\
%    0
%\end{array}
%\right]
%\end{align*}

In the iterative process, we have
\begin{align*}
 X_{2t} = 
\tilde{\mathcal{A}}_{OGDA}^tX_0= (\lambda') ^t X_0, \ x_{2t} =  (\lambda')^t x_0.    
\end{align*}
Since $A_1= [1,-1]$, then 
\begin{align*}
A_1^{\top} x_t =& [x_t, -x_t]^{\top},
A_1 y_t = y_{t,1} - y_{t,2}. \\
\Rightarrow
 \Delta_{1,t} = &\lVert A_1^{\top} x_t \lVert_2 + \lVert A_1 y_t  \lVert_2\\
=& \sqrt{2} \lvert x_t \lvert + \lvert y_{t,1} - y_{t,2} \lvert \\
\geq & \sqrt{2} \lambda'^{\frac{t}{2}} \min \{x_0, x_{-1}\}.
\end{align*}
By Theorem  \ref{thm:OGDA-Period},  $ \lambda'  > 1$. 
Let $\lambda = (\lambda')^{\frac{1}{2}}$, then $\lambda > 1 $.
According to the inequality above, we have
$$
\sup_{s \in [t]} \Delta_{1,s}
\geq  \Delta_{1,t}
\geq \sqrt{2} \min \{x_0, x_{-1}\}\lambda^t \in \Omega \left( \lambda ^t \right).
$$
Similarly, we have $\sup_{s \in [t]} \Delta_{2,s} \in \Omega \left(\lambda^t \right)$.
\end{proof}

\section{Proof for convergent perturbed games with invertible payoff matrix}

In this section, we provide a proof of a special case of Theorem \ref{thm:Condition-Pertub} under the assumption that the payoff matrix is an invertible square matrix. Furthermore, we can demonstrate that this assumption leads to an exponential convergence rate.

\begin{restatable}{prop}{InvertiblePertub}
\label{thm:Invertible-Pertub}
When the payoff matrix $A$ of the stable game is an invertible square matrix and $\lim_{t \to \infty} B_t = 0$, we have $\lim_{t \to \infty} (x_t, y_t) = (\boldsymbol{0},\boldsymbol{0}) \in \BR^{2n}$ in \eqref{Optimistic Gradient},\ \eqref{Extra Gradient},\ and\ \eqref{Negative Momentum}  with an exponential rate.
\end{restatable}
% \InvertiblePertub*
 \begin{proof}
According to Perron Theorem  \ref{thm:Perron},  we only need to prove maximum modulus of eigenvalues of iterative matrix $\mathcal{A}$ is less than 1. Lemma \ref{lm:eig-static} indicates that if the parameter condition on step sizes is satisfied, we have  maximum modulus of eigenvalues of iterative matrix $\mathcal{A}$ is less than 1. This complete the proof.
\end{proof}

The above proof cannot be generalized to non-invertible matrices, as we have shown in Lemma \ref{lm:eig-static} that when the payoff matrix is non-invertible, then iterative matrices of the difference system associated with the game dynamics must have an eigenvalue equals to $1$. 

In the following, we prove Theorem \ref{thm:Condition-Pertub} for the general case.

\section{Omitted Proofs from Theorem \ref{thm:Condition-Pertub}}\label{appx:Condition-Pertub}
 \ConditionPertub*
% According to subsection \ref{diagonalization}, we can use singular value decomposition to write matrix $A = U \Sigma_A V^{\top}$. 
% Let $ \tilde{x}_t = U^{\top}x_t$, $\tilde{y}_t = V^{\top}y_t$ and $\tilde{X}_t=(\tilde{x}_t^{\top}, \tilde{y}_t^{\top}, \tilde{x}_{t-1}^{\top},\tilde{y}_{t-1}^{\top})^{\top} $, then we reformulate the iterative process as follows:
% \begin{align*}
%    \tilde{X}_{t+1}= (\tilde{\mathcal{A}}+\tilde{\mathcal{B}}_{t})\tilde{X}_{t},
% \end{align*}
% where in optimistic gradient descent ascent, $\tilde{\mathcal{A}}$ and ${\tilde{\mathcal{B}}}_t$ are \eqref{SVD-OGDA} and \eqref{B-OGDA}; in extra-gradient, $\tilde{\mathcal{A}}$ and $\tilde{\mathcal{B}}_t$ are \eqref{A-EG} and \eqref{B-EG}; in negative momentum method, $\tilde{\mathcal{A}}$ and $\tilde{\mathcal{B}}_t$ are \eqref{SVD-NM} and \eqref{B-NM}.

We separate the proof into several lemmas. Before going into details, we present a road map of the proof in Figure \eqref{road map}. %\textcolor{red}{(It should be constraints in this figure. label in this figure should be revised.)}
\begin{figure}[!htbp]
    \centering
    \includegraphics[scale=0.3]{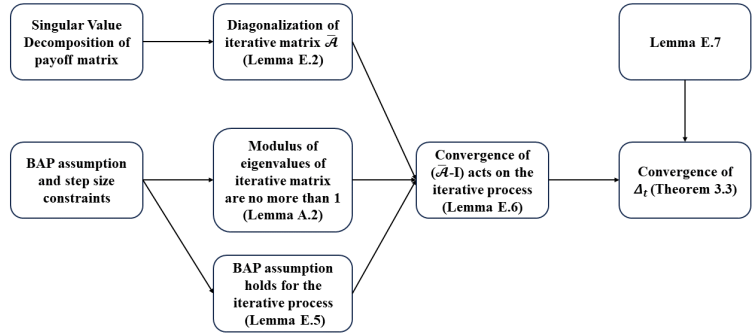}
    \caption{Road map for the prove of Theorem \ref{thm:Condition-Pertub}}
    \label{road map}
\end{figure}

% We also write the iterative process of three learning dynamics and denote $\Bar{\mathcal{A}}$ as the iterative matrix if payoff matrix is time invariant and equal to the matrix which is the singular value decomposition of $A$. 
% The assumptions of Theorem \ref{thm:Condition-Pertub} requires parameters of three algorithms satisfy some conditions, under these conditions, by Lemma \ref{lm:eig-static}, eigenvalues of iterative matrix $\Bar{\mathcal{A}}$ consists of 1 and other eigenvalues whose  modulus are less than 1. 
% In addition, by Proposition \ref{prop:diagonal}, iterative matrix $\Bar{\mathcal{A}}$ is diagonalizable. 
% We also give Lemma \ref{lm:norm bounded equivalence} showing that \ref{bap} for $B_t$ holds implies \ref{bap} for $\mathcal{B}_t$ holds. So far, we already satisfy the assumption of Lemma \ref{lm:SVD convergence}, which proves that $\lim_{t \to \infty} (\Bar{\mathcal{A}}- I)\Bar{X}_t =  \boldsymbol{0} \in \BR^{2(n+m)}$ with convergence rate $f(t)$.

%Lastly, assumptions of Lemma \ref{lm:convergence equivalent} is satisfied, leading to $\lim_{t \to \infty} (A^{\top}x_t, A y_t) = (\boldsymbol{0},\boldsymbol{0}) \in \BR^{n+m}$ with convergence rate $f(t)$ in three algorithms OGD, EG and Negative Momentum method has been proved. Since $\Delta_t = \lVert A^{\top}x_t \lVert_2 + \lVert A y_t \lVert_2$, then we complete the proof.

As a first step, we demonstrate that the iterative matrices of learning dynamics can be diagonalized using singular value decomposition (SVD), as shown in  Lemma \ref{prop:diagonal}. This phenomenon was also shown in \cite{gidel2019negative} for a general class of first order method. 
By singular value decomposition, we can write $A= U \Sigma_A V^{\top}$, where $U$, $V$ are unitary matrices, and $\Sigma_A$ is  rectangular diagonal matrix with its diagonal entries being singular values of $A$.
We denote this by
\[
\Sigma_A=
\begin{bmatrix}
\bf{\sigma_{r \times r}} & \bf{0_{r \times (m-r)}}\\
\\
\bf{0_{(n-r) \times r}} & \bf{0_{(n-r) \times (m-r)}}\\
\end{bmatrix}
\in \BR^{n \times m},
\]
and
\begin{align*}
\bf{\sigma_{r \times r}}  = 
\begin{bmatrix}
\sigma_1 \\
& \ddots & \\
& &  \sigma_r \\
\end{bmatrix}
\in \BR^{r \times r},
\end{align*}
where $\sigma_i>0$ are the singular values of $A$, $i\in[r]$. 
Let $ \Bar{x}_t = U^{\top}x_t$, $\Bar{y}_t = V^{\top}y_t$, then we can transform the iterative process of three algorithms in convergent perturbed game into the equivalent form as followings : 
\paragraph{SVD formulation for OGDA in convergent perturbed game:}
\begin{align*}%\label{SVD OGD}
& \Bar{x}_{t+1} = \Bar{x}_{t} - 2 \eta (\Sigma_A +U^{\top} B_{t} V )\Bar{y}_{t} + \eta (\Sigma_A +U^{\top} B_{t-1} V )\Bar{y}_{t-1}, \\
& \Bar{y}_{t+1}= \Bar{y}_{t} + 2 \eta  (\Sigma_A +V^{\top} B_{t}^{\top} U )\Bar{x}_{t} - \eta (\Sigma_A +V^{\top} B_{t-1}^{\top} U ) \Bar{x}_{t-1}.
\end{align*}
 We represent the above in the form of a linear difference system:
\begin{align}\label{SVD1}
\Bar{X}_{t+1}
= (\Bar{\mathcal{A}} +\Bar{\mathcal{B}}_{t})\Bar{X}_{t},
\end{align}
where $\Bar{X}_t=(\Bar{x}_t^{\top}, \Bar{y}_t^{\top}, \Bar{x}_{t-1}^{\top},\Bar{y}_{t-1}^{\top})^{\top} $,
\begin{align}\label{SVD-OGDA}
\setlength{\arraycolsep}{1pt}
\Bar{\mathcal{A}} =
\begin{bmatrix}
I & -2\eta \Sigma_A  & 0 & \eta \Sigma_A  \\
\\
2\eta \Sigma_A ^{\top} & I & -\eta \Sigma_A ^{\top} & 0\\
\\
I & 0& 0&0 \\
\\
0& I& 0&0 
\end{bmatrix} \in \BR^{2(m+n)},
\end{align}
and
\begin{align}\label{B-OGDA}
\Bar{\mathcal{B}}_{t}=
\begin{bmatrix}
0 & -2 \eta U^{\top}B_{t}V & 0 & \eta U^{\top}B_{t-1}V\\
\\
2 \eta V^{\top}B^{\top}_{t}U & 0 & -\eta V^{\top} B_{t-1}^{\top}U & 0 \\
\\
0 & 0& 0&0 \\
\\
0& 0& 0&0
\end{bmatrix} \in \BR^{2(m+n)}.
\end{align}

\vspace{3em}

\paragraph{SVD formulation for EG in convergent perturbed game:}
\begin{align*}%\label{ SVD EG }
& \Bar{x}_{t+1} = (I-\alpha \gamma ( \Sigma_A\Sigma_A^{\top} + U^{\top}(A B_{t}^{\top} + B_{t} A^{\top} + B_{t} B_{t}^{\top})U)\Bar{x}_{t} - \alpha  (\Sigma_A +U^{\top} B_{t} V)\Bar{y}_{t} , \\
& \Bar{y}_{t+1} = (I-\alpha \gamma (  \Sigma_A^{\top}\Sigma_A + V^{\top}(A^{\top} B_{t} + B_{t}^{\top} A + B_{t}^{\top} B_{t}) V)\Bar{y}_{t}+ \alpha (\Sigma_A^{\top} +V^{\top} B_{t}^{\top} U)\Bar{x}_{t}.
\end{align*}
 We represent the above in the form of a linear difference system:
\begin{align}\label{SVD2}
\Bar{X}_{t+1}
= (\Bar{\mathcal{A}} +\Bar{\mathcal{B}}_{t})\Bar{X}_{t},
\end{align}
where 
$\Bar{X}_t=(\Bar{x}_{t}^{\top}, \Bar{y}_{t}^{\top}) $,
\begin{align}\label{SVD-EG}
\setlength{\arraycolsep}{1pt}
\Bar{\mathcal{A}} =
\begin{bmatrix}
I-\alpha \gamma  \Sigma_A\Sigma_A^{\top} & - \alpha  \Sigma_A  \\
\\
\\
\alpha \Sigma_A ^{\top} & I -\alpha \gamma\Sigma_A^{\top} \Sigma_A \\
\end{bmatrix} \in \BR^{m+n}
\end{align}
and
\begin{align}
\Bar{\mathcal{B}}_{t}=
\begin{bmatrix}\label{B-EG}
-\alpha \gamma U^{\top}(A B_{t}^{\top} + B_{t} A^{\top} + B_{t} B_{t}^{\top})U & -\alpha U^{\top} B_{t} V \\
\\
\\
\alpha V^{\top} B_{t}^{\top} U &\alpha \gamma V^{\top}(A^{\top} B_{t} + B_{t}^{\top} A + B_{t}^{\top} B_{t}) V
\end{bmatrix}.
\end{align}

\vspace{3em}

\paragraph{SVD formulation for NM in convergent perturbed game:}
\begin{align*}%\label{NM OGD}
\Bar{x}_{t+1} &= \  (1+\beta_1) \Bar{x}_{t} - \eta
\left(\Sigma_A + U^{\top} B_{t} V\right)\Bar{y}_{t}- \beta_1 \Bar{x}_{t-1}, \\
\Bar{y}_{t+1} &= \  \left(I- \eta^2  \left(\Sigma_A^{\top}\Sigma_A + V^{\top}(A^{\top} B_{t} + B_{t+1}^{\top} A + B_{t+1}^{\top} B_{t}) V \right)\right)\Bar{y}_{t} \\
& +  \eta \left(\Sigma_A^{\top}+V^{\top} B_{t+1}^{\top} U \right)\left((1+\beta_1) \Bar{x}_{t}-\beta_1 \Bar{x}_{t-1}\right) -\beta_2 \Bar{y}_{t-1}.
\end{align*}

 We represent the above in the form of a linear difference system:
\begin{align}\label{SVD3}
\Bar{X}_{t+1}
= (\Bar{\mathcal{A}} +\Bar{\mathcal{B}}_{t})\Bar{X}_{t},
\end{align}
where 
$\Bar{X}_t=(\Bar{x}_t^{\top}, \Bar{y}_t^{\top}, \Bar{x}_{t-1}^{\top},\Bar{y}_{t-1}^{\top})^{\top} $,
\begin{align}\label{SVD-NM}
\setlength{\arraycolsep}{1pt}
\Bar{\mathcal{A}} =
\begin{bmatrix}
(1+\beta_1)I & & -\eta \Sigma_A  & & -\beta_1 I & & 0  \\
\\
\eta(1+\beta_1) \Sigma_A ^{\top} & & I- \eta^2  \Sigma_A^{\top}\Sigma_A & & -\eta\beta_1\Sigma_A ^{\top} &  &-\beta_2 I\\
\\
I & & 0& & 0& & 0 \\
\\
0& & I& & 0& &0 
\end{bmatrix} \in \BR^{2(m+n)},
\end{align}
and
\begin{align}
\Bar{\mathcal{B}}_{t}=
\begin{bmatrix}\label{B-NM}
0 & -\eta U^{\top} B_{t} V & 0 & &  0\\
\\
\eta(1+\beta_1) V^{\top}B^{\top}_{t+1}U & -\eta^2 V^{\top}(A^{\top} B_{t} + B_{t+1}^{\top} A + B_{t+1}^{\top} B_{t}) V & -\eta \beta_1 V^{\top}B^{\top}_{t+1}U & & 0 \\
\\
0 & 0& 0& & 0 \\
\\
0& 0& 0& & 0 
\end{bmatrix}.
\end{align}

\begin{lem}\label{normalEG}
The iterative matrix of SVD formulation for EG in convergent perturbed game in \eqref{SVD-EG} is a normal matrix.
\end{lem}

\begin{proof}Directly calculate shows
\begin{align}
\Bar{\mathcal{A}} \Bar{\mathcal{A} }^{\top}  = \Bar{ \mathcal{A}}^{\top}   \Bar{\mathcal{A} }
= 
\begin{bmatrix}
(I - \alpha \gamma \Sigma_A \Sigma_A^{\top})^2 + \alpha^2 \Sigma_A \Sigma_A^{\top} & 0\\
\\
0 & (I - \alpha\gamma \Sigma_A^{\top} \Sigma_A)^2 + \alpha^2 \Sigma_A^{\top} \Sigma_A 
\end{bmatrix}.
\end{align}
    
\end{proof}

\begin{lem}\label{prop:diagonal}
	For a fixed payoff matrix $ A $, the corresponding iterative matrices in \eqref{SVD-OGDA} for OGDA, \eqref{SVD-EG} for EG, and \eqref{SVD-NM} for negative momentum method are diagonalizable.
\end{lem}

Note that the claim is true for EG since the iterative matrix is normal as we have shown in lemma \ref{normalEG}. Therefore, we will only consider the cases of  OGDA and negative momentum method below.
The idea behind proving these two claims is the same, we construct a set of linearly independent eigenvectors of  \eqref{SVD-OGDA} or \eqref{SVD-NM} that form a basis of $\BR^{2(m+n)}$, and under this basis, \eqref{SVD-OGDA} or \eqref{SVD-NM} can be represented by a diagonal matrix.

\begin{proof}
We firstly define some notation. In the following, we denote $ e_i^n $ as an $ n $-dimensional unit vector with 1 in the $ i $-th position and 0 in other positions and denote $\sigma_p$ as the $p$-th singular value of the payoff matrix $A$ of the stable game, and denote $r$ as the rank of $A$. Thus for $p \in [r]$, $\sigma_p > 0$, and otherwise $\sigma_p = 0$. We will also denote the $n$-dimensional ($m$-dimensional) zero vector as $0^n (0^m)$.

	\paragraph{Part \uppercase\expandafter{\romannumeral1},  Diagonalization of \eqref{SVD-OGDA}:}
	Now we consider the diagonalization of matrix in $ \eqref{SVD-OGDA} $. Recall that
	\[
	\Sigma_A = 
	\begin{bmatrix}
	\bf{\sigma_{r \times r}} & \bf{0_{r \times (m-r)}}\\
 \\
	\bf{0_{(n-r) \times r}} & \bf{0_{(n-r) \times (m-r)}}\\
	\end{bmatrix}
	\in \BR^{n \times m},
	\]
	and
	\begin{align*}
	\bf{\sigma_{r \times r}}  = 
	\begin{bmatrix}
	\sigma_1 \\
	& \ddots & \\
	& &  \sigma_r \\
	\end{bmatrix}
	\in \BR^{r \times r}.
	\end{align*}
	To prove $\Bar{\mathcal{A}} $ is diagonalizable, we only need to find $ 2(n+m) $ linearly independent eigenvectors of the matrix. %Let $ e_i^n $ denotes an $ n $-dimensional unit vector with 1 in the $ i $-th position and 0 in other positions.
	%Similarly $ e_j^m $ denotes an $ m $-dimensional unit vector with 1 in the $ j $-th position and 0 in other positions.  
 
 Then we can check the equations below
	%Then from the form of $ \Sigma $ we have
	\begin{align*}
	\begin{cases}
	\Sigma_A e_p^m=\sigma_p e_p^n,  & \text{for  $1\le p\le r$},\\
    \\
	\Sigma_A e_j^m=0^n, & \text{for  $r+1\le j \le m$},
	\end{cases}
	\end{align*}
	and
	\begin{align*}
	\begin{cases}
	\Sigma_A^{\top} e_p^n=\sigma_p e_p^m,  & \text{for  $1\le p\le r$},\\
    \\
	\Sigma_A^{\top} e_i^n=0^m, & \text{for  $r+1\le i \le n$}.
	\end{cases}
	\end{align*}
	Now we respectively construct the eigenvectors corresponding to each eigenvalue, and prove these $ 2(n+m) $ vectors are linearly independent, forming a basis of $\BR^{2(n+m)}$.
	\begin{enumerate}
		\item[\textbf{Case 1}] Eigenvectors correspond to eigenvalue 1 :\\
		It can be verified that for $ r+1\le i\le n $
		\begin{align*}
		v_{1,i}=
		\begin{bmatrix}
		e_i^n\\
  \\
		0^m\\
  \\
		0^n\\
  \\
		0^m
		\end{bmatrix},
		\end{align*}
		and for $ r+1\le j\le m $
		\begin{align*}
		w_{1,j}=
		\begin{bmatrix}
		0^n\\
  \\
		e_j^m\\
  \\
		0^n\\
  \\
		0^m
		\end{bmatrix}
		\end{align*}
		are eigenvectors of $ \Bar{\mathcal{A}} $ belonging to eigenvalue 1. 		
		\item[\textbf{Case 2}] Eigenvectors correspond to eigenvalue 0 :\\
		It can be verified that for $ r+1\le i\le n $,
		\begin{align*}
		v_{0,i}=
		\begin{bmatrix}	
		0^n\\
  \\
		0^m\\
  \\
		e_i^n\\
  \\
		0^m
		\end{bmatrix}
		\end{align*}
		and for $ r+1\le j\le m $,
		\begin{align*}
		w_{0,j}=
		\begin{bmatrix}
		0^n\\
  \\
		0^m\\
  \\
		0^n\\
  \\
		e_j^m
		\end{bmatrix}
		\end{align*}
		are eigenvectors of $ \Bar{\mathcal{A}}  $ belonging to eigenvalue 0. 
		
		\item[\textbf{Case 3}] Other eigenvectors : \\
		For $ p=1,\cdots,r $, consider the roots of the following polynomial :
		\begin{align}\label{OGDAquarpoly}
		\lambda^2(\lambda-1)^2+\eta^2\sigma_p^2(1-2\lambda)^2=0
		\end{align}
		where $ \sigma_p $ is the $ p $-th diagonal element of $ \Sigma_A $, and the solution of these polynomials are
        eigenvalues of $\Bar{\mathcal{A}}$.
%  According to the multiple root discriminant of quartic polynomials, we can yield a polynomial of $\eta$. That means if we choose $\eta$ such that the polynomial of $\eta$ not equal to 0, equation (\ref{OGDAquarpoly}) doesn't have multiple root.

  We first claim that except for finite choices of $ \eta $, equation \eqref{OGDAquarpoly} has four different non-zero roots, denote them as $ \lambda_{p,q}$, $q=1,2,3,4 $. That is because a quartic polynomial equation has multiple roots if and only if its discriminant polynomial, a homogeneous polynomial with degree $6$ on the coefficients of the quartic polynomial equation, equals to $0$. Since a degree $6$ polynomial has at most $6$ roots, thus if $\eta$ is not a root of this discriminant polynomial, \eqref{OGDAquarpoly} will not have multiple roots. In the following, we will
  choose $\eta$ such that \eqref{OGDAquarpoly} has no multiple roots.
  According to Lemma \ref{lm:eig-static}, the modulus of these eigenvalues are less than 1.
  
  Let 
  $$ \alpha_{p,q}=\frac{\eta\sigma_p(1-2\lambda_{p,q})}{{\lambda_{p,q}}^2-\lambda_{p,q}}. $$
		
		It can be verified that for $1\le p\le r $ and $q = 1,2,3,4$,
		\begin{align*}
		u_{p,q}=
		\begin{bmatrix}
		\lambda_{p,q} \alpha_{p,q} e_p^n\\
  \\
		\lambda_{p,q} e_p^m\\
  \\
		\alpha_{p,q} e_p^n\\
  \\
		e_p^m
		\end{bmatrix}
		\end{align*}
	are the eigenvectors of $ \Bar{\mathcal{A}} $ corresponding to eigenvalue $ \lambda_{p,q} $.		
	\end{enumerate}
	Then we have constructed $ 2(n+m) $ eigenvectors, now we prove they are linearly independent.
	Suppose there exists coefficients $ k_{1,i},k_{0,i}$, where $i=r+1,\cdots,n$ , $g_{1,j},g_{0,j}$, where $j=r+1,\cdots,m$ and $f_{p,q}$, where $p=1,\cdots,r$ and $q=1,2,3,4 $, 
	%some of them not equal to zero, 
	such that
	\begin{equation}
	\begin{aligned}\label{OGDAsum0}
	&\sum_{i=r+1}^{n} k_{1,i} v_{1,i} + \sum_{j=r+1}^{m} g_{1,j} w_{1,j}+ 	\sum_{i=r+1}^{n} k_{0,i} v_{0,i}
 + \sum_{j=r+1}^{m} g_{0,j} w_{0,j}+ \sum_{p=1}^{r} \sum_{q=1}^{4} f_{p,q} u_{p,q}=0.
	\end{aligned} 
	\end{equation}
 
For $ r+1 \le i \le n $, only $ v_{1,i} $ has non-zero element  at the $ i $-th position of vector, so $ k_{1,i}=0 $. 
 
 %For $ j=r+1,\cdots,m $, at the $ (j+n) $-th position of vector, only $ w_{1,j} $ has non-zero element, so $ g_{1,j}=0 $ for $ j=r+1,\cdots,m $.

 For $ r+1 \le j \le m $, only $ w_{1,j} $ has non-zero element at the $ (j+n) $-th position of vector, so $ g_{1,j}=0 $.
 
 For $ r+1 \le i \le n $, only $ v_{0,i} $ has non-zero element at the $ (i+n+m) $-th position of vector, so $ k_{0,i}=0 $. 
 
 For $ r+1 \le j \le m $, only $ w_{0,j}$ has non-zero element at the $ (j+2n+m) $-th position of vector, so $ g_{0,j}=0 $. 
 
 For $ 1 \le p \le r $, at the $ p $-th position of vector, only $ u_{p,q}$, where $ q=1,2,3,4  $ has non-zero element. So we can yield
	\begin{align*}
	\sum_{q=1}^{4} f_{p,q}u_{p,q}=0.
	\end{align*}
	The above equation holds for $ p = 1,\cdots,r $. 
	Because the eigenvectors of different eigenvalues are linearly independent, we have $ f_{p,q}=0$, where $q=1,2,3,4$ and $ p=1,\cdots,r $. 
	Now we have concluded that all coefficients in \eqref{OGDAsum0} are zero, thus these eigenvectors are linearly independent. 
 
 Let $P$ be the matrix whose columns are consisted by the eigenvectors of $\Bar{\mathcal{A}}$ constructed above, and $D$ be the diagonal matrix whose diagonal elements are eigenvalues of $\Bar{\mathcal{A}}$.
 After an appropriate order arrangement of columns on $P$ and elements on $D$, we have
\begin{align*}
	\Bar{\mathcal{A}} P=PD.
	\end{align*}
Moreover, as we have shown above, the columns of $P$ are linearly independent, therefore $P$ is invertible, which implies $\Bar{\mathcal{A}}$ is diagonalizable. 
 
%	\begin{align*}
%	P= \{v_{1,r+1},\cdots ,v_{1,n},w_{1,r+1},\cdots ,w_{1,m},v_{0,r+1},\cdots ,v_{0,n},w_{0,r+1},\cdots ,w_{0,m},\\
%	u_{1,1},u_{1,2},u_{1,3},u_{1,4},\cdots,u_{r,1},u_{r,2},u_{r,3},u_{r,4}\}
%	\end{align*}
%	be the basis matrix. 
%$ P $ is invertible since these eigenvectors are linearly independent, and we have
%	\begin{align*}
%	\Bar{\mathcal{A}} P=PD,
%	\end{align*}
%	where $ D=\text{diag}\{\overbrace{1,\cdots,1}^{n-r},\overbrace{1,\cdots,1}^{m-r},\overbrace{0,\cdots,0}^{n-r},\overbrace{0,\cdots,0}^{m-r},\overbrace{\lambda_{1,1},\lambda_{1,2},\lambda_{1,3},\lambda_{1,4},\cdots,\lambda_{r,1},\lambda_{r,2},\lambda_{r,3},\lambda_{r,4}}^{4r}\} $. 

	%\textcolor{red}{(numbers of each eigenvalue)}
 % \begin{rem}
 %     For payoff matrix $A$, given $A=U\Sigma_A V^{\top} $, denote
 %     \begin{align*}
 %         Q=
 %         \begin{bmatrix}
 %             U & & &\\
 %             & V & &\\
 %             & & U &\\
 %             & & & V
 %         \end{bmatrix},
 %     \end{align*}
 %     where $Q$ is an unitary matrix. Then it can be varified that
 %     \begin{align*}
 %         Q^\top\mathcal{A}Q=\Bar{\mathcal{A}}
 %     \end{align*}
 %     which means ${\Bar{\mathcal{A}}} $ in (\ref{SVD-OGDA}) is diagonalizable implies that ${\mathcal{A}} $ in (\ref{STATIC-OGDA}) is diagonalizable.
 % \end{rem}
	% \textcolor{red}{Therefore, $ \Bar{\mathcal{A}} $ is diagonalizable .}

	\paragraph{Part \uppercase\expandafter{\romannumeral2},  Diagonalization of \eqref{SVD-NM}:}
	Now we consider the diagonalization of the matrix in $ (\ref{SVD-NM}) $ and denote it as $ \Bar{\mathcal{A}} $. 
%	\begin{align}\label{momentummatrix}
%	\begin{bmatrix}
%	(1+\beta_1) I& -\eta \Sigma & -\beta_1 I & 0 \\
%	\\
%	\eta(1+\beta_1) \Sigma^{\top} & (1+\beta_2) I-\eta^2\Sigma^{\top}\Sigma  & -\eta\beta_1 \Sigma^{\top} & -\beta_2 I\\
%	\\
%	I & 0& 0&0 \\
%	\\
%	0& I& 0&0 
%	\end{bmatrix},
%	\end{align}
	Similiarly, to prove this matrix is diagonalizable, we only need to find $ 2(n+m) $ linearly independent eigenvectors of the matrix. Now we respectively construct the eigenvectors corresponding to each eigenvalue, and prove these $2(n+m)$ vectors are linearly independent, forming a basis of $\BR^{2(n+m)}$.
	
	\begin{enumerate}
		\item[\textbf{Case 1:}] Eigenvectors correspond to eigenvalue 1 :\\
		It can be verified that for $ r+1\le i\le n $,
		\begin{align*}
		v_{1,i}=
		\begin{bmatrix}
		e_i^n\\
  \\
		0^m\\
  \\
		e_i^n\\
  \\
		0^m
		\end{bmatrix}
		\end{align*}
		and for $ r+1\le j\le m $,
		\begin{align*}
		w_{1,j}=
		\begin{bmatrix}
		0^n\\
  \\
		e_j^m\\
  \\
		0^n\\
  \\
		e_j^m
		\end{bmatrix}
		\end{align*}
		are eigenvectors of $ \Bar{\mathcal{A}} $ belonging to eigenvalue 1. 
		%	That means  $ \Bar{A} $ has $ n+m-2r $ eigenvectors correspongding to  eigenvalue 1.
		\item[\textbf{Case 2:}] Eigenvectors correspond to eigenvalue $ \beta_1 $.\\
		It can be verified that for $ r+1\le i\le n $,
		\begin{align*}
		v_{\beta_1,i}=
		\begin{bmatrix}
		\beta_1 e_i^n\\
  \\
		0^m\\
  \\
		e_i^n\\
  \\
		0^m
		\end{bmatrix}
		\end{align*}
		are eigenvectors of $ \Bar{\mathcal{A}} $ corresponding to eigenvalue $ \beta_1 $.

		\item[\textbf{Case 3:}] Eigenvectors correspond to eigenvalue $ \beta_2 $.\\
		It can be verified that for $ r+1\le j\le m $,
		\begin{align*}
		w_{\beta_2,j}=
		\begin{bmatrix}
		0^n\\
  \\
		\beta_2e_j^m\\
  \\
		0^n\\
  \\
		e_j^m
		\end{bmatrix}
		\end{align*}
		are eigenvectors of $\Bar{\mathcal{A}} $ corresponding to eigenvalue $ \beta_2 $.
		
		%	\item[\textbf{case 3:}] eigenvector of $ \beta_2 $.\\
		%	When we choose 
		%	It can be verified that 
		%	\begin{align}\label{quarpoly}
		%	v_j^{\beta_2}=
		%	\begin{bmatrix}
		%	0\\
		%	\beta_2e_j^m\\
		%	0\\
		%	e_j^m
		%	\end{bmatrix}
		%	\end{align}
		%	for $ r+1\le j\le n $
		%	are eigenvectors of $ \Bar{A} $ corresponding to eigenvalue $ \beta_2 $.
		\item[\textbf{Case 4:}] Other eigenvectors. \\
		For $ p=1,\cdots,r $, consider four roots of polynomial
		\begin{align}\label{quarpoly}
		(\lambda-1)^2(\lambda-\beta_1)(\lambda-\beta_2)+\eta^2\sigma_p^2\lambda^3=0
		\end{align}
		where $ \sigma_p $ is the $ p $-th diagonal element of $ \Sigma_A $. 
		
		Now we consider the effect of different value of $ \beta_1 $ and $ \beta_2 $.
		%	\begin{enumerate}
		%		\item $ \beta_1\neq 0 $ or $ \beta_2\neq 0 $\\
%		\textcolor{red}{(the sign of its discriminant of quartic polynomial)} 
		If $ \beta_1=0 $ and $ \beta_2=0 $, the model degenerates to gradient descent algorithm, we only consider when $ \beta_1 $ and $ \beta_2 $ are not both zero.
		%	So if $ \lambda_p^q=0 $, that means either $ \beta_1=0 $ or $ \beta_2=0 $. 
	Similar to the situation in the Case 3 of diagonalization of (\ref{SVD-OGDA}), except for several values for $ \eta $, equation (\ref{quarpoly}) has four different roots, denote them as $ \lambda_{p,q}$, $q=1,2,3,4 $. 
		If $ \lambda_{p,q}\neq 0 $, for $q=1,2,3,4 $,
		let 
  \begin{align*}
     \alpha_{p,q}=\frac{-{\lambda_{p,q}}^2+(1+\beta_1)\lambda_{p,q}-\beta_1}{\eta \sigma_p \lambda_{p,q}}. 
  \end{align*}
  We can check that 
		\begin{align*}
		u_{p,q}=
		\begin{bmatrix}
		\lambda_{p,q} e_p^n\\
  \\
		\lambda_{p,q} \alpha_{p,q} e_p^m\\
  \\
		e_p^n\\
  \\
		\alpha_{p,q} e_p^m
		\end{bmatrix}
		\end{align*}
		is the eigenvector of $ \Bar{\mathcal{A}} $ corresponding to eigenvalue $ \lambda_{p,q} $, that is $ \Bar{\mathcal{A}}u_{p,q}=\lambda_{p,q} u_{p,q} $. Else if $ \lambda_{p,q}=0 $, that means either $ \beta_1=0 $ or $ \beta_2=0 $. 
  
  If $ \beta_1=0 $, 
		\begin{align*}
		u_{p,q}=
		\begin{bmatrix}
		0^n\\
  \\
		0^m\\
  \\
		e_p^n\\
  \\
		0^m
		\end{bmatrix}
		\end{align*}
		is the eigenvector of $ \Bar{\mathcal{A}} $ corresponding to eigenvalue $ 0 $.
  
  If $ \beta_2=0 $,
		\begin{align*}
		u_{p,q}=
		\begin{bmatrix}
		0^n\\
  \\
		0^m\\
  \\
		0^n\\
  \\
		e_p^m
		\end{bmatrix}
		\end{align*}
		is the eigenvector of $ \Bar{\mathcal{A}} $ corresponding to eigenvalue $ 0 $.
		
		%	\end{enumerate}
		
	\end{enumerate}
	Now we obtain $ 2(n+m) $ eigenvectors, in the following we will prove these $ 2(n+m) $ eigenvectors are linearly independent.
	Suppose there exists coefficients $ k_{1,i},k_{\beta_1,i}$, where $i=r+1,\cdots,n$ , $g_{1,j},g_{\beta_2,j}$, where $j=r+1,\cdots,m$ and $f_{p,q}$, where $p=1,\cdots,r$ and $q=1,2,3,4 $, 
	%some of them not equal to zero, 
	such that
	\begin{align}\label{sum0}
	\sum_{i=r+1}^{n} k_{1,i} v_{1,i} + \sum_{j=r+1}^{m} g_{1,j} w_{1,j}+ 	\sum_{i=r+1}^{n} k_{\beta_1,i} v_{\beta_1,i} +	\sum_{j=r+1}^{m} g_{\beta_2,j} v_{\beta_2,j} +\sum_{p=1}^{r} \sum_{q=1}^{4} f_{p,q} u_{p,q}=0.
	\end{align} 
	First we prove $ f_{p,q} =0 $ for $ p=1,\cdots,r $ and $ q=1,2,3,4  $.  If $ \beta_1= 0 $, let $ l=1,\cdots,r $, 
	then at the $ l+n+m $-th position of vector, only $ u_{l,q}$,  $ q=1,2,3,4  $ has non-zero element. Else if $ \beta_1\neq 0 $, let  $ l=1,\cdots,r $, 
	then at the $ l+2n+m $-th position of vector, only $ u_{l,q}$,  $ q=1,2,3,4  $ has non-zero element. For these two case we both have 
	\begin{align*}
	\sum_{q=1}^{4} f_{p,q} u_{p,q}=0.
	\end{align*}
	The above equation holds for $ p=1,\cdots,r $. 
	Because the eigenvectors of different eigenvalues are linearly independent, we have $ f_{p,q}=0$, where $q=1,2,3,4$ and $ p=1,\cdots,r $.

	For $ i=r+1,\cdots,n $, at the $ i $-th position of vector, only $ v_{1,i} $ and $  v_{\beta_1,i} $ has non-zero element at this position, so we obtain 
	\begin{align*}
	k_{1,i} v_{1,i} +  k_{\beta_1,i}  v_{\beta_1,i}=0.
	\end{align*}
	Notice that $ \beta_1\neq 1 $, this means $ k_{1,i}=0$ and $ k_{\beta_1,i} =0$, where $i=r+1,\cdots,n $.
	
	For $ j=r+1,\cdots,m $, at the $ j $-th position of vector, only $ w_{1,j} $ and $  w_{\beta_2,j} $ has non-zero element, so we can yield 
	\begin{align*}
	g_{1,j} w_{1,j} +  g_{\beta_2,j} w_{\beta_2,j}=0.
	\end{align*}
	Similarly, because of $ \beta_2\neq 1 $, this means $ g_{1,j}=0$ and $ g_{\beta_2,j}=0$, where $j=r+1,\cdots,m $.
	
	We prove that if $ \eqref{sum0} $ holds, then all coefficients are zero, which illustrates that these eigenvectors are linearly independent. Same as the argument in Part $\uppercase\expandafter{\romannumeral1}$ of the proof, the existence of these $2(m+n)$ eigenvectors implies $\Bar{\mathcal{A}} $ is diagonalizable.

% Let
%	\begin{align*}
%	P= [v_{1,r+1},\cdots ,v_{1,n},w_{1,r+1},\cdots ,w_{1,m},v_{\beta_1,r+1},\cdots, v_{\beta_1,n},w_{\beta_2,r+1},\cdots, w_{\beta_2,m},\\
%	u_{1,1} ,u_{1,2},u_{1,3},u_{1,4},\cdots,u_{r,1},u_{r,2},u_{r,3},u_{r,4}]
%	\end{align*}
%	be the basis matrix. Linearly independece of these eigenvectors means $ P $ is inevertible. Then there holds
%	\begin{align*}
%	\Bar{\mathcal{A}}P=PD,
%	\end{align*}
%	where $ D=\text{diag}\{\overbrace{1,\cdots,1}^{n+m-r},\overbrace{\beta_1,\cdots,\beta_1}^{n-r},\overbrace{\beta_2,\cdots,\beta_2}^{m-r},\overbrace{\lambda_{1,1},\lambda_{1,2},\lambda_{1,3},\lambda_{1,4},\cdots,\lambda_{r,1},\lambda_{r,2},\lambda_{r,3},\lambda_{r,4}}^{4r}\} $. 
	%\textcolor{red}{(numbers of each eigenvalue)}
%	Therefore, $ \Bar{\mathcal{A}} $ is diagonalizable.
 % , moreover $ {\mathcal{A}} $ is diagonalizable
\end{proof}

 \begin{rem}
     For payoff matrix $A$, given its SVD decomposition $A=U\Sigma_A V^{\top} $, let
     \begin{align*}
         Q=
         \begin{bmatrix}
             U & & &\\
             & V & &\\
             & & U &\\
             & & & V
         \end{bmatrix},
     \end{align*}
     then $Q$ is a unitary matrix. Furthermore, it can be verified that
     \begin{align*}
         Q^\top\mathcal{A}Q=\Bar{\mathcal{A}}
     \end{align*}
     for both OGDA and negative momentum method. That means 
     \begin{enumerate}
         \item ${\Bar{\mathcal{A}}} $ in (\ref{SVD-OGDA}) is diagonalizable implies that ${\mathcal{A}} $ in (\ref{STATIC-OGDA}) is diagonalizable.
         \item ${\Bar{\mathcal{A}}} $ in (\ref{SVD-NM}) is diagonalizable implies that ${\mathcal{A}} $ in (\ref{STATIC-NM}) is diagonalizable.
     \end{enumerate}
     
 \end{rem}

\begin{lem}[Gronwall inequality, \cite{colonius2014dynamical}]\label{lm:Gronwall inequality}
Let for all $t \in \BN$, the functions $u,p,q,f : \BN \to \BR$ satisfy 
$$u(t) \le p(t) + q(t) \sum^{t-1}_{\ell = a} f(\ell)u(\ell).$$
Then, for all $t \in \BN$
\begin{align}
\tag{Gronwall inequality}
u(t) \le p(t) + q(t) \sum^{t-1}_{\ell = a} p(\ell)f(\ell) \prod^{k-1}_{\tau = \ell +1}(1 + q(\tau)f(\tau)).
\end{align}
\end{lem}

Gronwall inequality is a useful tool to treat linear difference equations, it also has an analogy in continuous time case. For more about Gronwall inequality, see Lemma 6.1.3 in \cite{colonius2014dynamical}.

\begin{lem}\label{lm:norm bounded equivalence}
	If $ \{B_t\}_t $ satisfy the BAP assumption, i.e.,  $\sum^{\infty}_{t=1} \lVert B_t \lVert_2 $ is bounded, then $ \{ \Bar{\mathcal{B}}_t\}_t $ defined in \eqref{B-OGDA}, \eqref{B-EG} and \eqref{B-NM}  also satisfy BAP assumption.
	%there exists $c_1$ such that 
%	$\sum_{t=1}^{\infty} \lVert \Bar{B}_t \lVert \leq c_1$.
\end{lem}

\begin{proof}
	We claim there exists some constant $ c $, such that for any $ t $, $\lVert \Bar{\mathcal{B}}_{t-1}\lVert_2 \le c\left( \lVert {B}_{t}\lVert_2 +\lVert {B}_t\lVert_2 +\lVert {B}_{t+1}\lVert_2 \right)$. With this property, we have 
 \begin{align*}
     \sum_{t=0}^{\infty} \lVert \Bar{\mathcal{B}}_t \lVert_2  \le 3c\sum_{t=0}^{\infty} \lVert B_t \lVert_2  < +\infty,
 \end{align*} 
 then we prove the statement. In the following, we prove above claim for OGDA, EG, and negative momentum method.
	\paragraph{Case of OGDA : }
	We consider the matrix $ \eqref{B-OGDA} $
	\begin{align*}
	\Bar{\mathcal{B}}_{t}=&
	\begin{bmatrix}
	0 & -2 \eta U^{\top}B_{t}V & 0 & \eta U^{\top}B_{t-1}V\\
 \\
	2 \eta V^{\top}B^{\top}_{t}U & 0 & -\eta V^{\top} B_{t-1}^{\top}U & 0 \\
 \\
	0 & 0& 0&0 \\
 \\
	0& 0& 0&0
	\end{bmatrix}\\
	&
	=2\eta
	\begin{bmatrix}
	0 &-U^{\top}B_{t}V & 0 & & 0\\
 \\
	V^{\top}B^{\top}_{t}U & 0 &0 & & 0 \\
 \\
	0 & 0& 0& & 0 \\
 \\
	0& 0& 0& & 0
	\end{bmatrix}
	+\eta
	\begin{bmatrix}
	0 & & 0 & 0 &  U^{\top}B_{t-1}V\\
 \\
	0 & & 0 & -V^{\top} B_{t-1}^{\top}U & 0 \\
 \\
	0 & & 0& 0&0 \\
 \\
	0& & 0& 0&0
	\end{bmatrix}.
	\end{align*}
	Denote the first matrix in right side of the equation as $ H_1 $, and the second one as $H_2 $. From the above equation, we can obtain that $ \lVert \Bar{\mathcal{B}}_{t}\lVert_2 \le \lVert H_1\lVert_2 +\lVert H_2\lVert_2  $. Recall the definition of 2-norm of matrix,
	\begin{align*}
	\lVert H_1\lVert_2  =\max \sqrt{\text{Eigenvalue}\{H_1^{\top}H_1\}},\\
 \\
	\lVert H_2\lVert_2  =\max \sqrt{\text{Eigenvalue}\{H_2^{\top}H_2\}},
	\end{align*}
	then, 
	\begin{align*}
	H_1^{\top}H_1=4\eta^2
	\begin{bmatrix}
	U^{\top}B_{t}B^{\top}_{t}U& 0 & 0 & & 0  \\
 \\
	0 & V^{\top}B^{\top}_{t}B_{t}V & 0 & & 0 \\
 \\
	0 & 0& 0 & & 0 \\
 \\
	0& 0& 0 & & 0
	\end{bmatrix}
	\end{align*}
	and
	\begin{align*}
	H_2^{\top}H_2=\eta^2
	\begin{bmatrix}
	0 & & 0& 0&0 \\
 \\
	0& & 0& 0&0\\
 \\
	0& & 0 & U^{\top}B_{t-1}B^{\top}_{t-1}U & 0  \\
 \\
	0 & & 0& 0 & V^{\top}B^{\top}_{t-1}B_{t-1}V    
	\end{bmatrix}.
	\end{align*}
	Because $ U $ and $ V $ are unitary matrices, we have
	\begin{align*}
	\lVert H_{1} \lVert_2 = \max \sqrt{\text{Eigenvalue}\{H_1^{\top}H_1\}}= 4 \eta^2 \max \sqrt{\text{Eigenvalue}\{B_{t}^{\top}B_{t}\}}= 4 \eta^2 \lVert B_{t} \lVert_2 
	\end{align*}
	and
	\begin{align*}
\lVert H_{2} \lVert_2 =	\max \sqrt{\text{Eigenvalue}\{H_2^{\top}H_2\}}=
\eta^2 \max \sqrt{\text{Eigenvalue}\{B_{t-1}^{\top}B_{t-1}\}}= \eta^2 \lVert B_{t-1} \lVert_2 .
	\end{align*}
	Let $ c=4\eta^2 $, then
	\begin{align*}
	\lVert \Bar{\mathcal{B}}_{t} \lVert_2 \le c\cdot( \lVert B_{t} \lVert_2 + \lVert B_{t-1} \lVert_2  ),
	\end{align*}
	we have completed the proof for OGDA.

	\paragraph{Case of EG : }
	We consider the matrix $ \eqref{B-EG} $
	\begin{align*}
	\Bar{\mathcal{B}}_{t}=&
	\begin{bmatrix}
	-\alpha \gamma U^{\top}(A B_{t}^{\top} + B_{t} A^{\top} + B_{t} B_{t}^{\top})U & -\alpha U^{\top} B_{t} V \\
 \\
	\alpha V^{\top} B_{t}^{\top} U &\alpha \gamma V^{\top}(A^{\top} B_{t} + B_{t}^{\top} A + B_{t}^{\top} B_{t}) V
	\end{bmatrix}
	\\
	=&
	\begin{bmatrix}
	-\alpha \gamma U^{\top}(A B_{t}^{\top} + B_{t} A^{\top} + B_{t} B_{t}^{\top})U &0\\
 \\
	0& 0
	\end{bmatrix}
	\\
	& +
	\begin{bmatrix}
	0 &0\\
 \\
	0& \alpha \gamma V^{\top}(A^{\top} B_{t} + B_{t}^{\top} A + B_{t}^{\top} B_{t}) V
	\end{bmatrix}
	\\
	& +
	\begin{bmatrix}
	0 & -\alpha U^{\top} B_{t} V\\
 \\
	0 & 0
	\end{bmatrix}
	+ 
	\begin{bmatrix}
	0 & 0\\
 \\
	\alpha V^{\top} B_{t}^{\top} U & 0
	\end{bmatrix}.
	\end{align*}
	
	We separate $\mathcal{B}_1$ into four matrices and denote these matrices in right side of the equation as $H_1$, $H_2$, $H_3$ and $H_4$, respectively. Then  
 \begin{align*}
     \lVert\Bar{\mathcal{B}}_{t} \lVert_2 \leq \lVert H_1 \lVert_2 + \lVert H_2 \lVert_2 + \lVert H_3 \lVert_2 + \lVert H_4 \lVert_2.
 \end{align*}

	Since $\sum_{t=1}^{\infty} \lVert B_{t} \lVert_2 \leq c$, then $\lVert B_{t} \lVert_2 \leq c$ for any $t$. We also assume that $c_2= \lVert A \lVert_2$. 
 
 Then we have,
	\begin{align*}
	\lVert H_1 \lVert_2  = & \alpha \gamma \lVert
	U^{\top}(A B_{t}^{\top} + B_{t} A^{\top} + B_{t} B_{t}^{\top})U \lVert_2 \\
 \\
	= &\alpha \gamma\left(\lVert
	A B_{t}^{\top} + B_{t} A^{\top} + B_{t} B_{t}^{\top} \lVert_2 \right) \\
 \\
	\leq &\alpha \gamma\left(\lVert A \lVert_2 \lVert B_{t} \lVert_2 + \lVert A \lVert_2 \lVert B_{t} \lVert_2 +  \lVert B_{t} \lVert_2  \lVert B_{t} \lVert_2\right) \\
 \\
	\leq &\alpha \gamma (2c_2+c) \lVert B_{t} \lVert_2,
	\end{align*}
	where the second equality is due to $U$ is unitary matrix. Similarly, $\lVert H_2 \lVert_2 \leq \alpha\gamma(2c_2+c) \lVert B_{t} \lVert_2$.
	In addition, $\lVert H_3 \lVert_2 = \lVert H_4 \lVert_2 = \alpha \lVert B_{t} \lVert_2$. 
 
    Thus for any $t$,  we have the inequality between 
	$\lVert\Bar{\mathcal{B}}_{t} \lVert_2$ and $\lVert B_{t} \lVert_2$:
	\begin{align*}
	\lVert\Bar{\mathcal{B}}_{t} \lVert_2 \leq \lVert H_1 \lVert_2 + \lVert H_2 \lVert_2 + \lVert H_3 \lVert_2 + \lVert H_4 \lVert_2 \leq \alpha \left((4c_2+c)\gamma +2\right)\lVert B_{t} \lVert_2 .
	\end{align*}
	Let $c_1 = c(4c_2+c)\gamma +2c$, summing the above inequality over $t$, we have
	\begin{align*}
	\sum_{t=1}^{\infty} \lVert\Bar{\mathcal{B}}_{t} \lVert_2 \leq \left((4c_2+c)\gamma +2\right)\sum_{t=1}^{\infty}\lVert B_{t} \lVert_2 \leq c\left((4c_2+c)\gamma +2\right) = c_1.
	\end{align*}
	\paragraph{Case of Negative Momentum Method : }
	We consider the matrix \eqref{B-NM} ,
	\begin{align*}
\Bar{\mathcal{B}}_{t}=&
\begin{bmatrix}
0 & -\eta U^{\top} B_{t} V & 0 & &  0\\
\\
\eta(1+\beta_1) V^{\top}B^{\top}_{t+1}U & -\eta^2 V^{\top}(A^{\top} B_{t} + B_{t+1}^{\top} A + B_{t+1}^{\top} B_{t}) V & -\eta \beta_1 V^{\top}B^{\top}_{t+1}U & & 0 \\
\\
0 & 0& 0& & 0 \\
\\
0& 0& 0& & 0 
\end{bmatrix}\\
\\
	&=\begin{bmatrix}
	0 & 0 & 0 & 0\\
 \\
	\eta(1+\beta_1) V^{\top}B^{\top}_{t}U & 0 & 0 & 0 \\
 \\
	0 & 0& 0&0 \\
 \\
	0& 0& 0&0 
	\end{bmatrix}
	+\begin{bmatrix}
	0 & -\eta U^{\top} B_{t} V & 0 & 0\\
 \\
	0 & 0 & -\eta \beta_1 V^{\top}B^{\top}_{t+1}U & 0 \\
 \\
	0 & 0& 0&0 \\
 \\
	0& 0& 0&0 
	\end{bmatrix}\\
 \\
	&+\begin{bmatrix}
	0 & 0 & 0 & 0\\
 \\
	0 & -\eta^2 V^{\top}(A^{\top} B_{t} + B_{t+1}^{\top} A + B_{t+1}^{\top} B_{t}) V & 0 & 0 \\
 \\
	0 & 0& 0&0 \\
 \\
	0& 0& 0&0 
	\end{bmatrix}.
	\end{align*}
	Denote the first matrix at the right side of equation as $ H_1 $, the second one as $ H_2 $ and the third as $ H_3 $. Then we have $ \lVert \Bar{\mathcal{B}}_{t}\lVert_2\le \lVert H_1\lVert_2+\lVert H_2\lVert_2+\lVert H_3\lVert_2 $. 

 By definition, 
	\begin{align*}
	\lVert H_1\lVert_2 =\max \sqrt{\text{Eigenvalue}\{H_1^{\top}H_1\}},\\
	\lVert H_2\lVert_2 =\max \sqrt{\text{Eigenvalue}\{H_2^{\top}H_2\}}.
	%\lVert H_3\lVert =\max \sqrt{\text{eig}\{H_3^{\top}H_3\}}.
	\end{align*}
	Because $ U $ and $ V $ are unitary matrices, 
	\begin{align*}
	&H_1^{\top}H_1=\eta^2(1+\beta_1)^2
	\begin{bmatrix}
	U^{\top}B_{t}B^{\top}_{t}U& 0 & 0 & 0  \\
	0 & 0 & 0 & 0 \\
	0 & 0& 0&0 \\
	0& 0& 0&0
	\end{bmatrix}\\
 \\
	\Rightarrow
	&
	\max \sqrt{\text{Eigenvalue}\{H_1^{\top}H_1\}}=\eta^2(1+\beta_1)^2\max \sqrt{\text{Eigenvalue}\{B_{t}^{\top}B_{t}\}}\\
 \\
	\Rightarrow
	&
	\lVert H_1\lVert_2=\eta^2(1+\beta_1)^2\lVert B_{t}\lVert_2		
	\end{align*}
	and
	\begin{align*}
	&H_2^{\top}H_2=\eta^2
	\begin{bmatrix}
	0 & 0& 0&0 \\
	0& V^{\top}B^{\top}_{t}B_{t}V& 0&0\\
	0& 0 & \beta_1^2 U^{\top}B_{t+1}B^{\top}_{t+1}U & 0  \\
	0 & 0& 0 & 0    
	\end{bmatrix}\\
	\Rightarrow
	&
	\max \sqrt{\text{Eigenvalue}\{H_2^{\top}H_2\}}=\eta^2
	(\max \sqrt{\text{Eigenvalue}\{B_{t}^{\top}B_{t}\}}+ \beta_1^2\max \sqrt{\text{Eigenvalue}\{B_{t+1}^{\top}B_{t+1}\}})  \\
	\Rightarrow
	&
	\lVert H_2\lVert_2=\eta^2(\lVert B_{t}\lVert_2 +\beta_1^2\lVert B_{t+1}\lVert_2)
	\end{align*}
	and
	\begin{align*}
	\lVert H_3\lVert_2 &\le \eta^2\left(\lVert A^{\top} B_{t} \lVert_2  +\lVert B_{t}^{\top} A\lVert_2 +\lVert B_{t}^{\top} B_{t}\lVert_2    \right)\\
 \\
	&\le \eta^2\left(\lVert A\lVert_2  \lVert B_{t} \lVert_2   +\lVert A\lVert_2  \lVert B_{t} \lVert_2 +\lVert B_{t}\lVert_2  \lVert B_{t} \lVert_2  \right)\\
 \\
	&\le c\lVert B_{t} \lVert_2 ,
	\end{align*}
	where $ c=2\lVert A\lVert_2 +c' $, $ c'=\max_{t\ge 0} \lVert B_t\lVert_2  $.  From $ \sum_{t=0}^{\infty} \lVert B_t \lVert_2  < +\infty $ and $ \lVert B_t \lVert_2 \ge 0 $, we know  $c'$ is a bounded constant.
	By combining the bounds for $ H_1 $, $ H_2 $ and $ H_3 $, we have completed the proof for negative momentum method.
\end{proof}

\begin{lem}\label{lm:SVD convergence}
    Assume that there exists a constant $c$ such that $\sum_{t=1}^{\infty} \lVert \Bar{\mathcal{B}}_t \lVert_2 \leq c$,
    % Then the iterative matrix $\Bar{\mathcal{A}}$ is diagonalizable and  its eigenvalues satisfy the following properties:
    % \begin{itemize}
    %     \item If payoff matrix $A$ is non-singular, then the modulus of eigenvalues of these iterative matrices $\mathcal{A}$  are strictly less than $1$. 
    % \item  If payoff matrix $A$ is singular, then $1$ is an eigenvalue of the iterative matrix $\mathcal{A}$,  and other eigenvalues of $\mathcal{A}$ have  modulus strictly less than $1$,
    % \end{itemize}
    and $\Bar{\mathcal{A}}$ is as defined in \eqref{SVD-OGDA}, \eqref{SVD-EG}, or \eqref{SVD-NM},
     then  $ \lVert (\Bar{\mathcal{A}}- I)\Bar{X}_t \lVert_2 $ converges to $0$ with rate $\CO(f(t))$, where
    \begin{align*}
    f(t) = \max\{ \lambda^t, \sum^{\infty}_{i =t/2} \lVert B_i \lVert_2 \}.
\end{align*}
 Here $\lambda \in (0,1)$ is determined by the eigenvalues of the iterative matrix $\Bar{\mathcal{A}}$ of corresponding learning dynamics and the payoff matrix $A$ of the stable game.
\end{lem}
\begin{proof}
Recall that we denote the SVD formulation of iterative process in \eqref{SVD1} \eqref{SVD2} and \eqref{SVD3} as follows:
\begin{align*}
\Bar{X}_{t+1}
= (\Bar{\mathcal{A}}+\Bar{\mathcal{B}}_{t})\Bar{X}_{t},
\end{align*}
Since $\Bar{\mathcal{A}}$ is a diagonalizable matrix from Lemma \ref{prop:diagonal}, thus, there exists an invertible matrix $P$ such that $P \Bar{\mathcal{A}} P^{-1} = D$, where $D$ is a diagonal matrix with the eigenvalues of $\Bar{\mathcal{A}}$ as its entries. Since maximum modulus of eigenvalues of iterative matrix $\Bar{\mathcal{A}}$ is no more than 1, then $\lVert D \lVert_2 \leq 1$.
Let $$\hat{X}_{t}=P \Bar{X}_{t},\ \textnormal{and} \ \ \hat{\mathcal{B}}_t=P \Bar{\mathcal{B}}_{t} P^{-1},$$ then the iterative process becomes $\hat{X}_{t+1}=(D + \hat{\mathcal{B}}_{t})\hat{X}_{t}$. 

By induction,
% \textcolor{red}{(the general form of solution of dynamical system with perturbation)}
we have
\begin{align}
 &\hat{X}_{t}=(D + \hat{\mathcal{B}}_{t-1})\hat{X}_{t-1}
=D^t\hat{X}_{0}+\sum_{l=1}^t D^{t-l}\hat{B}_{l-1} \hat{X}_{l-1} \label{eq: induction Xt}\\
\Longrightarrow
&\hat{\mathcal{B}}_t \hat{X}_{t} =\hat{\mathcal{B}}_tD^t\hat{X}_{0}+\hat{\mathcal{B}}_t\sum_{l=1}^t D^{t-l}\hat{\mathcal{B}}_{l-1} \hat{X}_{l-1}. \notag
\end{align}
Since $\lVert D^{l} \lVert_2 \leq \lVert D \lVert_2^{l} \leq 1$ for any $l \in [t]$, taking norm on both sides, we have
\begin{align*}
\lVert \hat{\mathcal{B}}_t \hat{X}_{t} \lVert_2
\leq \lVert \hat{\mathcal{B}}_t \lVert_2 \lVert \hat{X}_{0} \lVert_2 + \lVert \hat{\mathcal{B}}_t \lVert_2 \sum_{l=1}^t \lVert \hat{\mathcal{B}}_{l-1} \hat{X}_{l-1}  \lVert_2,
\end{align*}
Now we apply  Gronwall inequality, let $u_t = \lVert \hat{\mathcal{B}}_t \hat{X}_t \lVert_2$, $p_t= \lVert \hat{\mathcal{B}}_t \lVert_2 \lVert \hat{X}_{0} \lVert_2$, $q_t= \lVert \hat{\mathcal{B}}_t \lVert_2$ and  $f_t \equiv 1$ in  Gronwall inequality, see Lemma \ref{lm:Gronwall inequality},  then we have
$$
\lVert \hat{\mathcal{B}}_t \hat{X}_{t} \lVert_2
\leq \lVert \hat{\mathcal{B}}_t \lVert_2 \lVert \hat{X}_{0}\lVert_2+  \lVert \hat{\mathcal{B}}_t \lVert_2(\sum_{l=1}^t \lVert \hat{\mathcal{B}}_l \lVert_2 \prod_{k=l-1}^{t-l}(1+\lVert \hat{\mathcal{B}}_l \lVert_2)) \lVert \hat{X}_{0} \lVert_2.
$$
Let $c_1 =\lVert  P \lVert_2  \lVert P^{-1} \lVert_2$. According to the assumption, there exists a constant $c$ such that $\sum_{t=1}^{\infty} \lVert \Bar{\mathcal{B}}_t \lVert_2 \leq c$, then 
\begin{align*}
\sum_{t=1}^{\infty} \lVert \hat{\mathcal{B}}_t \lVert_2
& = \sum_{t=1}^{\infty} \lVert  P \Bar{\mathcal{B}}_t P^{-1} \lVert_2 \\
&\leq  \sum_{t=1}^{\infty} \lVert  P \lVert_2
\lVert \Bar{\mathcal{B}}_t \lVert_2  \lVert P^{-1} \lVert_2\\
&\le  c_1 c.
\end{align*}
% \textcolor{red}{change a constant c}\\
Note that $\hat{\mathcal{B}}_t=P \Bar{\mathcal{B}}_{t} P^{-1}$, so $\lVert \hat{\mathcal{B}}_t \lVert_2\le \lVert P\lVert_2 \lVert \Bar{\mathcal{B}}_{t}\lVert_2 \lVert P^{-1}\lVert_2 \le c_1 \lVert \Bar{\mathcal{B}}_{t}\lVert_2$. Since $e^x \ge 1+x $ for $x \in \BR$, we obtain 
\begin{align*}
    \prod_{t=1}^\infty (1+\lVert \hat{\mathcal{B}}_t \lVert_2)& \leq \prod_{t=1}^\infty e^{\lVert \hat{\mathcal{B}}_t \lVert_2} \\
    &= e^{\sum_{t=1}^{\infty} \lVert \hat{\mathcal{B}}_t \lVert_2} \\
    &\leq e^{c_1 c}. 
\end{align*}

Let $c_2=(1+c_1 c \cdot e^{c_1 c})\lVert \hat{X}_{0} \lVert_2$, then
\begin{align*}
\lVert \hat{\mathcal{B}}_t \hat{X}_{t} \lVert_2
& \leq 
\lVert \hat{\mathcal{B}}_t \lVert_2 \lVert \hat{X}_{0}\lVert_2(1+\sum_{l=1}^t \lVert \hat{\mathcal{B}}_l \lVert_2 \prod_{k=l-1}^{t-l}(1+\lVert \hat{\mathcal{B}}_l \lVert_2))\\
& \leq \lVert \hat{\mathcal{B}}_t \lVert_2 \lVert \hat{X}_{0}\lVert_2(1+ e^{c_1 c}\sum_{l=1}^t \lVert \hat{\mathcal{B}}_l \lVert_2)\\
& \leq \lVert \hat{\mathcal{B}}_t \lVert_2 \lVert \hat{X}_{0}\lVert_2(1+  c_1 c \cdot e^{c_1 c})\\
\\
& \leq c_2 \lVert \hat{\mathcal{B}}_t \lVert_2.
\end{align*}
Multiplying $(D-I)$ on the equality \eqref{eq: induction Xt}of both sides, we have
$$
(D-I)  \hat{X}_{t}
=(D-I)  D^t\hat{X}_{0}+
 \sum_{l=1}^t (D-I) D^{t-l}\hat{\mathcal{B}}_{l-1} \hat{X}_{l-1}.\\
$$
Let $c_3 = \max\{c_2,c_2\sum_{l=1}^{\frac{1}{2}t}  \lVert \hat{\mathcal{B}}_l \lVert_2 \}$, taking the norm on both sides, we have
\begin{align*}
\lVert (D-I)  \hat{X}_{t} \lVert_2
&\leq \delta^t \lVert \hat{X}_{0}\lVert_2+ \sum_{l=1}^t \delta^{t-l} \lVert \hat{\mathcal{B}}_{l-1}\hat{X}_{l-1} \lVert_2
\leq \delta^t \lVert \hat{X}_{0}\lVert_2+ \sum_{l=1}^t c_2 \delta^{t-l} \lVert \hat{\mathcal{B}}_l \lVert_2 \\
&\leq \delta^{\frac{1}{2}t}(\lVert \hat{X}_{0}\lVert_2 + c_2\sum_{l=1}^{\frac{1}{2} t}  \lVert \hat{\mathcal{B}}_l \lVert_2) + c_2 \sum_{l=\frac{1}{2} t}^t \lVert \hat{\mathcal{B}}_l \lVert_2\\
&\leq c_3 f(t).
\end{align*}

Let $\lambda = \delta^{\frac{1}{2}}$, recall that $ f(t) = \max\{ \lambda^t, \sum^{\infty}_{i =t/2} \lVert B_i \lVert_2 \}$. The last inequality is due to Lemma \ref{lm:norm bounded equivalence}, we can see that  there is constant $c_4$ such that $\sum^{\infty}_{i =t/2} \lVert \hat{\mathcal{B}}_i \lVert_2  \leq c_4 \sum^{\infty}_{i =t/2} \lVert B_i \lVert_2 $.
Recall that $ f(t) = \max\{ \lambda^t, \sum^{\infty}_{i =t/2} \lVert B_i \lVert_2 \}$. 
% \\
% \textcolor{red}{inequality not sure}.\\
Then, there exists a constant $c_5$ such that
$$
\lVert (\Bar{\mathcal{A}}-I)\Bar{X}_{t} \lVert_2  \leq \lVert  P^{-1} \lVert_2  \lVert (D-I)  \hat{X}_{t} \lVert_2  \leq c_5 f(t). 
$$
 \end{proof}

\begin{lem}\label{lm:convergence equivalent}
    If $\lVert (\Bar{\mathcal{A}}- I)\Bar{X}_t \lVert_2$ converges to 0  with rate $\CO(f(t))$ as $t$ tends to infinity, then for  OGD, EG and negative momentum method, $\lVert A^{\top}x_t  \lVert_2 +\lVert A y_t \lVert_2$ converges to 0 with rate $\CO(f(t))$ when $t$ tends to infinity .
\end{lem}
\begin{proof}[Proof of Lemma \ref{lm:convergence equivalent}]
We break the proof into three parts. Recall that $A = U \Sigma_A V^{\top}$, and $h =\lVert \Sigma_A \lVert_2 $. 

Firstly we prove the lemma for OGDA.

\textbf{OGDA}
Writing $(\Bar{\mathcal{A}}-I)\Bar{X}_{t}$ into matrix form:
\begin{align*}
\begin{bmatrix}
0 & -2\eta \Sigma_A & 0 & \eta \Sigma_A \\
\\
2\eta \Sigma_A^{\top} & 0 & -\eta \Sigma_A^{\top} & 0\\
\\
I & 0& -I&0 \\
\\
0& I& 0&-I
\end{bmatrix} 
\begin{bmatrix}
U^{\top} x_{t-1} \\
\\
V^{\top} y_{t-1} \\
\\
U^{\top} x_{t-2} \\
\\
V^{\top} y_{t-2}
\end{bmatrix}
= \begin{bmatrix}
-2\eta \Sigma_A V^{\top} y_{t-1} + \eta \Sigma_A V^{\top} y_{t-2}\\
\\
2\eta \Sigma_A U^{\top} x_{t-1} - \eta \Sigma_A U^{\top} x_{t-2}\\
\\
U^{\top}x_{t-1} -U^{\top}x_{t-2} \\
\\
V^{\top}y_{t-1}-V^{\top}y_{t-2}
\end{bmatrix}
\end{align*}
Since there is a constant $c$ such that$\lVert(\Bar{\mathcal{A}}-I)\Bar{X}_{t} \lVert_2\leq c f(t) $, then 
\begin{align*}
\lVert 2\eta \Sigma_A U^{\top} x_{t-1} - \eta \Sigma_A U^{\top} x_{t-2} \lVert_2 \leq c f(t),\\
\\
\lVert U^{\top}x_{t-1} -U^{\top}x_{t-2} \lVert_2 \leq c f(t). 
\end{align*}
Using these two inequalities to bound $\lVert A^{\top}  x_{t}\lVert_2$, we have
 \begin{align*}
  &\lVert \eta \Sigma_A U^{\top} x_{t-1}\lVert_2\\
  \\
  =&\lVert 2\eta \Sigma_A U^{\top} x_{t-1} - \eta \Sigma_A U^{\top} x_{t-2} - \eta \Sigma_A (U^{\top} x_{t-1} - U^{\top} x_{t-2})\lVert_2\\
  \\
  \leq& \lVert 2\eta \Sigma_A U^{\top} x_{t-1} - \eta \Sigma_A U^{\top} x_{t-2}\lVert_2 + \lVert \eta \Sigma_A (U^{\top} x_{t-1} - U^{\top} x_{t-2})\lVert_2\\
  \\
  \leq & c f(t) 
  + \eta c f(t).\\
 \end{align*}
 Since $A^{\top}  x_{t}=V  \Sigma_A U^{\top} x_{t}$, then
 \begin{align*}
 \lVert A^{\top}  x_{t}\lVert_2
 &=\lVert V  \Sigma_A U^{\top} x_{t}\lVert_2 \\
 \\
 & \leq \lVert V \lVert_2 \lVert \Sigma_A U^{\top} x_{t}\lVert_2 \\
 \\
 & \leq \frac{ c h (1+\eta)f(t)}{\eta},
 \end{align*}
 where the last inequality is due to $ \lVert \Sigma_A \lVert_2 = h$ and $V$, $U$ are unitary matrices.
 Similarly, we can obtain $\lVert A y_{t}\lVert_2 \leq \frac{(1+\eta)c h f(t) }{\eta}$.

Next, we prove the lemma for extra-gradient.

\textbf{EG}
Writing $(\Bar{\mathcal{A}}-I)\Bar{X}_{t}$ into matrix form:
\begin{align*}
\begin{bmatrix}
-\alpha \gamma  \Sigma_A\Sigma_A^{\top} & -\alpha \Sigma_A  \\
\\
\alpha  \Sigma_A ^{\top} & -\alpha \gamma\Sigma_A^{\top} \Sigma_A 
\end{bmatrix}
\begin{bmatrix}
U^{\top} x_{t-1} \\
\\
V^{\top} y_{t-1} \\
\end{bmatrix}
= \begin{bmatrix}
-\alpha \gamma \Sigma_A \Sigma_A^{\top}  U^{\top} x_{t-1} - \alpha \Sigma_A V^{\top} y_{t-1}\\
\\
\alpha  \Sigma_A ^{\top} U^{\top} x_{t-1} - \alpha \gamma\Sigma_A^{\top} \Sigma_A V^{\top} y_{t-1}.
\end{bmatrix}
\end{align*}
Since there is a constant $c$ such that $\lVert(\Bar{\mathcal{A}}-I)\Bar{X}_{t}\lVert_2 \leq c f(t) $, then
\begin{align*}
\lVert - \gamma \Sigma_A \Sigma_A^{\top}  U^{\top} x_{t-1} -  \Sigma_A V^{\top} y_{t-1}\lVert_2  \leq \frac{c f(t)}{\alpha},\\
\lVert 
 \Sigma_A ^{\top} U^{\top} x_{t-1} - \gamma\Sigma_A^{\top} \Sigma_A V^{\top} y_{t-1}\lVert_2 \leq \frac{c f(t)}{\alpha}.
\end{align*}
Using these two inequalities to bound
 $\lVert A^{\top} x_{t}\lVert_2$, we have
 \begin{align*}
  &\lVert (\gamma^2  \Sigma_A^{\top}\Sigma_A+I)\Sigma_A^{\top}U^{\top} x_{t-1}\lVert_2 \\
  \\
  =&\lVert \Sigma_A ^{\top} U^{\top} x_{t-1} - \gamma\Sigma_A^{\top} \Sigma_A V^{\top} y_{t-1} -  \gamma\Sigma_A^{\top} \left(- \gamma \Sigma_A \Sigma_A^{\top}  U^{\top} x_{t-1} -  \Sigma_A V^{\top} y_{t-1}
  \right)\lVert_2\\
  \\
  \leq& \lVert \Sigma_A ^{\top} U^{\top} x_{t-1} - \gamma\Sigma_A^{\top} \Sigma_A V^{\top} y_{t-1} \lVert_2 +\gamma \lVert  \Sigma_A^{\top}  \lVert_2 \lVert - \gamma \Sigma_A \Sigma_A^{\top}  U^{\top} x_{t-1} -  \Sigma_A V^{\top} y_{t-1} \lVert_2\\
  \\
  \leq & \frac{(1+\gamma h) c f(t)}{\alpha}.
 \end{align*}
Since matrix $\gamma^2  \Sigma_A^{\top}\Sigma_A+I$ is invertible, then
\begin{align*}
\lVert \Sigma_A U^{\top} x_{t}\lVert_2
=&\lVert(\gamma^2  \Sigma_A^{\top}\Sigma_A+I)^{-1}(\gamma^2  \Sigma_A^{\top}\Sigma_A+I)\Sigma_A^{\top}U^{\top} x_{t-1}\lVert_2 \\
\\
\leq& \lVert (\gamma^2  \Sigma_A^{\top}\Sigma_A+I)^{-1} \lVert_2 
\lVert (\gamma^2  \Sigma_A^{\top}\Sigma_A+I)\Sigma_A^{\top}U^{\top} x_{t-1}\lVert_2\\
\\
\leq &\lVert (\gamma^2  \Sigma_A^{\top}\Sigma_A+I)\Sigma_A^{\top}U^{\top} x_{t-1}\lVert_2\\
\\
\leq& \frac{(1+\gamma h)c f(t)}{\alpha},
\end{align*}
where the last inequality is due to $\lVert (\gamma^2  \Sigma_A^{\top}\Sigma_A+I)^{-1} \lVert_2 \leq 1$. Since $A^{\top}  x_{t}=V  \Sigma_A U^{\top} x_{t}$, then
 \begin{align*}
 \lVert A^{\top}  x_{t}\lVert_2
 &=\lVert V  \Sigma_A U^{\top} x_{t}\lVert_2
 \leq \lVert V \lVert_2 \lVert \Sigma_A U^{\top} x_{t}\lVert_2
 \leq \frac{(1+\gamma h)c f(t)}{\alpha},
 \end{align*}
 where the last inequality is due to $V$ is  unitary matrices.
 Similarly, we can obtain 
 $$\lVert A y_{t}\lVert_2 \leq \frac{(1+\gamma h)}{\alpha}f(t) c.
 $$

Finally, we prove the lemma for negative momentum method.
 
\textbf{Negative Momentum Method}
Writing $(\Bar{\mathcal{A}}-I)\Bar{X}_{t}$ into matrix form:
\begin{align*}
&\begin{bmatrix}
\beta_1I & -\eta \Sigma_A  & -\beta_1 I & 0  \\
\\
\eta(1+\beta_1) \Sigma_A ^{\top} & - \eta^2  \Sigma_A^{\top}\Sigma_A & -\eta\beta_1\Sigma_A ^{\top} & -\beta_2 I\\
\\
I & 0& -I&0 \\
\\
0& I& 0&-I 
\end{bmatrix}
\begin{bmatrix}
U^{\top} x_{t-1} \\
\\
V^{\top} y_{t-1} \\
\\
U^{\top} x_{t-2} \\
\\
V^{\top} y_{t-2}
\end{bmatrix}
\\
&= \begin{bmatrix}
\beta_1 U^{\top} x_{t-1}-\eta \Sigma_A V^{\top} y_{t-1} -\beta_1 U^{\top}x_{t-2}\\
\\
\eta(1+\beta_1)\Sigma_A^{\top} U^{\top} x_{t-1}-\eta^2 \Sigma_A^{\top}\Sigma_A V^{\top} y_{t-1} -\eta\beta_1\Sigma_A ^{\top}U^{\top} x_{t-2}  - \beta_2 V^{\top} y_{t-2} \\
\\
U^{\top} x_{t-1}-U^{\top}x_{t-2} \\
\\
V^{\top} y_{t-1}-V^{\top}y_{t-2}
\end{bmatrix}
\end{align*}
Since there is a constant $c$ such that $\lVert(\Bar{\mathcal{A}}-I)\Bar{y}_{t} \lVert_2\leq c f(t) $, then 
\begin{align*}
\lVert \beta_1 U^{\top} x_{t-1}-\eta \Sigma_A V^{\top} y_{t-1} -\beta_1 U^{\top}x_{t-2}\lVert_2  \leq c f(t),\\
\\
\lVert 
 U^{\top} x_{t-1}-U^{\top}x_{t-2}\lVert_2 \leq c f(t).
\end{align*}
 Using these two inequalities to bound
 $\lVert A y_{t}\lVert_2$, we have
 \begin{align*}
    &\lVert \eta \Sigma_A V^{\top} y_{t-1} \lVert_2 \\
    \\
  =&\lVert \beta_1\left( U^{\top} x_{t-1}-U^{\top}x_{t-2}\right)-\left(\beta_1 U^{\top} x_{t-1}-\eta \Sigma_A V^{\top} y_{t-1} -\beta_1 U^{\top}x_{t-2}\right)\lVert_2\\
  \\
  \leq & \lVert \beta_1\left( U^{\top} x_{t-1}-U^{\top}x_{t-2}\right) \lVert_2
  + \lVert\beta_1 U^{\top} x_{t-1}-\eta \Sigma_A V^{\top} y_{t-1} -\beta_1 U^{\top}x_{t-2}\lVert_2\\
  \\
  \leq & \beta_1 c f(t) +c f(t).
 \end{align*}
  Since $A  y_{t}=U \Sigma_A V^{\top} y_{t}$, then
 \begin{align*}
 \lVert A  y_{t}\lVert_2
 &=\lVert  U \Sigma_A V^{\top}  y_{t}\lVert_2
 \leq \lVert  U^{\top}\lVert_2 \lVert \Sigma_A V^{\top}  y_{t}\lVert_2
 \leq \frac{ c(1+\beta_1)f(t)}{\eta},
 \end{align*}
 where the last inequality is due to $\lVert U \lVert_2 =1$.

 We also have 

 \begin{align*}
     \lVert \eta(1+\beta_1)\Sigma_A^{\top} U^{\top} x_{t-1}-\eta^2 \Sigma_A^{\top}\Sigma_A V^{\top} y_{t-1} -\eta\beta_1\Sigma_A ^{\top}U^{\top} x_{t-2}  - \beta_2 V^{\top} y_{t-2} \lVert_2 \leq c f(t) %and \lVert V^{\top} y_{t-1}-V^{\top}y_{t-2} \lVert_2  \leq c %f(t).
 \end{align*}
 and

  \begin{align*}
 \lVert V^{\top} y_{t-1}-V^{\top}y_{t-2} \lVert_2  \leq c f(t).
 \end{align*}
 
 Then, 
\begin{align*}
    &\lVert \eta \Sigma_A \Sigma_A^{\top} U^{\top} x_{t-1}\lVert_2\\
    \\
=&\lVert \Sigma_A \cdot
( [\eta(1+\beta_1)\Sigma_A^{\top} U^{\top} x_{t-1}-\eta^2 \Sigma_A^{\top}\Sigma_A V^{\top} y_{t-1} -\eta\beta_1\Sigma_A ^{\top}U^{\top} x_{t-2} 
  - \beta_2 V^{\top} y_{t-2} ] \\
 & - \eta \Sigma_A^{\top}\left(\beta_1 U^{\top} x_{t-1}-\eta \Sigma_A V^{\top} y_{t-1} -\beta_1 U^{\top}x_{t-2}\right) + \beta_2 \Sigma_A^{\top} V^{\top} y_{t-2} ) \lVert_2\\
  \\
\leq & \lVert \Sigma_A \lVert_2 c f(t) + \lVert \Sigma_A \lVert_2 c f(t) + \frac{\beta_2(1+\beta_1)}{\eta} c f(t)\\
\\
\leq & \left(2h + \frac{\beta_2(1+\beta_1)}{\eta}\right) c f(t)
\end{align*}
Since  $\lVert \Sigma_A \lVert_2 =h$, then
$\lVert \Sigma_A^{\top} U^{\top} x_{t}\lVert_2 \leq \frac{2 \eta h + \beta_2(1+\beta_1)}{\eta^2 h} c f(t)$.
Since $A^{\top}  x_{t}=V  \Sigma_A U^{\top} x_{t}$, then
 \begin{align*}
 \lVert A^{\top}  x_{t}\lVert_2
 &=\lVert V  \Sigma_A U^{\top} x_{t}\lVert_2 \\
 &\leq \lVert V \lVert_2 \lVert \Sigma_A U^{\top} x_{t}\lVert_2 \\
 & \leq \frac{2 \eta h + \beta_2(1+\beta_1)}{\eta^2 h} c f(t).
 \end{align*}
\end{proof}

Now we are ready to prove Theorem \ref{thm:Condition-Pertub}.

\begin{proof}[proof of Theorem \ref{thm:Condition-Pertub}]
 % We also write the iterative process of three learning dynamics and denote $\Bar{\mathcal{A}}$ as the iterative matrix if payoff matrix is time invariant and equal to the matrix which is the singular value decomposition of $A$. 
 %The assumptions of Theorem \ref{thm:Condition-Pertub} require parameters of three algorithms satisfy some conditions, under these conditions, by Lemma \ref{lm:eig-static}, eigenvalues of iterative matrix $\Bar{\mathcal{A}}$ consist of 1 and other eigenvalues whose  modulus are less than 1. 
 %In addition, by Lemma \ref{prop:diagonal}, iterative matrix $\Bar{\mathcal{A}}$ is diagonalizable. 
 %We also give Lemma \ref{lm:norm bounded equivalence} showing that \ref{bap} for $\{B_t\}_t$ holds implies \ref{bap} for $\{\hat{\mathcal{B}_t}\}_t$ holds.
 
 According to Lemma \ref{lm:norm bounded equivalence}, assumptions of Lemma \ref{lm:SVD convergence} have been satisfied by the difference equations associated to our learning dynamics, thus we have $\lVert  (\Bar{\mathcal{A}}- I)\Bar{X}_t \lVert_2$ converges to 0 with rate $f(t)$. Moreover, by Lemma \ref{lm:convergence equivalent} , $\Delta_t$ converges to 0 with rate $f(t)$ in OGD, EG and negative momentum method. We complete the proof. 
\end{proof}

 \section{Omitted Proofs from Theorem \ref{thm:EG-Pertub}}\label{extra_convergent}
 \EGPertub*
%  The Extra gradient method in this case can be written as 
% \begin{align*}\label{Time varying Extra Gradient Descent Ascent}
% \tag{Time varying Extra GDA}
% &x_{t+\frac{1}{2}} = x_t - \gamma_1 A_ty_t, \ \ y_{t+\frac{1}{2}} = y_t + \gamma_2 A_t^{\top}x_t, \\
% &x_{t+1} = x_t - \alpha_1 A_ty_{t+\frac{1}{2}},\ \  y_{t+1} = y_t + \alpha_2 A_t^{\top}x_{t+\frac{1}{2}}.
% \end{align*}
% where $ \lim_{t \to \infty} A_t=A $. We can write this as the following equation when assuming $\alpha_1 = \alpha_2 = \alpha$ and $\gamma_1 = \gamma_2 = \gamma$.
% \begin{align}\label{extramatrix}
% \begin{bmatrix}
% x_{t+1}\\
% \\
% y_{t+1}
% \end{bmatrix}
% =
% \begin{bmatrix}
% I - \alpha\gamma A_{t}A_{t}^{\top}  & -\alpha A_{t}\\
% \\
% \alpha A_{t}^{\top} & I - \alpha\gamma A_{t}^{\top}A_{t}
% \end{bmatrix}
% \begin{bmatrix}
% x_{t}\\
% \\
% y_{t}
% \end{bmatrix}.
% \end{align}

Recall that Extra Gradient satisfies the linear difference equation \eqref{eq:EG}, denote the iterative matrix in equation \eqref{eq:EG} 
with  payoff matrix $A_t$ as $\mathcal{A}_t$.
 According to the convergence of payoff matrix, we have $\lim_{t\rightarrow \infty}A_t=A$. Let $B_t=A_t-A$, then we have $\lim_{t\rightarrow \infty}B_t=0$. Denote $\mathcal{A}$ as the iterative matrix when payoff matrix is time invariant and equal to $A$. Let 
$$\mathcal{B}_t=\mathcal{A}_t-\mathcal{A}.$$
  To prove the theorem, we first establish several necessary lemmas.
\begin{lem}
    Given $\lim_{t\rightarrow \infty} B_t=0$, we have $\lim_{t\rightarrow \infty} \mathcal{B}_t=0$.
\end{lem}
\begin{proof}
Recall that 
\begin{align*}
\mathcal{A}_t=
    \left[
\begin{array}{cc}
 I - \alpha\gamma A_tA_t^{\top}  & -\alpha A_t\\
 \\
\alpha A_t^{\top} & I - \gamma\alpha A_t^{\top}A_t
\end{array}
\right]
\end{align*}
and 
\begin{align*}
    \mathcal{A}=\left[
\begin{array}{cc}
 I - \alpha\gamma AA^{\top}  & -\alpha A\\
 \\
\alpha A^{\top} & I - \gamma\alpha A^{\top}A
\end{array}
\right],
\end{align*}
we can obtain that
\begin{align*}
{\mathcal{B}}_{t}&=
\begin{bmatrix}
-\alpha \gamma (A B_{t}^{\top} + B_{t} A^{\top} + B_{t} B_{t}^{\top}) & -\alpha B_{t}  \\
\\
\\
\alpha  B_{t}^{\top}  &\alpha \gamma (A^{\top} B_{t} + B_{t}^{\top} A + B_{t}^{\top} B_{t}) 
\end{bmatrix}\\
    &=\begin{bmatrix}
-\alpha \gamma (A B_{t}^{\top} + B_{t} A^{\top} + B_{t} B_{t}^{\top}) & 0  \\
\\
\\
 0 & 0
\end{bmatrix}+
\begin{bmatrix}
  0& 0 \\
\\
\\
 0 &\alpha \gamma (A^{\top} B_{t} + B_{t}^{\top} A + B_{t}^{\top} B_{t}) 
\end{bmatrix}\\&+
\begin{bmatrix}
0& -\alpha B_{t}  \\
\\
\\
0  &0
\end{bmatrix}+
\begin{bmatrix}
0& 0 \\
\\
\\
\alpha B_{t}^\top  &0
\end{bmatrix}.
\end{align*}
We separate ${\mathcal{B}}_{t}$ into four matrices and denote these matrices in right side of the equation as $H_1$, $H_2$ ,  $H_3$ and $H_4$, respectively. Then 
\begin{align*}
    \lVert \mathcal{B}_t \lVert_2\le \lVert H_1 \lVert_2 +\lVert H_2 \lVert_2+\lVert H_3 \lVert_2+\lVert H_4 \lVert_2.
\end{align*}
    Since $\lim_{t\rightarrow \infty} B_t=0$ , then $\lim_{t\rightarrow\infty} \lVert B_t\lVert_2=0$, so we can yield that there exists $c$ such that $\lVert B_t\lVert_2\le c$ for any t. We also assume that $c_1=\lVert A\lVert_2$.
    
    Then we have 
    \begin{align*}
        \lVert H_1 \lVert_2
        &=\alpha \gamma \lVert A B_{t}^{\top} + B_{t} A^{\top} + B_{t} B_{t}^{\top} \lVert_2\\
         \\
	&\leq \alpha \gamma\left(\lVert A \lVert_2 \lVert B_{t} \lVert_2 + \lVert A \lVert_2 \lVert B_{t} \lVert_2 +  \lVert B_{t} \lVert_2  \lVert B_{t} \lVert_2\right) \\
 \\
	&\leq \alpha \gamma (2c_1+c) \lVert B_{t} \lVert_2,
    \end{align*}
    Similarly, $\lVert H_2 \lVert_2 \leq \alpha\gamma(2c_1+c) \lVert B_{t} \lVert_2$.
	In addition, $\lVert H_3 \lVert_2 = \lVert H_4 \lVert_2 = \alpha \lVert B_{t} \lVert_2$.

 Then we can obtain that there exists a constant $c_2$, such that $\lVert \mathcal{B}_t \lVert_2\le c_2\lVert {B}_t \lVert_2$, which implies that $\lim_{t\rightarrow \infty} \mathcal{B}_t =0$.
\end{proof}
  With the lemma above, we directly utilize $\lim_{t\rightarrow \infty} \mathcal{B}_t =0$ in proving Theorem $\ref{thm:EG-Pertub}$.
\begin{lem}\label{extramono}
	Let  $ X_t=(x_t^\top,y_t^\top)^\top $, then there exists $t_0>0$, such that when $t>t_0$, $\norm{X_t}_2 $  is monotonically non-increasing. Moreover, $ \exists c_0 \ge 0,\ \lim_{t \to \infty}\norm{X_t}_2=c_0 $. %where $ c_0\ge 0 $.
\end{lem}
\begin{proof}
First we prove that there exists $t_0$ such that when $t>t_0$, $\lVert \mathcal{A}_t\lVert_2\le 1$.

    From the proof of Lemma \ref{lm:eig-static}, we know that if $\alpha=\gamma\le \frac{1}{\sigma_t}$, where $\sigma_t$ is the maximal singular value of payoff matrix $A_t$, then the discriminant in equation (\ref{eq:discriminant-of-EG}) is satisfied, i.e., $\lVert \mathcal{A}_t\lVert_2\le 1$. Here we choose $\alpha=\gamma\le \frac{1}{2\sigma}$, where $\sigma$ is the maximal singular value of payoff matrix $A$. Because $A_t$ converges to $A$, we can conclude that there exists $t_0$, such that when $t\ge t_0$, $\sigma_t<2\sigma$. This implies that $\alpha=\gamma\le \frac{1}{2\sigma}< \frac{1}{\sigma_t}$, which means $\lVert \mathcal{A}_t\lVert_2\le 1$. Therefore we prove that here exists $t_0$ such that when $t>t_0$, $\lVert X_t\lVert_2\le \lVert X_{t-1}\lVert_2$. For any $t>t_0$, we have
    \begin{align*}
        \lVert X_t\lVert_2&=\lVert \mathcal{A}_{t-1} X_{t-1}\lVert_2\\
        & \le \lVert \mathcal{A}_{t-1}\lVert_2 \lVert X_{t-1}\lVert_2\\
        & \le \lVert X_{t-1}\lVert_2.
    \end{align*}
    % where 
    Therefore, we have that when $t>t_0$, $\lVert X_t\lVert_2$  is monotonically non-increasing.

    From the fact that $\lVert X_t\lVert_2$  is monotonically non-increasing and no smaller than 0, we obtain $ \exists c_0 \ge 0,\ \lim_{t \to \infty}\norm{X_t}_2=c_0 $.
\end{proof}
% From the above lemma we know that there exists $t_0$, when $t>t_0$ the largest norm of eigenvalue of iterative matrix is not larger than 1.
In fact, the property that $\lVert X_t\lVert_2$ is monotonically non-increasing is closely related to the iterative matrix of EG is normal, which causes part of the difference between EG and OGDA or negative momentum method.

% \begin{thm}
% 	$ \lim_{t \to \infty} (A^{\top}x_t,A y_t) = (\boldsymbol{0},\boldsymbol{0}) \in \BR^{2n}$ in \eqref{Time varying Extra Gradient Descent Ascent}. 
% 	Moreover, if there exists $c$ such that for any $\sum_{t=0}^{\infty} \lVert B_t \lVert \leq c$, then we have convergence rate \textcolor{red}{$f(t)$}.
% \end{thm}
\begin{lem}\label{extracompression}
	Decompose $ \BR^{n+m}=V_1\oplus V_2 $ where $ V_1 $ is the eigenspace of eigenvalue $ 1 $ of matrix $ \mathcal{A} $, $ V_1 $ and $ V_2 $ are mutually perpendicular. Define $ \lambda=\max_{s\neq 1, s \in \text{Eigenvalue} \mathcal{A}}\left\vert{s}\right\vert  $. Then if $  v\in V_2 $, $ \norm{\mathcal{A} v}_2\le \lambda \norm{v}_2 $. 
	%		 because $ \mathcal{A} $ is normal.	
\end{lem}
\begin{proof}
	Let $ W_s=\{v\in \BR^{n+m}\mid \mathcal{A}v=sv\} $, that is $ W_s $ is the eigenspace of eigenvalue $ s $ of $ \mathcal{A} $. Let $ V_1=W_1 $ and $ V_2= \oplus_{s\neq 1} W_s $. By $ \mathcal{A} $ is normal, we have $ \BR^{n+m}=V_1\oplus V_2 $. $ V_1 $ and $ V_2 $ are mutually perpendicular.
	
	 Then we only need to prove if $  v\in V_2 $, $ \norm{\mathcal{A} v}_2\le \lambda \norm{v}_2 $. From Lemma \ref{lm:eig-static}, we know that $ \lambda<1 $. Because 
  \begin{align*}
  V_2= \oplus_{s\neq 1, s  \in \text{Eigenvalue} \mathcal{A}} W_s ,
  \end{align*}  
  $ v $ can be decomposed as $ v=\sum_{s\neq 1 , s \in \text{Eigenvalue}\mathcal{A}} k_sw_s $, where $ w_s\in W_s$, $ k_s $ is the coefficient and for different $ s_1 $ and $ s_2 $ which are eigenvalues of $\mathcal{A}$, $ w_{s_1} $ and $ w_{s_2} $ are perpendicular. 
  
  Therefore $ \norm{v}_2^2=\sum_{s\neq 1,  s \in \text{Eigenvalue} \mathcal{A}} k_s^2\norm{w_s}_2^2 $. Then $ \mathcal{A}v=\sum_{s\neq 1,  s \in \text{Eigenvalue} \mathcal{A}} k_s\mathcal{A}w_s $, therefore we have
	\begin{align*}
	\norm{\mathcal{A}v}_2^2 
	&=\sum_{s\neq 1, s \in \text{Eigenvalue} \mathcal{A}} k_s^2\norm{\mathcal{A}w_s}_2^2\\
	&=\sum_{s\neq 1, s \in \text{Eigenvalue} \mathcal{A}} {\left\vert{s}\right\vert}^2 k_s^2\norm{w_s}_2^2\\
	&\le \lambda^2\sum_{s\neq 1, s \in \text{Eigenvalue} \mathcal{A}} k_s^2\norm{w_s}_2^2\\
	&=\lambda^2\norm{v}_2^2
	\end{align*}
	which means that $ \norm{\mathcal{A} v}_2\le \lambda \norm{v}_2 $, this complete the proof.
\end{proof}

Now we can decompose $ X_t=v_t^1+v_t^2 $ where $ v_t^1\in V_1 $ and $ v_t^2\in V_2 $. Similarly, we also decompose $\mathcal{B}_t X_t=w_t^1+w_t^2 $ where $ w_t^1\in V_1 $ and $ w_t^2\in V_2 $. 

\begin{lem}\label{lm:EG-Perturb3}
	If  $ \lim_{t \to \infty}\norm{v_t^2}_2=0 $, then $ \lim_{t \to \infty} (A^{\top}x_t,A y_t) = (\boldsymbol{0},\boldsymbol{0})$, which implies that $\lim_{t\rightarrow \infty}\Delta_t=0$. 
\end{lem}
\begin{proof}
	Let $ v_t^1=\left(\begin{array}{l}
x_t^1 \\
y_t^1
\end{array}\right)
 $, $ v_t^2=\left(\begin{array}{l}
x_t^2 \\
y_t^2
\end{array}\right) $, then by $ X_t=\left(\begin{array}{l}
x_t \\
y_t
\end{array}\right)=v_t^1+v_t^2 $, we have $ x_t=x_t^1+x_t^2 $ and $ y_t=y_t^1+y_t^2 $. 
First, we prove $ A^{\top}x_t^1 = 0$ and $A y_t^1=0 $. By $ v_t^1\in V_1 $, we have
\begin{align*}
\begin{bmatrix}
I - \alpha\gamma AA^{\top}  & -\alpha A\\
\\
\alpha A^{\top} & I - \alpha\gamma A^{\top}A
\end{bmatrix}
\begin{bmatrix}
x_t^1\\
\\
y_t^1
\end{bmatrix}
=
\begin{bmatrix}
x_t^1\\
\\
y_t^1
\end{bmatrix}
\end{align*}
that is 
\begin{align*}
&\begin{bmatrix}
x_t^1 - \alpha\gamma AA^{\top}x_t^1-\alpha Ay_t^1\\
\\
\alpha A^{\top}x_t^1 + y_t^1 - \alpha\gamma A^{\top}Ay_t^1
\end{bmatrix}
=
\begin{bmatrix}
x_t^1\\
\\
y_t^1
\end{bmatrix}\\
\\
\Longrightarrow 
&\begin{cases}
-\gamma AA^{\top}x_t^1-Ay_t^1=0,\\
A^{\top}x_t^1-\gamma A^{\top}Ay_t^1=0,
\end{cases}\\
\\
\Longrightarrow 
&\begin{cases}
A^{\top}x_t^1=0,\\ 
Ay_t^1=0, 
\end{cases}\\
\end{align*}
where the second double arrow symbols is due to $\gamma^2AA^{\top}+I$ is invertible.

According to $A^{\top}x_t^1=0$ and $A y_t^1=0$, 
we have
\begin{align*}
\begin{bmatrix}
A^{\top}x_t\\
A y_t
\end{bmatrix}
=
\begin{bmatrix}
A^{\top} & \\
&A
\end{bmatrix}
X_t & =\begin{bmatrix}
A^{\top} & \\
&A
\end{bmatrix}
(v_t^1+v_t^2) \\
&=
\begin{bmatrix}
A^{\top} & \\
&A
\end{bmatrix}
\begin{bmatrix}
x_t^1\\
y_t^1
\end{bmatrix}+
\begin{bmatrix}
A^{\top} & \\
&A
\end{bmatrix}
v_t^2\\
&=
\begin{bmatrix}
A^{\top} x_t^1\\
A y_t^1
\end{bmatrix}+
\begin{bmatrix}
A^{\top} & \\
&A
\end{bmatrix}
v_t^2 \\
&=
\begin{bmatrix}
A^{\top} & \\
&A
\end{bmatrix}
v_t^2.
\end{align*}
We can see that if $ \lim_{t \to \infty}\norm{v_t^2}_2=0 $, then $ \lim_{t \to \infty} (A^{\top}x_t,A y_t) = (0,0)$ .
	%then we prove that $ (A^{\top}x_t^1,A y_t^1)=0 $.
\end{proof}

Now we are ready to prove Theorem \ref{thm:EG-Pertub}.

\begin{proof}[proof of Theorem \ref{thm:EG-Pertub}]
 According to Lemma \ref{lm:EG-Perturb3},
 we directly obtain $ \lim_{t \to \infty}\Delta_t =0$ if $\lim_{t \to \infty}\norm{v_t^2}_2=0 $.
In the following We prove $\lim_{t \to \infty}\norm{v_t^2}_2=0 $ by contradiction. Assuming that $\{\norm{v_t^2}_2\}_t  $ doesn't converge to 0, i.e.,   
\begin{align}\label{extraassume}
\exists \delta>0, \ \exists t_1,t_2,\cdots, \ s.t. \ \norm{v_{t_i}^2}_2>\delta
\end{align}
where  $t_i$ tends to $+\infty$ as $i \to \infty$. 	

Let $ \epsilon=\frac{1}{8}\delta (1-\lambda^2) $, then we can find such $ t_j $ that for any $ t>t_j $, $ \norm{X_t}_2^2-\norm{X_{t+1}}_2^2\le \epsilon $ by Lemma \ref{extramono} , latter we will prove that under the assumption \eqref{extraassume}, there exists $ t>t_j $ such that $  \norm{X_t}_2^2-\norm{X_{t+1}}_2^2> \epsilon $ which contradicts to Lemma $\ref{extramono}$.
Then we can also find a $ t_k\ge t_j $
such that for any $ t>t_k $, $$\norm{\mathcal{B}_{t}}_2 \le\min\{\frac{\delta (1-\lambda^2)}{8\norm{X_0}_2},\frac{\delta (1-\lambda^2)}{8\norm{X_0}_2^2}\} $$ by $ \lim_{t \to \infty}\mathcal{B}_t=0 $, and for any $ t_s\ge t_k $, $ \norm{v_{t_s}^2}_2>\delta $. We choose such a $ t_s $ and denote it as $ t $. 

Now we give the bound for $ \norm{\mathcal{B}_tX_t}_2 $, $ \norm{w_t^1}_2 $ and $ \norm{w_t^2}_2 $,
\begin{align*}
\norm{\mathcal{B}_tX_t}_2
&\le \norm{\mathcal{B}_t}_2\cdot\norm{X_t}_2\\
&\le \min\{\frac{\delta (1-\lambda^2)}{8\norm{X_0}_2},\frac{\delta (1-\lambda^2)}{8\norm{X_0}_2^2}\} \cdot \norm{X_0}_2\\
&=\min\{\frac{\delta (1-\lambda^2)}{8},\frac{\delta (1-\lambda^2)}{8\norm{X_0}_2}\}\\
&\le \frac{\delta (1-\lambda^2)}{8},
\end{align*}
where the second inequality comes from Lemma  \ref{extramono} .
Together with $ \norm{w_t^1}_2 $ and $ \norm{w_t^2}_2 $ are perpendicular, which implies $ \norm{\mathcal{B}_tX_t}_2^2=\norm{w_t^1}_2^2+ \norm{w_t^2}_2^2$, for $ i=1,2 $, we have
\begin{align*}
\norm{w_t^i}_2
\le \norm{\mathcal{B}_tX_t}_2
\le \frac{\delta (1-\lambda^2)}{8\norm{X_0}_2}.
\end{align*}
Now we try to determine the relationship between $ X_{t+1} $ and $ X_t $,
\begin{align*}
X_{t+1}
&=(\mathcal{A}+\mathcal{B}_t)X_t\\
\\
&=\mathcal{A}X_t+\mathcal{B}_tX_t\\
\\
&=\mathcal{A}(v_t^1+v_t^2)+w_t^1+w_t^2\\
\\
&=(v_t^1+w_t^1)+(\mathcal{A}v_t^2+w_t^2),
\end{align*}
where $ v_t^1+w_t^1\in V_1 $ and $ \mathcal{A}v_t^2+w_t^2\in V_2 $, so $ v_t^1+w_t^1$ and $ \mathcal{A}v_t^2+w_t^2 $ are perpendicular, and 
\begin{align*}
\norm{X_{t+1}}_2^2
&=\norm{v_t^1+w_t^1}_2^2+\norm{\mathcal{A}v_t^2+w_t^2 }_2^2\\
\\
&\le \norm{v_t^1}_2^2+\norm{w_t^1}_2^2+2\norm{v_t^1}_2\norm{w_t^1}_2+
\norm{\mathcal{A}v_t^2}_2^2+\norm{w_t^2}_2^2+2\norm{\mathcal{A}v_t^2}_2\norm{w_t^2}_2\\
\\
&\le \norm{v_t^1}_2^2+\lambda^2\norm{v_t^2}_2^2+\norm{w_t^1}_2^2+\norm{w_t^2}_2^2+2\norm{v_t^1}_2\norm{w_t^1}_2+2\lambda\norm{v_t^2}_2\norm{w_t^2}_2\\
\\
&\le \norm{v_t^1}_2^2+\lambda^2\norm{v_t^2}_2^2+\norm{\mathcal{B}_tX_t}_2^2+2\norm{X_0}_2\frac{\delta (1-\lambda^2)}{8\norm{X_0}_2}+2\lambda\norm{X_0}_2\frac{\delta (1-\lambda^2)}{8\norm{X_0}_2}\\
\\
&\le \norm{v_t^1}_2^2+\lambda^2\norm{v_t^2}_2^2+\frac{\delta (1-\lambda^2)}{8}+2\frac{\delta (1-\lambda^2)}{8}+2\lambda\frac{\delta (1-\lambda^2)}{8}\\
\\
&\le \norm{v_t^1}_2^2+\lambda^2\norm{v_t^2}_2^2+\frac{5}{8}\delta (1-\lambda^2).
\end{align*}
The second inequality comes from Lemma $ \ref{extracompression} $, while the third inequality comes from $\norm{\mathcal{B}_tX_t}_2^2=\norm{w_t^1}_2^2+\norm{w_t^2}_2^2$ and the upper bound for $\norm{w_t^1}_2^2$ and $\norm{w_t^2}_2^2$. 
Then we conclude that
\begin{align*}
\norm{X_t}_2^2-\norm{X_{t+1}}_2^2
&\ge (\norm{v_t^1}_2^2+\norm{v_t^2}_2^2)-(\norm{v_t^1}_2^2+\lambda^2\norm{v_t^2}_2^2+\frac{5}{8}\delta (1-\lambda^2))\\
&\ge \delta (1-\lambda^2)-\frac{5}{8}\delta (1-\lambda^2)\\
&= \frac{3}{8}\delta (1-\lambda^2)>\frac{1}{8}\delta (1-\lambda^2)=\epsilon,
\end{align*}
where a  contradiction appears.

This completes the proof.
\end{proof}

\section{More Experiments}\label{meme}

   We provide additional experiments to demonstrate the behaviors of the optimistic gradient and negative momentum methods in convergent perturbed games that do not satisfy the BAP assumption. The numerical results reveal cases where optimistic gradient/momentum method converge and cases where they do not converge.
     
    In the same setting as the experiments on Theorem 3.4, we find that both optimistic gradient descent ascent and negative momentum method converge as shown in Figure \eqref{0}.

    However, there are other cases in which these two algorithms do not converge. In Figure \eqref{1}, we present one such example. Here the payoff matrix is chosen as $A = [ [1,0],[0,0]], B = [[0,8],[0,0]]$ and 

\begin{align}\label{p2mg}
 A_t=
\begin{cases}
A, & t \textnormal{\ \ is \ odd} \\ 
A + (1/t^{0.1}) * B , & t\textnormal{\ \ is \ even}
\end{cases}.
\end{align}

 In Figure \eqref{1}, the numerical results show when using a step size of $0.015$, optimistic gradient and negative momentum algorithms will diverge, but extra gradient will converge. Based on these numerical results, we believe that beyond the setting that satisfies the BAP assumption, there exists a more complex  dynamical behaviors of optimistic gradient and negative momentum methods, which presents an interesting question for future exploration. 

  \begin{figure}[H]
    \centering 
    \includegraphics[width=0.4\textwidth]{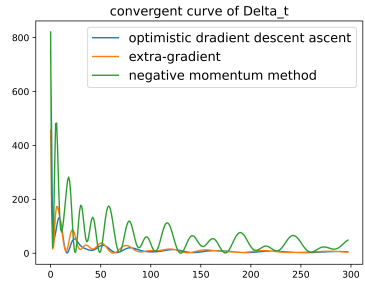}
    \caption{ Function curves of $\Delta_t$ for one game presented in  experiment of Theorem 3.4. in the paper. All these three algorithms converge. }
    \label{0}
    \end{figure}

    \begin{figure}[H]
    \centering 
    \includegraphics[width=0.4\textwidth]{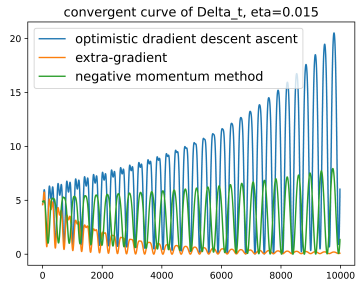}
    \caption{ Function curves of $\Delta_t$. When using step size = 0.015, extra-gradient converges, while both
optimistic gradient descent ascent and negative momentum method diverge.}
    \label{1}
    \end{figure}